\documentclass[10pt]{article}
\usepackage{ucs} 
\usepackage[utf8x]{inputenc}
\usepackage[english]{babel}
\usepackage{graphicx}
\usepackage{amsmath,amsfonts,amssymb,amsthm,amsopn,amscd}
\usepackage{bbm,wasysym,stmaryrd}
\usepackage{subfigure}
\usepackage{color}
\usepackage{setspace}
\usepackage{authblk}
\usepackage[normalem]{ulem}

\graphicspath{{Figures/}}

\newcommand{\R}{\mathbbm{R}}
\newcommand{\N}{\mathbbm{N}}
\renewcommand{\L}{\mathbbm{L}}
\newcommand{\F}{\mathcal{F}}

\newcommand{\der}[2]{ \frac{\text{d} #1}{\text{d} #2} }  
\newcommand{\Exp}{\mathbbm{E}} 

\newcommand{\M}{\mathcal{M}}
\newcommand{\C}{\mathcal{C}}

\theoremstyle{plain}
\newtheorem{theorem}{Theorem}

\newtheorem{proposition}[theorem]{Proposition}
\newtheorem{lemma}[theorem]{Lemma}

\theoremstyle{definition}

\newcommand{\john}[1]{#1}
\newcommand{\johnNew}[1]{#1}
\newcommand{\diego}[1]{#1}
\newcommand{\diegoNew}[1]{#1}
\newcommand{\olivier}[1]{#1}

\begin{document}

\title{Mean Field description of and propagation of chaos in recurrent multipopulation networks of Hodgkin-Huxley and Fitzhugh-Nagumo neurons}
\author[1,*]{Javier Baladron} \author[1,*]{Diego Fasoli} \author[1,*]{Olivier Faugeras}
\author[2,*,$\dagger$]{Jonathan Touboul}
\affil[1]{{\small NeuroMathComp Laboratory, INRIA/ENS Paris, France}}
\affil[2]{{\small BANG Laboratory, INRIA Paris, and The Mathematical Neuroscience Lab}}

\maketitle

\thispagestyle{empty}

\let\oldthefootnote\thefootnote
\renewcommand{\thefootnote}{\fnsymbol{footnote}}
\footnotetext[1]{{\tt \small email: firstname.name@inria.fr}}
\footnotetext[2]{Coll\`ege de France, Center of Interdisciplinary Research in Biology,, CNRS UMR 7241m INSERM U1050, Universit\'e Pierre et Marie Curie ED 158. MEMOLIFE Laboratory of excellence and Paris Science Lettre. 11, place Marcelin Berthelot, 75005 Paris, France.}
\let\thefootnote\oldthefootnote

\vspace{-1.5cm}

\section*{Abstract}
%

We  derive the mean-field equations arising as the limit of a network of interacting spiking neurons, as the number of neurons goes to infinity. The neurons belong to a fixed number of populations and are represented either by the Hodgkin-Huxley model or by one of its simplified version, the Fitzhugh-Nagumo model. The synapses between neurons  are either electrical or chemical. The network is assumed to be fully connected. The maximum conductances vary randomly. Under the condition that all  neurons initial conditions are drawn independently from the same law that depends only on the population they belong to, we prove that a propagation of chaos phenomenon takes places, namely that in the mean-field limit, any finite number of neurons become independent and, within each population, have the same probability distribution. This probability distribution is solution of a set of  implicit equations, either nonlinear stochastic differential equations resembling the McKean-Vlasov equations, or non-local partial differential equations resembling the McKean-Vlasov-Fokker-Planck equations. We prove the well-posedness of these equations, i.e. the existence and uniqueness of a solution. We also show the results of some preliminary numerical experiments that indicate that the mean-field equations are a good representation of the mean activity of a finite size network, even for modest sizes. These experiment also indicate that the  McKean-Vlasov-Fokker-Planck equations may be a good way to understand the mean-field dynamics through, e.g., a bifurcation analysis.

\noindent
{\bf AMS 2010 subject classification:} 60F99, 60B10, 92B20, 82C32, 82C80, 35Q80\\

\noindent
{\bf Keywords and phrases:} Mean-Field Limits, Propagation of chaos, Stochastic differential equations, McKean-Vlasov equations, Fokker-Planck equations, Neural Networks, Neural assemblies, Hodgkin-Huxley neurons, Fitzhugh-Nagumo neurons.

\section{Introduction}\label{section:intro}
Cortical activity displays highly complex behaviors which are often characterized by the presence of noise.  Reliable responses to specific stimuli often arise at the level of population assemblies (cortical areas or cortical columns) featuring a very large number of neuronal cells, each of these presenting a highly nonlinear behavior, that are interconnected in a very intricate fashion. Understanding the global behavior of large-scale neural assemblies has been a great endeavor in the past decades. One of the main interests of large scale modeling is characterizing brain function, which most imaging techniques are recording. Moreover, anatomical data recorded in the cortex reveal the existence of structures, such as the cortical columns, with a diameter of about $50 \mu m$ to $1 mm$, containing of the order of one hundred to one hundred thousand neurons belonging to a few different types. These columns have specific functions. For example, in the human visual area V1, they respond to preferential orientations of bar-shaped visual stimuli. In this case, information processing does not occur at the  scale of individual neurons but rather corresponds to an  activity integrating the individual dynamics of many interacting neurons and resulting in a mesoscopic signal arising through averaging effects, and that effectively depends on a few effective control parameters. This vision, inherited from statistical physics requires that the space scale be large enough to include sufficiently many neurons and small enough so that the region considered is homogeneous. This is in effect the case of cortical columns.

\medskip

In the field of mathematics, studying the limits of systems of particles systems in interaction has been a longstanding problem and presents many technical difficulties. One of the questions addressed in mathematics was to characterize the limit of the probability distribution of an infinite set of interacting diffusion processes, and the fluctuations around the limit for a finite number of processes. The first breakthroughs to find answers to this question are due to Henry McKean (see e.g.~\cite{mckean:66,mckean:67}). It was then investigated in various contexts by a large number of authors such as Braun-Hepp~\cite{braun-hepp:77}, Dawson~\cite{dawson:83}, Dobrushin~\cite{dobrushin:70}, and most of the theory was achieved by Tanaka and collaborators~\cite{tanaka:83,tanaka:78,tanaka:84,tanaka-hitsuda:81}, and of course Alain-Sol Sznitman~\cite{sznitman:89,sznitman:84,sznitman:84a,sznitman:86}. When considering that all particles (in our case neurons) have the same, independent initial condition, they mathematically proved using stochastic theory (Wasserstein distance, large deviation techniques) that in the limit where the number of particles tend to infinity, any finite number of particles behave independently of the other ones, and they all present the same probability distribution, which satisfies a nonlinear Markov equation. Finite-size fluctuations around the limit are derived in a general case in \cite{sznitman:84a}. Most of these models use standard hypothesis of global Lipschitz continuity and linear growth condition of the drift and diffusion coefficients of the diffusions, as well as Lipschitz continuity of the interaction function. Extensions to discontinuous c\`adl\`ag processes including singular interactions (through a local time process) were developed in \cite{sznitman:86}. Problems involving singular interaction variables (e.g. nonsmooth functions) are also widely studied in the field, but are not relevant in our case.

\medskip

In the present article, we apply this mathematical approach to the problem of interacting neurons arising in neuroscience. To this end we extend the theory to encompass a wider class of models. This implies the use of locally (instead of globally) Lipschitz coefficients and of a Lyapunov-like growth condition replacing the customary linear growth assumption for some of the functions appearing in the equations. The contributions of this article are threefold:
\begin{enumerate}
\item We derive in a rigorous manner the mean-field equations resulting from the interaction of infinitely many neurons in the case of widely accepted  models of spiking neurons and synapses.
\item We prove a propagation of chaos property which shows that in the mean-field limit neurons become independent, in agreement with some recent experimental work \cite{ecker-berens-etal:10} and with the idea that the brain processes information in a somewhat optimal way.
\item We show numerically that the mean-field limit is a good approximation of the mean activity of the network even for fairly small sizes of neuronal populations.
\item We suggest numerically that the changes in the dynamics of the mean-field limit when varying parameters can be understood by studying the mean-field Fokker-Planck equation.
\end{enumerate}
We start by reviewing such models in section \ref{sec:Motiv} to  motivate the present study. It is in section \ref{sec:Instant} that we provide  the limit equations describing the behaviours of an infinite number of interacting neurons, and state and prove the existence and uniqueness of solutions in the case of conductance-based models. The detailed proof of the second main theorem, that of the convergence of the network equations to the mean-field limit, is given in appendix \ref{section:proofs}. In section \ref{section:numerics} we begin to address the difficult problem of the numerical simulation of the mean-field equations and show some results indicating that they may be an efficient way of representing the mean activity of a finite-size network as well as to study the changes in the dynamics when varying biological parameters. The final section \ref{section:Discussion} focuses on the conclusions of our mathematical and numerical results and raises some important questions for future work.


\section{Spiking Conductance-based models}\label{sec:Motiv}
\olivier{This section sets the stage for our results. We review in section  \ref{subsection:HH} the Hodgkin-Huxley model equations in the case where both the membrane potential and the ion channels equations include noise. We then proceed in section \ref{subsection:FN} with the Fitzhugh-Nagumo equations in the case where the membrane potential equation includes noise. We next discuss in section \ref{subsection:synapses} the connectivity models of networks of such neurons, starting with the synapses, electrical and chemical, and finishing with several stochastic models of the synaptic weights. In section \ref{subsection:alltogether} we write the network equations in the various cases considered in the previous section and express them in a general \john{abstract} mathematical form that is the one used for stating and proving the results about the mean-field limits in section \ref{sec:Instant}. Before we jump into this we conclude in section \ref{subsection:mean-field} with a brief overview of the mean-field methods popular in computational neuroscience.}

From the mathematical point of view, each neuron is a complex system, whose dynamics is often described by a set of stochastic nonlinear differential equations. Such models aim at reproducing the biophysics of ion channels governing the membrane potential and therefore the spike emission. This is the case of the classical model of Hodgkin and Huxley \cite{hodgkin-huxley:52} and of its reductions~\cite{fitzhugh:66,fitzhugh:69,izhikevich:07}. Simpler models use discontinuous processes mimicking the spike emission by modeling the membrane voltage and considering that spikes are emitted when it reaches a given threshold. These are called integrate-and-fire models~\cite{lapicque:07,tuckwell:88} \john{ and will not be addressed here}. The models of large networks \john{we deal with here} therefore consist of systems of coupled nonlinear diffusion processes. 

\subsection{Hodgkin-Huxley model}\label{subsection:HH}
One of the most important models in computational neuroscience is the Hodgkin-Huxley model. Using pioneering experimental techniques of that time, Hodgkin and Huxley \cite{hodgkin-huxley:52} determined that the activity of the giant squid axon is controlled by three major currents: voltage-gated persistent $K^+$ current with four activation gates, voltage-gated transient $Na^+$ current with three activation gates and one inactivation gate, and Ohmic leak current, $I_L$, which is carried mostly by chloride ions $Cl^-$. \olivier{In this paper we only use the space-clamped Hodgkin-Huxley model which we slightly generalize to a stochastic setting in order to better take into account the variability of the parameters.}

The basic electrical relation between the membrane potential and the currents is simply:
\begin{equation*}
 C\der{V}{t} = I^{\rm ext}(t) - I_K - I_{Na} - I_L,
\end{equation*}
\olivier{$I^{\rm ext}(t)$ is an external current. The detailed expressions for $I_K$, $I_{Na}$ and $I_L$ can be found in several textbooks, e.g., \cite{izhikevich:07,ermentrout-terman:10b}:
\begin{align*}
I_K &= \bar{g}_K n^4 (V-E_K)\\
I_{Na} &= \bar{g}_{Na} m^3 h (V-E_{Na})\\
I_L &= g_L(V-E_L)
\end{align*}
where $\bar{g}_K$ (respectively $\bar{g}_{Na}$) is the maximum conductance of the potassium (respectively the sodium) channel. $g_L$ is the conductance of the Ohmic channel. $n$ (respectively $m$) is the activation variable for $K^+$ (respectively for $Na$). There are four (respectively 3) activation gates for the $K^+$ (respectively the $Na$) current which accounts for the power 4 (respectively 3) in the expression of $I_K$ (respectively $I_{Na}$). $h$ is the inactivation variable for $Na$. These activation/deactivation variables, denoted by $x \in \{n,\,m,\,h\}$ in what follows represent a proportion of open channels (they vary between 0 and 1).  
This proportion can be computed through a Markov chain modeling assuming the channels to open with rate $\rho_x(V)$ (the dependence in $V$ accounts for the voltage-gating of the channel) and to close with rate $\zeta_x(V)$. These processes can be shown to converge, under standard assumptions, towards the following ordinary differential equations:}
\[
\dot{x}=\rho_x(V)(1-x)-\zeta_x(V)\,x,\ x \in \{n,\,m,\,h\}
\]
\olivier{The functions $\rho_x(V)$ and $\zeta_x(V)$ are smooth and bounded functions whose exact values can be found in several textbooks such as the ones cited above.
A more precise model taking into account the finite number of channels results in the stochastic differential equation (see e.g.~\cite{wainrib:2010} and references therein):
\[
dx_t=(\rho_x(V)(1-x)-\zeta_x(V)\,x)\,dt + \sqrt{\rho_x(V)(1-x)+\zeta_x(V)\,x}\;\chi(x)\,dW^x_t
\]
where $W^x_t$, $x \in \{n,\,m,\,h\}$ are independent standard Brownian motions. $\chi(x)$ is a function that vanishes outside $(0,1)$. This guarantees that the solution remains a proportion i.e. lies between 0 and 1 for all times. We define
\begin{equation}\label{eq:diffusion}
\sigma_x(V,x)=\sqrt{\rho_x(V)(1-x)+\zeta_x(V)\,x}\;\chi(x).
\end{equation}
}\\

\olivier{In order to complete our stochastic Hodgkin-Huxley neuron model, we assume that the external current $I^{\rm ext}(t)$ is the sum of a deterministic part, noted $I(t)$, and a stochastic part, a white noise with variance $\sigma_{ext}$ built from a standard Brownian motion $W_t$ independent of $W^x_t$, $x \in \{n,\,m,\,h\}$.
Considering the current produced by the income of ion through these channels, we end up with the following stochastic differential equation: 
\begin{equation}\label{eq:HodgkinHuxley}
   \begin{cases}
   C d{V}_t &= \Big(I(t) - \bar{g}_K n^4 (V-E_K) - \bar{g}_{Na} m^3 h (V-E_{Na}) - \bar{g}_L (V-E_L)\Big)\,dt\\
            & \qquad \qquad +\sigma_{ext}\,dW_t \\
   d{x}_t   & = \Big(\rho_x(V) (1-x) -\zeta_x(V) x\Big)\,dt +\sigma_x(V,x)\,dW^x_t \quad x \in \{n,\,m,\,h\}
   \end{cases}
\end{equation}
}
%
%
%
The advantages of this model are numerous, and one of the most prominent aspects in its favor is its correspondence with the  most widely accepted formalism to describe the dynamics of the nerve cell membrane. A very extensive literature can also be found about the mathematical properties of this system, and it is now quite well understood. 


The stochastic version of the \olivier{space-clamped} Hodgkin-Huxley model we study is, from the mathematical point of view, a four-dimensional nonlinear stochastic differential equation. It is easy to check that the drift and diffusion functions are all Lipchitz-continuous and \olivier{hence} satisfy the linear growth condition, using the fact that all variables $n$, $m$ and $h$ are in the interval $[0,1]$ (we recall that these variables are proportions of open channels). Standard theorems of stochastic calculus ensure under these two conditions existence and uniqueness of strong solutions \john{ for any sufficiently regular (square integrable)} initial condition (see e.g.~\cite{karatzas-shreve:91,mao:08}).

\olivier{To illustrate the model we show in figure~\ref{fig:HHneuron} the time evolution of the three ion channel variables $n$, $m$ and $h$ as well as that of the membrane potential $V$ for a constant input $I=20.0$.}
The system of ODEs has been solved using a Runge-Kutta 4 scheme with integration time step $\Delta t=0.01$.

In figure \ref{fig:HHneuronnoisy} we show the same time evolution when noise is added to the channel variables and the membrane potential.\\
For the membrane potential, we have used $\sigma_{ext}=3.0$ (see equation~\eqref{eq:HodgkinHuxley}), while for the noise in the ion channels we have used the following $\chi$ function (see equation \eqref{eq:diffusion}):

\begin{equation}\label{eq:ChiFunction}
\chi(x) = \left\{
  \begin{array}{l l}
    \Gamma e^{-\frac{\textstyle \Lambda}{\textstyle 1-(2x-1)^{2}}} & \quad \text{if $0<x<1$}\\
    0 & \quad \text{if $x \leq 0 \vee x \geq 1$}\\
  \end{array} \right.
\end{equation}

with $\Gamma=0.1$ and $\Lambda=0.5$ for all the ion channels.
The system of SDEs has been integrated using the Euler-Maruyama scheme with $\Delta t=0.01$.
\begin{figure}[htbp]
\centerline{
\includegraphics[width=0.5\textwidth]{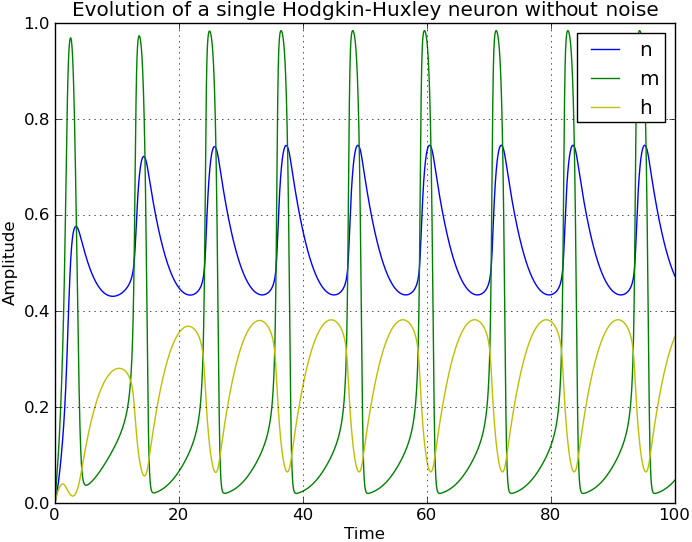}\includegraphics[width=0.5\textwidth]{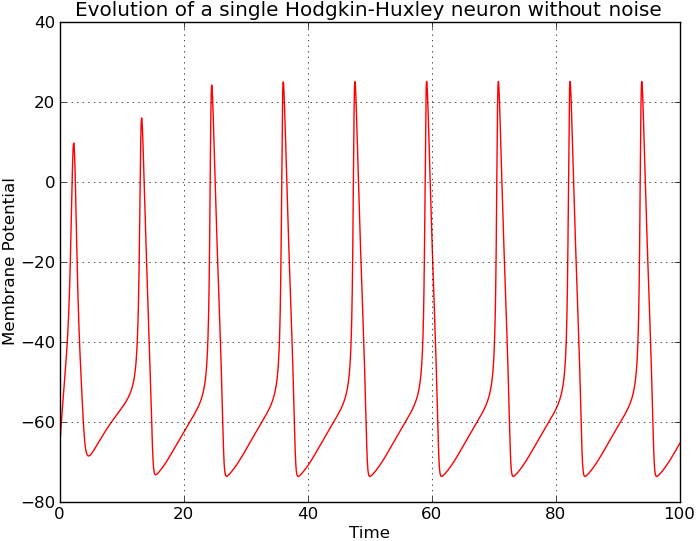}
}
\caption{\john{Solution} of the noiseless Hodgkin-Huxley model. Left: Time evolution of the three ion channel variables $n$, $m$ and $h$. Right: Corresponding time evolution of the membrane potential. \john{Parameters given in the text}.}
\label{fig:HHneuron}
\end{figure}

\begin{figure}[htbp]
\centerline{
\includegraphics[width=0.5\textwidth]{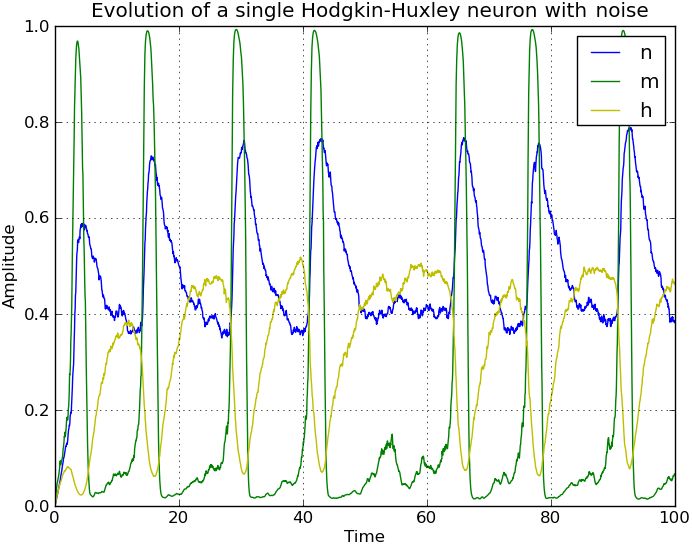}\includegraphics[width=0.5\textwidth]{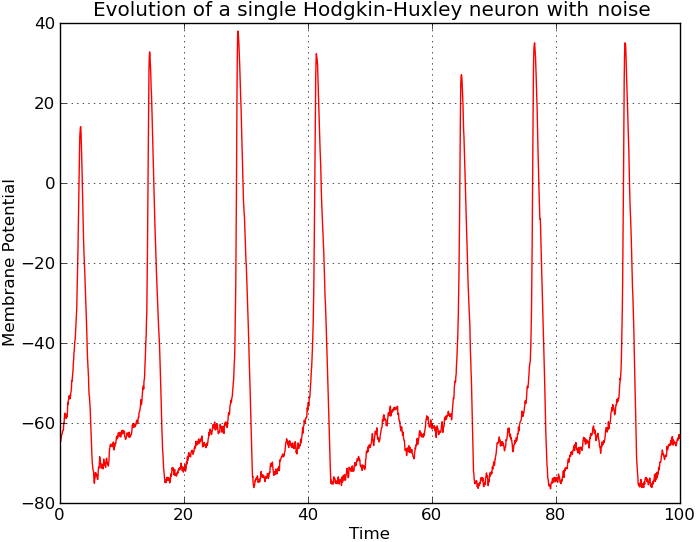}
}
\caption{\john{Noisy Hodgkin-Huxley model. Left: Time evolution of the three ion channel variables $n$, $m$ and $h$. Right: Corresponding time evolution of the membrane potential. Parameters given in the text.}}
\label{fig:HHneuronnoisy}
\end{figure}

\medskip 
 \olivier{Because the Hodgkin-Huxley model is rather complicated and high-dimensional, many reductions have been proposed, in particular to two dimensions instead of four.}
These reduced models include the famous Fitzhugh-Nagumo and Morris-Lecar models.  These two models are two-dimensional approximations of the original Hodgkin-Huxley model based on quantitative observations of the time scale of the dynamics of each variable and identification of variables. Most reduced models still comply with the Lipschitz and linear growth conditions ensuring existence and uniqueness \olivier{of a solution}, except for the Fitzhugh-Nagumo model \olivier{which} we now introduce. 

\subsection{The Fitzhugh-Nagumo model}\label{subsection:FN}
In order to reduce the dimension of the Hodgkin-Huxley model, Richard Fitzhugh \cite{fitzhugh:55,fitzhugh:66,fitzhugh:69} introduced a simplified two-dimensional model. \olivier{The motivation was} to isolate conceptually essential mathematical features yielding excitation and transmission properties from the analysis of the biophysics of sodium and potassium flows. Nagumo and collaborators~\cite{nagumo-arimoto-etal:62} followed up with an electrical system reproducing the dynamics of this model and studied its properties. The model consists of two equations, one governing a voltage-like variable $V$ having a cubic nonlinearity and a slower recovery variable $w$. It can be written as:
\begin{equation}\label{eq:FNOriginal}
	\begin{cases}
		\dot{V}=f(V)-w+I^{\rm ext}\\
		\dot{w}=c(V+a-bw)
	\end{cases}
\end{equation}
\olivier{where $f(V)$ is a cubic polynomial in $V$ which we choose, w.l.o.g., to be  $f(V)=V-V^3/3$.} The parameter $I^{\rm ext}$ models the input current the neuron receives, the parameters $a$, $b>0$ and $c>0$ describe the kinetics of the recovery variable $w$. \olivier{As in the case of the Hodgkin-Huxley model, the current $I^{\rm ext}$ is assumed to be the sum of a deterministic part, noted $I$,}  and a stochastic white noise accounting for the randomness of the environment. The stochastic Fitzhugh-Nagumo equation is deduced from \eqref{eq:FNOriginal} and reads:
\begin{equation}\label{eq:FNStoch}
	\begin{cases}
		dV_t=(V_t-\frac{V_t^3}{3} - w_t+I)\,dt + \sigma_{\rm ext}\,dW_t\\
		dw_t=c\,(V_t+a-bw_t)\,dt
	\end{cases}
\end{equation}
\olivier{Because the function $f(V)$ is not {\emph globally} Lipschitz continuous (only locally), we prove the following proposition stating that the stochastic differential equation~\eqref{eq:FNStoch} is well-posed.}

\begin{proposition}\label{prop:FN}
 Let $T > 0$ be a fixed time. If $\vert I(t)\vert \leq I_m$ on $[0,T]$, equation \eqref{eq:FNStoch} together with an initial condition $(V_0,w_0)$ with a given distribution has a unique strong solution which belongs to ${\rm L}^{\olivier{2}}([0,T];\R^2)$.
\end{proposition}
\begin{proof}
The function $h=(h_1,h_2): (V,w) \to (f(V)-w+I_d,\,c(V+a-bw))$ is locally Lipschitz in $\R^2$. According to \cite[Chapter \olivier{2}, Theorem 3.5]{mao:08} it is sufficient to show that, if we note $X$ the vector $(V,w)$ of $\R^2$, there exists a positive constant $K$ such that for all $(X,t)$ in $\R^2 \times [0,T]$
\[X^T h(X,t)+\olivier{\frac{1}{2}}\sigma_{\rm ext}^2 \leq K(1+|X|^2)\]
\olivier{Let $L=(c-1)/2$, we have:
\begin{multline*}
 X^T h(X,t) = V^2-\frac{V^4}{3}+(c-1)V\,w-b\,c\,w^2+V\,I+a\,c\,w  \leq V^2+2L Vw+2V\frac{I}{2}+2 \frac{ac}{2}w\\
\leq \left(V+Lw\right)^2+\left(V+\frac{I}{2}\right)^2+(w+ac)^2  \leq 2\left(2V^2+(L^2+1)w^2+\frac{I^2+a^2c^2}{4}\right) \\
 \leq 2\max\left\{2,L^2+1,\frac{I^2+a^2c^2}{4}\right\}(1+V^2+w^2)
\end{multline*}
}
\olivier{The conclusion follows since $I$ is bounded.}
\end{proof}

\noindent
\olivier{We show in figure~\ref{fig:FN} the time evolution of the adaptation variable and the membrane potential in the case where the input $I$ is constant and equal to $0.7$. The lefthand side of the figure shows the case with no noise while the righthand side shows the case where noise of intensity $\sigma_{ext} = 0.25$ (see equation \eqref{eq:FNStoch}) has been added.}\\
The deterministic model has been solved with a Runge-Kutta 4 method, while the stochastic model with the Euler-Maruyama scheme.
In both cases, we have used an integration time step $\Delta t=0.01$.

\begin{figure}
\centerline{
\includegraphics[width=0.5\textwidth]{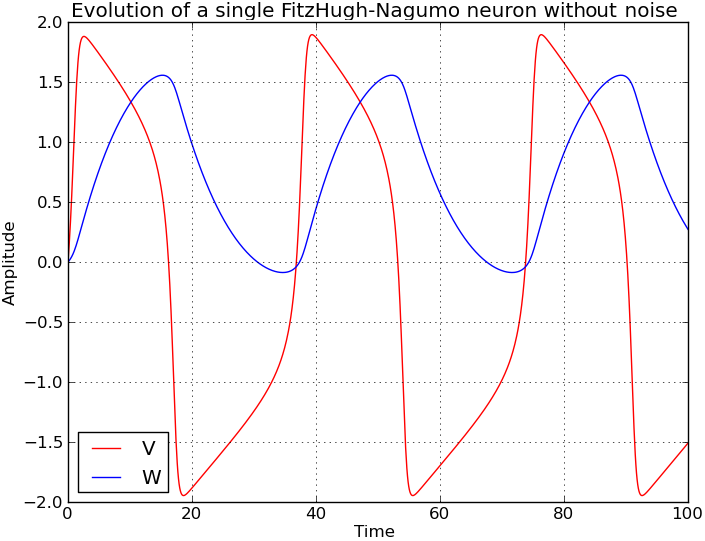}
\includegraphics[width=0.5\textwidth]{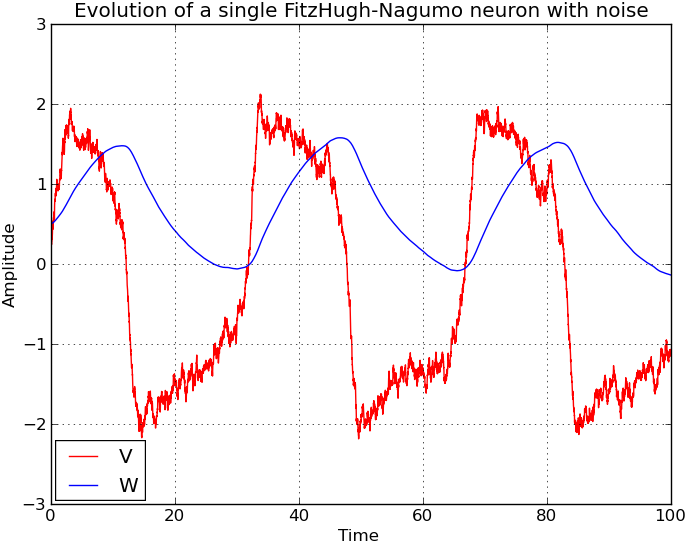}
}
\caption{Time evolution of the membrane potential and the adaptation variable in the Fitzhugh-Nagumo model. Left: without noise. Right: with noise. See text.}
\label{fig:FN}
\end{figure}

\olivier{
\subsection{Partial conclusion}
\john{We have reviewed two main models of space-clamped single neurons, the Hodgkin-Huxley and the Fitzhugh-Nagumo models. These models are stochastic, including various sources of noise, external, and internal. The noise sources are supposed to be independent Brownian processes. We have shown that the resulting stochastic differential equations~\eqref{eq:HodgkinHuxley} and~\eqref{eq:FNStoch} were well-posed. As pointed out above, this analysis extends to a large number of reduced versions of the Hodgkin-Huxley such as those that can be found in the book \cite{izhikevich:07}.}
}

\subsection{Models of synapses and maximum conductances}\label{subsection:synapses}
\olivier{We now study the situation in which several of these neurons are connected to one another forming a network, which we will assume to be fully connected. Let $N$ be the total number of neurons. These neurons belong to  $P$ populations, e.g. pyramidal cells, interneurons\ldots. If the index of a neuron is $i$, $1 \leq i \leq N$, we note $p(i)=\alpha$, $1 \leq \alpha \leq P$ the population it belongs to. We note $N_{p(i)}$ the number of neurons in population $p(i)$. Since we want to be as close to biology as possible while keeping the possibility of a mathematical analysis of the resulting model, we consider two types of simplified, but realistic, synapses, chemical and electrical or gap junctions. The following material concerning synapses is standard and can be found in textbooks \cite{ermentrout-terman:10b}. The new, and we think important, twist is add noise to our models. To unify notations, in what follows $i$ is the index of a postsynaptic neuron belonging to population $\alpha=p(i)$, and $j$ is the index of a presynaptic neuron to neuron $i$ belonging to population $\gamma=p(j)$.}
 
\subsubsection{Chemical synapses}
The principle of functioning of chemical synapses is based on the release of neurotransmitter in the presynaptic neuron synaptic button, which binds to specific receptors on the postsynaptic cell. This process is, similarly to the currents described in the Hodgkin and Huxley model, governed by the value of the cell membrane potential. We use the model described in~\cite{destexhe-mainen-etal:94b,ermentrout-terman:10b}, which \olivier{features} a quite realistic biophysical representation of the processes \olivier{at work} in the spike transmission, \olivier{and} is consistent with the \olivier{previous} formalism used to describe the conductances of other ion channels. The model emulates the fact that following the arrival of an action potential at the presynaptic terminal, neurotransmitter is released in the synaptic cleft and binds to postsynaptic receptor with a first order kinetic scheme. Let $j$ be a presynaptic neuron to the postynaptic neuron $i$. The synaptic current induced by the synapse from $j$ to $i$ can be modeled as the product of a conductance $g_{ij}$ with a voltage difference:
\begin{equation}\label{eq:SynCur}
I^{\rm syn}_{ij}=g_{ij}(t)(V^i-V_{\rm rev}^i)
\end{equation}
\olivier{The synaptic reversal potential $V_{\rm rev}^i$ is approximately constant within each population: $V_{\rm rev}^i:=V_{\rm rev}^{\alpha}$.}
The conductance $g_{ij}$ is the product of the maximum conductance $J_{ij}(t)$ with a function $y^j(t)$ that denotes the fraction of open channels and depends only upon the presynaptic neuron $j$:
\olivier{
\begin{equation}\label{eq:SynCond}
g_{ij}(t)=J_{ij}(t) y^j(t)
\end{equation}
}
 The function $y^j(t)$  is often modelled \cite{ermentrout-terman:10b} as satisfying the following ordinary differential equation
\[
\dot{y}^j(t)=\olivier{a_r^{j}} S_{\olivier{j}}(V^j) (1-y^j(t))-\olivier{a_d^j} y^j(t),
\]
\olivier{The positive constants $a_r^j$ and $a_d^j$ characterize the rise and decay rates, respectively, of the synaptic conductance. Their values depend only on the population of the presynaptic neuron $j$, i.e. $a_r^{j}:=a_r^\gamma$ and $a_d^{j}:=a_d^\gamma$, but may vary significantly from one population to the next. For example, ${\rm GABA}_{\rm B}$ synapses are slow to activate and slow to turn off while the reverse is true for ${\rm GABA}_{\rm A}$ and AMPA synapses \cite{ermentrout-terman:10b}}. $S_{\olivier{j}}(V^j)$ denotes the concentration of transmitter released into the synaptic cleft by a presynaptic spike. \olivier{We assume that the function $S_j$ is sigmoidal and that its exact form depends only upon the population of the neuron $j$}.  \olivier{Its expression is given by (see e.g. \cite{ermentrout-terman:10b}):
\begin{equation}\label{eq:S}
S_{\gamma}(V^j)=\frac{T_{\rm max}^\gamma}{1+e^{-\lambda_\gamma(V^j-V_T^\gamma)}}
\end{equation}
}
Destexhe et al. \cite{destexhe-mainen-etal:94b} \olivier{give some typical values of the parameters,} $T_{\rm max} =1 mM$, $V_T = 2 mV$, and $1/\lambda = 5 mV$.

Because of the dynamics of ion channels \olivier{and of their finite number}, similarly to the channel noise models derived through the Langevin approximation in the Hodgkin-Huxley model \eqref{eq:HodgkinHuxley}, we assume that the proportion of active channels \olivier{is actually governed by} a stochastic differential equation with diffusion coefficient \olivier{$\sigma_\gamma(V,y)$ depending only on the population $\gamma$ of $j$} of the form~\eqref{eq:diffusion}:
\[
dy^j_t=\left(\olivier{a_r^\gamma} S_{\gamma}(V^j) (1-y^j(t))-\olivier{a_d^\gamma} y^j(t)\right)\,dt + \sigma_{\olivier{\gamma}}^{y}(V^{\olivier{j}},y^{\olivier{j}})\,dW^{j,y}_t
\]
\olivier{In detail we have
\begin{equation}\label{eq:siggam}
\sigma_\gamma^{y}(V^j,y^j)=\sqrt{a_r^\gamma S_\gamma(V^j)(1-y^j)+a_d^\gamma y^j}\chi(y^j)
\end{equation}
}
\olivier{Remember that the form of the diffusion term guarantees that the solutions to this equation with appropriate initial conditions stay between 0 and 1. The Brownian motions $W^{j,y}$ are assumed to be independent from one neuron to the next.}

\subsubsection{Electrical synapses}
The electrical synapse transmission is rapid and stereotyped, and is mainly used to send simple depolarizing signals for systems requiring the fastest possible response. At the location of an electrical synapse, the separation between two neurons is very small ($\approx 3.5nm$). This narrow gap is bridged by the \emph{gap junction channels}, specialized protein structures that conduct the flow of ionic current from the presynaptic to the postsynaptic cell (see e.g.~\cite{kandel-schwartz-etal:00}). 

Electrical synapses thus work by allowing ionic current to flow passively through the gap junction pores from one neuron to another. The usual source of this current is the potential difference generated locally by the action potential. Without the need for receptors to recognize chemical messengers, signaling at electrical synapses is more rapid than that which occurs across chemical synapses, the predominant kind of junctions between neurons. The relative speed of electrical synapses also allows for many neurons to fire synchronously.

\olivier{
We model the current for this type of synapse as
\begin{equation}\label{eq:chemical}
I_{ij}^{\rm che}=J_{ij}(t)(V_i-V_j),
\end{equation}
where $J_{ij}(t)$ is the maximum conductance.
}

\subsubsection{The maximum conductances}
\olivier{
As shown in equations~\eqref{eq:SynCur}, \eqref{eq:SynCond} and \eqref{eq:chemical}, we model the current going through the synapse connecting neuron $j$ to neuron $i$ as being proportional to the maximum conductance $J_{ij}$. Because the synaptic transmission through a synapse is affected by the nature of the environment, the maximum conductances are affected by dynamical random variations (we do not take into account such phenomena as plasticity). What kind of models can we consider for these random variations?}

\olivier{The simplest idea is to assume that the maximum conductances are independent diffusion processes with mean $\frac{\bar{J}_{\olivier{\alpha\gamma}}}{N_{\gamma}}$ and standard deviation $\frac{\sigma_{\alpha\gamma}^{J}}{N_{\gamma}}$, i.e. that depend only  on the populations. The quantities $\bar{J}_{\alpha\gamma}$, being conductances, are positive. We write the following equation
\begin{equation}\label{eq:weightssimple}
J_{i \gamma}(t)=\frac{\bar{J}_{\alpha\gamma}}{N_{\gamma}}+\frac{\sigma_{\alpha\gamma}^{J}}{N_{\gamma}} \xi^{i,\,\gamma}(t),
\end{equation}
where the $\xi^{i,\,\gamma}(t)$, $i=1,\cdots,N$, $\gamma=1,\cdots,P$, are $NP$ independent zero mean unit variance white noise processes derived from $NP$ independent standard Brownian motions $B^{i,\,\gamma}(t)$, i.e. $\xi^{i,\,\gamma}(t)=\frac{dB^{i,\,\gamma}(t)}{dt}$, which we also assume to be independent of all the previously defined Brownian motions. This dynamics main advantage is its simplicity. Its main disadvantage is that if we increase the noise level $\sigma_{\alpha\gamma}$, the probability that $J_{ij}(t)$ becomes negative increases also: this would result in a negative conductance!}

\olivier{One way to alleviate this problem is  to modify the dynamics~\eqref{eq:weightssimple} to a slightly more complicated one \john{whose solutions do not change sign,} such as \john{for instance} the Cox-Ingersoll-Ross model \cite{cox-ingersoll-jr-etal:85} \john{given by:}
\begin{equation}\label{eq:weightscox}
dJ_{ij}(t)=\theta_{\alpha \gamma}(\frac{\bar{J}_{\alpha \gamma}}{N_{\gamma}}-J_{ij}(t))dt+\frac{\sigma_{\alpha \gamma}^{J}}{N_{\gamma}}\sqrt{J_{ij}(t)}dB^{i,\, \gamma}(t).
\end{equation}
Note that the righthand side only depends upon the population $\gamma=p(j)$. Let $J_{ij}(0)$ be the initial condition, it is known, \cite{cox-ingersoll-jr-etal:85}, that
\begin{align*}
\Exp\left[J_{ij}(t) \right]&=J_{ij}(0)e^{-\theta_{\alpha\gamma}t}+\frac{\bar{J}_{\alpha \gamma}}{N_{ \gamma}}(1-e^{-\theta_{\alpha\gamma}t}) \\
{\rm Var}\left(J_{ij}(t)\right)&=J_{ij}(0)\frac{(\sigma_{\alpha \gamma}^{J})^2}{N_{ \gamma}^2\theta_{\alpha\gamma}}(e^{-\theta_{\alpha\gamma}t}-e^{-2\theta_{\alpha\gamma}t})+\frac{\bar{J}_{\alpha \gamma}(\sigma_{\alpha \gamma}^{J})^2}{2N_{ \gamma}^3\theta_{\alpha\gamma}}(1-e^{-\theta_{\alpha\gamma}t})^2
\end{align*}

This shows that if the initial condition $J_{ij}(0)$ is equal to the mean $\frac{\bar{J}_{\alpha \gamma}}{N_{ \gamma}}$, the mean of the process is constant over time and equal to $\frac{\bar{J}_{\alpha \gamma}}{N_{ \gamma}}$. Otherwise, if the initial condition $J_{ij}(0)$ is of the same sign as $\bar{J}_{\alpha \gamma}$, i.e. positive, then the long term mean is $\frac{\bar{J}_{\alpha \gamma}}{N_{ \gamma}}$ and the process is guaranteed not to touch 0 if the condition $2 N_\gamma \theta_{\alpha\gamma}\bar{J}_{\alpha \gamma} \geq (\sigma_{\alpha\gamma}^{J})^2$ holds \cite{cox-ingersoll-jr-etal:85}. Note that the long term variance is $\frac{\bar{J}_{\alpha \gamma}(\sigma_{\alpha \gamma}^{J})^2}{2N_{ \gamma}^3\theta_{\alpha\gamma}}$.}



\olivier{
\subsection{Putting everything together}\label{subsection:alltogether}
We are ready to write the equations of a network of Hodgkin-Huxley or Fitzhugh-Nagumo neurons, study their properties and their limit, if any, when the number of neurons becomes large. The external current for neuron $i$ has been modelled as the sum of a deterministic part and a stochastic part:
\[
I^{\rm ext}_i(t)=I_i(t)+\sigma_{\rm ext}^i \frac{dW^i_t}{dt}
\]
We will assume that the deterministic part is the same for all neurons in the same population, $I_i:=I_\alpha$, and that the same is true for the variance, $\sigma_{\rm ext}^i:=\sigma_{\rm ext}^\alpha$. We further assume that the $N$ Brownian motions $W^i_t$ are $N$ independent Brownian motions and independent of all the other Brownian motions defined in the model. In other words
\begin{equation}\label{eq:ext}
I^{\rm ext}_i(t)=I_\alpha(t)+\sigma_{\rm ext}^\alpha \frac{dW^i_t}{dt}\quad \alpha=p(i),\ i=1,\cdots,N
\end{equation}
We only cover the case of chemical synapses, and leave it to the reader to derive the equations in the \john{simpler} case of gap junctions.
}
\olivier{
\subsubsection{Network of Fitzhugh-Nagumo neurons}
We assume that the parameters $a_i$, $b_i$ and $c_i$ in the equation~\eqref{eq:FNStoch} of the adaptation variable $w^i$ of neuron $i$ are only functions of the population $\alpha=p(i)$.}
\olivier{
\paragraph{Simple maximum conductances variation}\ If we assume that the maximum conductances fluctuate according to equation~\eqref{eq:weightssimple},
the state of the $i$th neuron in a fully connected network of Fitzhugh-Nagumo neurons with chemical synapses is determined by the variables $(V^i,w^i,y^i)$ that satisfy the following set of $3N$ stochatic differential equations:}
\olivier{
\begin{equation}\label{eq:FNNetwork}
 \begin{cases}
 	dV^i_t & = \left(V^i_t-\frac{(V^i_t)^3}{3}-w^i_t+I^{\alpha}(t)\right)dt+ \\
               & \qquad \left(\sum_{\gamma=1}^P \frac{1}{N_\gamma}  \sum_{j,\,p(j)=\gamma}\bar{J}_{\alpha\gamma}(V^i_t-V_{\rm rev}^{\alpha})y^{j}_t\right)dt+\\
			& \qquad   \sum_{\gamma=1}^P \frac{1}{N_\gamma}\left(\sum_{j,\,p(j)=\gamma} \sigma_{\alpha\gamma}^{J} (V^i_t-V_{\rm rev}^\alpha)y^{j}_t\right) dB^{i,\,\gamma}_t+\\
                        & \qquad \qquad \sigma_{\rm ext}^\alpha\,dW_t^i \\
	d{w}^i_t & = c_{\alpha}\left(V^i_t+a_{\alpha}-b_{\alpha}w^i_t\right)dt\\
	dy^i_t&= \left(a_r^\alpha S_{\alpha}(V^i_t) (1-y^i_t)-a_d^\alpha y^i_t\right)dt+\sigma_{\alpha}^{y}(V^i_t,y^i_t) dW^{i,\,y}_t
 \end{cases}
\end{equation}
}
$S_{\olivier{\alpha}}(V^i_t)$ is given by equation~\eqref{eq:S}, $\sigma_\alpha^{y}$ by equation \eqref{eq:siggam}, and  $W^{i,\,y}_t$, $i=1,\cdots,N$ are $N$ independent Brownian processes that model noise in the process of transmitter release into the synaptic clefts.
\olivier{
\paragraph{Sign preserving maximum conductances variation}\ If we assume that the maximum conductances fluctuate according to equation~\eqref{eq:weightscox} the situation is slightly more complicated. In effect, the state space of the neuron $i$ has to be augmented by the $P$ maximum conductances $J_{i\gamma}$, $\gamma=1,\cdots,P$. We obtain
\begin{equation}\label{eq:FNNetwork1}
 \begin{cases}
 	dV^i_t & = \left(V^i_t-\frac{(V^i_t)^3}{3}-w^i_t+I^{\alpha}(t)\right)dt+ \\
               & \qquad \left(\sum_{\gamma=1}^P \frac{1}{N_\gamma}  \sum_{j,\,p(j)=\gamma} J_{ij}(t)(V^i_t-V_{\rm rev}^{\alpha})y^{j}_t\right)dt+\\
                        & \qquad \qquad \sigma_{\rm ext}^\alpha\,dW_t^i \\
	d{w}^i_t & = c_{\alpha}\left(V^i_t+a_{\alpha}-b_{\alpha}w^i_t\right)dt\\
	dy^i_t&= \left(a_r^\alpha S_{\alpha}(V^i_t) (1-y^i_t)-a_d^\alpha y^i_t\right)dt+\sigma_{\alpha}^{y}(V^i_t,y^i_t) dW^{i,\,y}_t\\
        dJ_{i\gamma}(t) &= \theta_{\alpha \gamma}(\frac{\bar{J}_{\alpha \gamma}}{N_{ \gamma}}-J_{i\gamma}(t))dt+\frac{\sigma_{\alpha \gamma}^{J}}{N_{\gamma}}\sqrt{\left|J_{i\gamma}(t)\right|}dB^{i,\, \gamma}(t)\quad \gamma=1,\cdots,P
 \end{cases}
\end{equation}
which is a set of $N(P+3)$ stochastic differential equations.
}
\olivier{
\subsubsection{Network of Hodgkin-Huxley neurons}
We provide a  similar description in the case of Hodgkin-Huxley neurons. We assume that the functions $\rho_x^i$ and $\zeta_x^i$, $x\in\{n.\,m,\,h\}$ that appear in equation~\eqref{eq:HodgkinHuxley} only depend upon $\alpha=p(i)$.
}
\olivier{
\paragraph{Simple maximum conductances variation}\ If we assume that the maximum conductances fluctuate according to equation~\eqref{eq:weightssimple},
the state of the $i$th neuron in a fully connected network of Hodgkin-Huxley neurons with chemical synapses is therefore determined by the variables $(V^i,n^i,m^i,h^i,y^i)$ that satisfy the following set of $5N$ stochatic differential equations:
\begin{equation}\label{eq:HHNetwork}
 \begin{cases}
   C dV^i_t &= \left(I^\alpha(t) - \bar{g_K} n_i^4 (V^i_t-E_K) - \bar{g_{Na}} m_i^3 h_i (V^i_t-E_{Na}) - \bar{g_L} (V^i_t-E_L)+\right)dt\\
              & \qquad \left(\sum_{\gamma=1}^P \frac{1}{N_\gamma}  \sum_{j,\,p(j)=\gamma}\bar{J}_{\alpha\gamma}(V^i_t-V_{\rm rev}^{\alpha})y^{j}_t\right)dt+\\

			& \qquad   \sum_{\gamma=1}^P \frac{1}{N_\gamma}\left(\sum_{j,\,p(j)=\gamma} \sigma_{\alpha\gamma}^{J} (V^i_t-V_{\rm rev}^\alpha)y^{j}_t\right) dB^{i,\,\gamma}_t+\\

                        & \qquad \qquad \sigma_{\rm ext}^\alpha\,dW_t^i \\
   dx_i(t)   & = (\rho_x^\alpha(V^i) (1-x_i) -\zeta_x(V^i) x_i)\,dt + \sigma_x(V^i,x_i) dW^{x,i}_t\quad x \in \{n,\,m,\,h\}\\
   	dy^i_t&= \left(a_r^\alpha S_{\alpha}(V^i_t) (1-y^i_t)-a_d^\alpha y^i_t\right)dt+\sigma_{\alpha}^{y}(V^i_t,y^i_t) dW^{i,\,y}_t
 \end{cases}
\end{equation}
}
\olivier{
\paragraph{Sign preserving maximum conductances variation}\ If we assume that the maximum conductances fluctuate according to equation~\eqref{eq:weightscox} we use the same idea as in the Fitzhugh-Nagumo case of augmenting the state space of each individual neuron and obtain the following set of $(5+P)N$ stochastic differential equations.
\begin{equation}\label{eq:HHNetwork1}
 \begin{cases}
   C dV^i_t &= \left(I^\alpha(t) - \bar{g_K} n_i^4 (V^i_t-E_K) - \bar{g_{Na}} m_i^3 h_i (V^i_t-E_{Na}) - \bar{g_L} (V^i_t-E_L)+\right)dt\\
              & \qquad \left(\sum_{\gamma=1}^P \frac{1}{N_\gamma}  \sum_{j,\,p(j)=\gamma}J_{ij}(t)(V^i_t-V_{\rm rev}^{\alpha})y^{j}_t\right)dt+\\
                        & \qquad \qquad \sigma_{\rm ext}^\alpha\,dW_t^i \\
   dx_i(t)   & = (\rho_x^\alpha(V^i_t) (1-x_i) -\zeta_x^\alpha(V^i_t) x_i)\,dt + \sigma_x(V^i_t,x_i) dW^{x,i}_t\quad x \in \{n,\,m,\,h\}\\
   	dy^i_t&= \left(a_r^\alpha S_{\alpha}(V^i_t) (1-y^i_t)-a_d^\alpha y^i_t\right)dt+\sigma_{\alpha}^{y}(V^i_t,y^i_t) dW^{i,\,y}_t\\
        dJ_{i\gamma}(t) &= \theta_{\alpha \gamma}(\frac{\bar{J}_{\alpha \gamma}}{N_{ \gamma}}-J_{i\gamma}(t))dt+\frac{\sigma_{\alpha \gamma}^{J}}{N_{\gamma}}\sqrt{\left|J_{i\gamma}(t)\right|}dB^{i,\, \gamma}(t)\quad \gamma=1,\cdots,P
 \end{cases}
\end{equation}
}

\olivier{
\subsubsection{Partial conclusion}\label{subsubsection:conclusion2}
Equations~\eqref{eq:FNNetwork}-\eqref{eq:HHNetwork1} have a quite similar structure. They are well-posed, i.e. given any initial condition and any time $T>0$ they have a unique solution on $[0,T]$ which is square integrable. A little bit of care has to be taken when choosing these initial conditions for some of the parameters, i.e. $n$, $m$ and $h$ which take values between 0 and 1, and the maximum conductances when one wants to preserve their signs.
}

\olivier{In order to prepare the grounds for section \ref{sec:Instant} we explore a bit more the aforementioned common structure. Let us first consider the case of the simple maximum conductance variations for the Fitzhugh-Nagumo network. Looking at equation \eqref{eq:FNNetwork}, we define the three-dimensional state vector of neuron $i$ to be $X^i_t=(V^i_t,\,w^i_t,y^i_t)$. Let us now define $f_\alpha: \R \times \R^3 \to \R^3$, $\alpha=1,\cdots,P$,
by
\[
f_\alpha(t,X^i_t)=\left[
\begin{array}{c}
V^i_t-\frac{(V^i_t)^3}{3}-w^i_t+I^{\alpha}(t)\\
c_{\alpha}\left(V^i_t+a_{\alpha}-b_{\alpha}w^i_t\right)\\
a_r^\alpha S_{\alpha}(V^i_t) (1-y^i_t)-a_d^\alpha y^i_t
\end{array}
\right]
\]
Let us next define $g_\alpha: \R \times \R^3 \to \R^{3 \times 2}$ by
\[
g_\alpha(t,X^i_t)=\left[
                \begin{array}{cc}
                \sigma_{\rm ext}^\alpha & 0 \\
                0 & 0 \\
                0  & \sigma_{\alpha}^{y}(V^i_t,y^i_t)
                \end{array}
            \right]
\]
It appears that the intrinsic dynamics of the neuron $i$ is conveniently described by the equation
\[
dX^i_t=f_\alpha(t,X^i_t)\,dt+g_\alpha(t,X^i_t)\left[
                                                \begin{array}{c}
                                                dW_t^i\\
                                                dW^{i,\,y}_t
                                                \end{array}
                                                \right]
\]
We next define the functions $b_{\alpha\gamma}: \R^3 \times \R^3 \to \R^3$, for $\alpha,\,\gamma=1,\cdots,P$ by
\[
b_{\alpha\gamma}(X^i_t,X^j_t)=
\left[
\begin{array}{c}
\bar{J}_{\alpha\gamma}(V^i_t-V_{\rm rev}^{\alpha})y^{j}_t\\
0\\
0
\end{array}
\right],
\]
and the function $\beta_{\alpha\gamma}: \R^3 \times \R^3 \to \R^{3 \times 1}$ by
\[
\beta_{\alpha\gamma}(X^i_t,X^j_t)=
\left[
\begin{array}{c}
\sigma_{\alpha\gamma}^{J}(V^i_t-V_{\rm rev}^{\alpha})y^{j}_t\\
0\\
0
\end{array}
\right].
\]
It appears that the full dynamics of the neuron $i$, corresponding to equation~\eqref{eq:FNNetwork} can be described compactly by
\begin{multline*}
dX^i_t=f_\alpha(t,X^i_t)\,dt+g_\alpha(t,X^i_t)\left[
                                                \begin{array}{c}
                                                dW_t^i\\
                                                dW^{i,\,y}_t
                                                \end{array}
                                                \right]
+\sum_{\gamma=1}^P \frac{1}{N_\gamma}  \sum_{j,\,p(j)=\gamma} b_{\alpha\gamma}(X^i_t,X^j_t)\,dt+\\
\sum_{\gamma=1}^P \frac{1}{N_\gamma}  \sum_{j,\,p(j)=\gamma}\beta_{\alpha\gamma}(X^i_t,X^j_t)dB^{i,\,\gamma}_t
\end{multline*}
}
\olivier{
Let us now move to the case of the sign preserving variation of the maximum conductances, still for Fitzhugh-Nagumo neurons. The state of each neuron is now $P+3$-dimensional: we define $X^i_t=(V^i_t,\,w^i_t,\,y^i_t,J_{i1}(t),\cdots,J_{iP}(t))$.
We then define the functions $f_\alpha: \R \times \R^{P+3} \to \R^{P+3}$, $\alpha=1,\cdots,P$,
by
\[
f_\alpha(t,X^i_t)=\left[
\begin{array}{c}
V^i_t-\frac{(V^i_t)^3}{3}-w^i_t+I^{\alpha}(t)\\
c_{\alpha}\left(V^i_t+a_{\alpha}-b_{\alpha}w^i_t\right)\\
a_r^\alpha S_{\alpha}(V^i_t) (1-y^i_t)-a_d^\alpha y^i_t\\
\theta_{\alpha \gamma}(\frac{\bar{J}_{\alpha \gamma}}{N_{ \gamma}}-J_{i\gamma}(t))\quad \gamma=1,\cdots,P
\end{array}
\right],
\]
}
\olivier{
and the functions $g_\alpha: \R \times \R^{P+3} \to \R^{(P+3) \times (P+2)}$ by
\[
g_\alpha(t,X^i_t)=\left[
                \begin{array}{ccccc}
                \sigma_{\rm ext}^\alpha & 0                             & 0                              &   \cdots                                                      & 0\\
                0                       & 0                             & 0                              &\cdots                                                        & 0\\
                0                       & \sigma_{\alpha}^{y}(V^i_t,y^i_t) & 0                               &\cdots                                                          & 0\\
                0                       & 0                            & \frac{\sigma_{\alpha 1}^{J}}{N_{1}}\sqrt{\left|J_{i1}(t)\right|} & \cdots & 0\\
                \vdots                  & \vdots                       & \vdots                                                       & \vdots     & \vdots \\
                0                        & 0                           & 0                                                            & \cdots & \frac{\sigma_{\alpha P}^{J}}{N_{P}}\sqrt{\left|J_{iP}(t)\right|}
                \end{array}
            \right]
\]
It appears that the intrinsic dynamics of the neuron $i$ isolated from the other neurons is conveniently described by the equation
\[
dX^i_t=f_\alpha(t,X^i_t)\,dt+g_\alpha(t,X^i_t)\left[
                                                \begin{array}{c}
                                                dW_t^i\\
                                                dW^{i,\,y}_t\\
                                                dB^{i,1}_t\\
                                                \vdots\\
                                                dB^{i,P}_t
                                                \end{array}
                                                \right]
\]
}
\olivier{
Let us finally define the functions $b_{\alpha\gamma}: \R^{P+3} \times \R^{P+3} \to \R^{P+3}$, for $\alpha,\,\gamma=1,\cdots,P$ by
\[
b_{\alpha\gamma}(X^i_t,X^j_t)=
\left[
\begin{array}{c}
\bar{J}_{\alpha\gamma}(V^i_t-V_{\rm rev}^{\alpha})y^{j}_t\\
0\\
\vdots\\
0\\
\end{array}
\right].
\]
It appears that the full dynamics of the neuron $i$, corresponding to equation~\eqref{eq:FNNetwork1} can be described compactly by
\[
dX^i_t=f_\alpha(t,X^i_t)\,dt+g_\alpha(t,X^i_t)\left[
                                                \begin{array}{c}
                                                dW_t^i\\
                                                dW^{i,\,y}_t\\
                                                dB^{i,1}_t\\
                                                \vdots\\
                                                dB^{i,P}_t
                                                \end{array}
                                                \right]
+\sum_{\gamma=1}^P \frac{1}{N_\gamma}  \sum_{j,\,p(j)=\gamma} b_{\alpha\gamma}(X^i_t,X^j_t)\,dt
\]
}
\olivier{We let the reader apply the same machinery to the network of Hodgkin-Huxley neurons.}


We are interested in the  behavior of the solutions of these equations as the number of neurons tends to infinity. This problem has been longstanding in neuroscience arousing the interest of many researchers in different domains. We discuss the different approaches developed in the field in the next subsection. 

\subsection{Mean-field methods in computational neuroscience: a quick overview}\label{subsection:mean-field}
Obtaining the equations of evolution of the effective mean-field from microscopic dynamics is a very complex problem. Many approximate solutions have been provided, mostly based on the statistical physics literature.

Many models describing the emergent behavior arising from the interaction of neurons in large-scale networks have relied on continuum limits ever since the seminal work of Wilson and Cowan and Amari \cite{amari:72,amari:77,wilson-cowan:72,wilson-cowan:73}. Such models represent the activity of the network by macroscopic variables, e.g., the population-averaged firing rate, which are generally assumed to be deterministic. \olivier{When the spatial dimension is not taken into account in the equations they are referred to as neural masses otherwise as neural fields. The equations that relate these variables are ordinary differential equations for neural masses, integro-differential equations for neural fields. In the second case they fall in a category studied in \cite{hammerstein:30} or can be seen as ordinary differential equations
defined on specific functional spaces \cite{faugeras-grimbert-etal:08}.} Many analytical and numerical results have been derived from these equations and related to cortical phenomena, for instance for the problem of  spatio-temporal pattern formation in spatially extended models (see e.g.~\cite{coombes-owen:05,ermentrout:98,ermentrout-cowan:79,laing-troy-etal:02}). \olivier{The use of bifurcation theory has also proven to be quite powerful \cite{veltz-faugeras:10,chossat-faugeras:09}}. \olivier{Despite all its qualities, this approach implicitly makes the assumption that the effect of noise vanishes at the mesoscopic and macroscopic scales and hence that the behavior of such populations of neurons is deterministic.} 

A different approach has been to study regimes where the activity is uncorrelated. A number of computational studies on the integrate-and-fire neuron showed that under certain conditions neurons in large assemblies end up firing asynchronously, producing null correlations \cite{abbott-van-vreeswijk:93,amit-brunel:97,brunel-hakim:99}. In these regimes, the correlations in the firing activity decrease towards zero in the limit where the number of neurons tends to infinity. The emergent global activity of the population in this limit is deterministic, and evolves according to a mean-field firing rate equation. 
\olivier{However, according to the theory, these states only exist in the limit where the number of neurons is infinite, thereby}
\olivier{raising the question of how the finiteness of the number of neurons impacts the existence and behavior of asynchronous states.}
The study of finite-size effects for asynchronous states is generally not reduced to the study of mean firing rates and can include higher order moments of firing activity \cite{mattia-del-giudice:02,elboustani-destexhe:09,brunel:00}. In order to go beyond asynchronous states and take into account the stochastic nature of the firing and \olivier{understand} how this activity scales as the network size increases, different approaches have been developed, such as the population density method and related approaches \cite{cai-tao-etal:04}. Most of these approaches involve expansions in terms of the moments of the corresponding random variables, and the moment hierarchy needs to be truncated which is not a simple task that can raise a number of difficult technical issues (see e.g.\cite{ly-tranchina:07}).

However, increasingly many researchers now believe that the different intrinsic or extrinsic noise sources are part of the neuronal signal, and rather than being a pure disturbing effect related to the intrinsically noisy biological substrate of the neural system, they suggest that noise conveys information that can be an important principle of brain function \cite{rolls-deco:10}. At the level of a single cell, various studies have shown that the firing statistics are highly stochastic with probability distributions close to Poisson distributions \cite{softky-koch:93}, and several computational studies confirmed the stochastic nature of single-cells firings \cite{brunel-latham:03,plesser:99,touboul-faugeras:07b,touboul-faugeras:08}.  How does the variability at the single neuron level affect the dynamics of cortical networks is less well-established. Theoretically, the interaction of a large number of neurons that fire stochastic spike trains can naturally produce correlations in the firing activity of the population. For instance power-laws in the scaling of avalanche-size distributions has been studied both via models and experiments~\cite{beggs-plenz:04,benayoun-cowan:10,levina-etal:09,touboul-destexhe:10}. In  these regimes the randomness plays a central role. 

In order to study the effect of the stochastic nature of the firing in large networks, many authors strived to introduce randomness in a tractable form. Some of the models proposed in the area are based on the definition of a Markov chain governing the firing dynamics of the neurons in the network, where the transition probability satisfies a differential equation\john{, \sout{called} the} \emph{master equation}. Seminal works of the application of such modeling for neuroscience date back to the early 90s and have been recently developed by several authors \cite{ohira-cowan:93,buice-cowan:07,buice-cowan:10,bressloff:09,elboustani-destexhe:09}. Most of these approaches are proved correct in some parameter regions using statistical physics tools such as path integrals and Van-Kampen expansions, and \john{their analysis often} involve a moment expansion and truncation. Using a different approach, a static mean-field study of multi population network activity was developed by Treves in \cite{treves:93}.  This author did not consider external inputs but incorporated dynamical synaptic currents and adaptation effects. His analysis was completed in \cite{abbott-van-vreeswijk:93}, where the authors proved, using a Fokker-Planck formalism, the stability of an asynchronous state in the network. Later on, Gerstner in \cite{gerstner:95} built a new approach to characterize the mean-field dynamics for the Spike Response Model, via the introduction of suitable kernels propagating the collective activity of a neural population in time.
\olivier{Another approach is based on the use of large deviation techniques to study large networks of neurons \cite{faugeras-touboul-etal:09}. This approach is inspired by the work on spin-glass dynamics, e.g., \cite{guionnet:97}. It takes into account the randomness of the maximum conductances and the noise at various levels. The individual neuron models are rate models, hence already mean-field models. The mean-field equations are not rigorously derived from the network equations in the limit of an infinite number of neurons but they are shown to have a unique, non-Markov solution, i.e. with infinite memory, for each initial condition.}

Brunel and Hakim considered a network of integrate-and-fire neurons connected with constant maximum conductances \cite{brunel-hakim:99}. In the case of sparse connectivity, stationarity, and in a regime where individual neurons emit spikes at low rate, they were able to study analytically the dynamics of the network and to show that it exhibits a sharp transition between a stationary regime and a regime of fast collective oscillations weakly synchronized. Their approach was based on a perturbative analysis of the Fokker-Planck equation. A similar formalism was used in \cite{mattia-del-giudice:02} which, when complemented with self-consistency equations, resulted in the dynamical description of the mean-field equations of the network, and was extended to a multi population network.
Finally, Chizhov and Graham \cite{chizhov-graham:07} have recently proposed
a new method based on a population density approach allowing to characterize the mesoscopic
behaviour of neuron populations in conductance-based models.

\olivier{Let us finish this very short and incomplete survey by mentioning the work of Sompolinsky and colleagues.} Assuming a linear intrinsic dynamics for the individual neurons \olivier{described by a rate model} and random centered maximum conductances for the connections, they showed~\cite{sompolinsky-crisanti-etal:88,sompolinsky-zippelius:82,crisanti-sommers-etal:90} that the system undergoes a phase transition between two different stationary regimes: a ``trivial'' regime where the system has a unique null and uncorrelated solution, and a ``chaotic'' regime in which the firing rate converges towards a non-zero value and correlations stabilize on a specific curve which they were able to characterize. 

All these approaches have in common that they are not based on the most widely accepted microscopic dynamics (such as the ones represented by the Hodgkin-Huxley equations or some of their simplifications), and/or involve approximations or moment closures. Our approach is distinct in that it aims at deriving rigorously and without approximations the mean-field equations of populations of neurons whose individual neurons are described by biological, if not correct at least plausible, representations. The price to pay is the complexity of the resulting mean-field equations. The specific study of their solutions is therefore a crucial step, that will be developed in forthcoming papers.

\section{Mean-field equations for conductance-based models}\label{sec:Instant}
In this section, we \olivier{give a general formulation of  the neural networks models} introduced in the previous section and \olivier{use it in a probabilistic framework to} address the problem of the asymptotic behavior of the networks, as the number of neurons $N$ goes to infinity. \olivier{In other words} we derive the limit in law of $N$ interacting neurons, each of which satisfying a nonlinear stochastic differential equation \olivier{of the type} described in section~\ref{sec:Motiv}. In the remaining of this section, we work in a complete probability space $(\Omega, \mathcal{F},\mathbb{P})$ satisfying the usual conditions, and endowed with a filtration $\big(\mathcal{F}_t\big)_t$.

\subsection{Setting of the problem}

We recall that the neurons in  the network  fall into $P$ different populations. The populations differ through the intrinsic properties of their neurons and the input they receive. We assume that the number of neurons in each population $\alpha \in \{1,\ldots,P\}$, denoted by $N_{\alpha}$, increases as the network size increases, and moreover that the asymptotic proportion of neurons in population $\alpha$ is non-trivial, i.e. $N_{\alpha}/N\to \lambda_{\alpha}\in (0,1)$ as $N$ goes to infinity\footnote{\johnNew{As we will see in the proof, most properties are valid as soon as $N_{\alpha}$ tends to infinity as $N$ goes to infinity for all $\alpha\in\{1,\cdots,P\}$, the previous assumption will allow quantifying the speed of convergence towards the asymptotic regime.}}.

\olivier{We use the notations introduced in section~\ref{subsubsection:conclusion2} and the reader should refer to this section to give a concrete meaning to the rather abstract (but required by the mathematics) setting that we now establish.}

Each neuron \olivier{$i$} in population $\alpha$ \olivier{is described by a state vector noted $X^{i,N}_t$ in $\R^d$} and has an intrinsic dynamics governed by a drift function $f_{\alpha}:\R\times \R^d \mapsto \R^d$ and a diffusion matrix $g_{\alpha}:\R\times \R^d \mapsto \R^{d\times m}$ assumed \olivier{uniformly} locally Lipschitz continuous with respect to the second variable. For a neuron $i$ in population $\alpha$, the dynamics of the  $d$-dimensional process $(X^i_t)$ governing the evolution of the membrane potential and additional variables (adaptation, ionic concentrations), when there is no interaction, is governed by the equation:
\[
dX^{i,N}_t= f_{\alpha}(t,X^{i,N}_t)\, dt + g_{\alpha}(t,X^{i,N}_t)\,dW_t^i.
\]
We moreover assume, \olivier{as it is the case for all the models described in section \ref{sec:Motiv}}, that the solutions of this stochastic differential equation exist for all time. 

When included in the network, these processes interact with \olivier{those of all the other} neurons through a \olivier{a set of continuous functions that only depend on the population $\alpha=p(i)$ the neuron $i$ belongs to and the populations $\gamma$ of  the presynaptic neurons}. These functions $b_{\alpha\gamma}(x,y): \R^d \times \R^d \mapsto \R^d$, are scaled  by the coefficients $1/N_{\gamma}$, so that the maximal interaction is independent of the size of the network (in particular, neither diverging nor vanishing as $N$ goes to infinity). 

As discussed in section \ref{sec:Motiv}, due to the stochastic nature of ionic currents and the noise effects linked with the discrete nature of charge carriers, the maximum conductances are perturbed dynamically through the \olivier{$N \times P$} independent Brownian motions \olivier{$B^{i,\,\alpha}_t$} of dimension $\delta$ \olivier{that were previously} introduced. The interaction between the neurons \olivier{and} the noise term is \olivier{represented by} the function $\beta_{\alpha\gamma} : \R^d \times \R^d \mapsto \R^{d\times \delta}$. 

In order to introduce the stochastic current and stochastic maximum conductances, we define two independent sequences of independent  $m-$ and $\delta$-dimensional Brownian motions noted $(W^{i}_t)_{i\in\N}$ and $(B^{i\alpha}_t)_{i\in\N, \alpha\in \{1\cdots P\}}$ which are adapted to the filtration $\mathcal{F}_t$. 

The resulting equation for the $i$th neuron, including the noisy interactions reads:
\begin{multline}\label{eq:Network}
	d\,X^{i,N}_t=f_{\alpha}(t,X^{i,N}_t) \, dt + \sum_{\gamma=1}^P \frac{1}{N_{\gamma}} \sum_{j,\, p(j)=\gamma}b_{\alpha\gamma}(X^{i,N}_t,X^{j,N}_t) \,dt + g_{\alpha}(t,X^{i,N}_t)  dW^{i}_t \\
	+ \sum_{\gamma=1}^P \frac{1}{N_{\gamma}} \sum_{j, \,p(j)=\gamma}\beta_{\alpha\gamma}(X^{i,N}_t	,X^{j,N}_t) dB^{i\gamma}_t.
\end{multline}
\olivier{Note that this implies that $X^{i,N}$ and $X^{j,N}$ have the same law whenever $p(i)=p(j)$.}\\
These equations are similar to the equations studied in another context by a number of mathematicians, among which McKean, Tanaka and Sznitman (see the Introduction), in that \olivier{they involve} a very large number of particles (here particles are neurons) in interaction. Motivated by the study of McKean-Vlasov equations, these authors studied special cases of equations~\eqref{eq:Network}. This theory, referred to as the kinetic theory, is chiefly interested in the study of the thermodynamics questions. They show the property that in the limit where the number of particles tends to infinity, provided that the \olivier{initial} state of each particle is drawn independently from the same law, each particle behaves independently and has the same law, which is \olivier{given by} an implicit stochastic equation. They also evaluate the fluctuations around this limit, under diverse conditions~\cite{tanaka-hitsuda:81,tanaka:78,tanaka:84,sznitman:84a,sznitman:86,mckean:66,mckean:67}. \olivier{Some} extensions to biological problems where the drift term is not globally Lipschitz but satisfies the monotone growth condition \eqref{eq:MonotoneCondition} were studied in~\cite{bolley-canizo-etal:10}. \johnNew{This is the approach we undertake here.}

We now spell out our assumptions on the function appearing in the  model:
\renewcommand{\theenumi}{(H\arabic{enumi})}
\begin{enumerate}
	\item {\bf Locally Lipschitz \john{dynamics}:} \label{Assump:LocLipsch} for all $\alpha\in\{1,\ldots,P\}$, the functions $f_{\alpha}$ and $g_{\alpha}$ are \olivier{uniformly} locally Lipschitz-continuous with respect to the second variable. In detail, for all $U>0$ there exists $K_U>0$ and independent of $t \in [0,T]$ such that for all $x,\,y\ \in B^d_U$, the ball of $\R^d$ of radius $U$:
\[
	\Vert f_\alpha(t,x) - f_\alpha(t,y) \Vert + \Vert g_\alpha(t,x) - g_\alpha(t,y) \Vert \leq K_U \, \Vert x-y \Vert\quad \alpha=1,\cdots,P,
\]
	\item {\bf Lipschitz \john{interactions}:}\label{Assump:LocLipschb} for all $\alpha, \gamma \in\{1,\ldots,P\}$, the functions $b_{\alpha\gamma}$ and $\beta_{\alpha\gamma}$ are Lipschitz-continuous with Lipschitz constant $L$:
for all $(x,y)$ and $(x',y')$ in $\R^d \times \R^d$ we have:
\[
\Vert b_{\alpha\gamma}(x,y)-b_{\alpha\gamma}(x',y') \vert  +\| \beta_{\alpha\gamma}(x,y)-\beta_{\alpha\gamma}(x',y') \| \leq L \, (\Vert x-x'\Vert + \Vert y-y'\Vert )
 \]
\item {\bf Linear growth \john{of the interactions}:}\label{Assump:bBound} There exists a $\tilde{K}>0$ such that:
\[\max(\Vert b_{\alpha\gamma}(x,z)\Vert^2,\Vert \beta_{\alpha\gamma}(x,z)\Vert^2) \leq \tilde{K}(1+ \Vert x \Vert^2)\]
\item {\bf Monotone growth \john{of the dynamics}:}\label{Assump:MonotoneGrowth} 
We assume that $f_\alpha$ and $g_\alpha$ satisfy the following monotonous condition for all $\alpha=1,\cdots,P$:
	\begin{equation}\label{eq:MonotoneCondition}
	 	x^T f_{\alpha}(t,x) + \frac 1 2 \Vert g_{\alpha}(t,x) \Vert^2 \leq K \; (1+\Vert x \Vert^2)
	\end{equation}

\end{enumerate}

\subsection{Convergence of the network equations to the mean-field equations and properties of those equations}
\olivier{We now show that the same type of phenomena that were predicted for systems of interacting particles happen in networks of neurons. In detail we prove
that in the limit of large populations the network displays the property of propagation of chaos. This means that any finite number of diffusion processes become independent, and all neurons belonging to a given population $\alpha$ have asymptotically the same probability distribution, which is the solution of the following mean-field equation:}
\begin{multline}\label{eq:MFE}
	d\bar{X}^{\alpha}_t= f_{\alpha}(t,\bar{X}^{\alpha}_t)\,dt + \sum_{\gamma=1}^P\Exp_{\bar{Z}}[b_{\alpha\gamma}(\bar{X}^{\alpha}_t, \bar{Z}^{\gamma}_t)] \,dt + g_{\alpha}(t,\bar{X}^{\alpha}_t)\,dW^{\alpha}_t \\
	+ \sum_{\gamma=1}^P\Exp_{\bar{Z}}[\beta_{\alpha\gamma}(\bar{X}^{\alpha}_t, \bar{Z}^{\gamma}_t)] \,dB^{\alpha\gamma}_t\quad \alpha=1,\cdots,P
\end{multline}
where $\bar{Z}$ is a process independent of $\bar{X}$ that has the same law, and $\Exp_{\bar{Z}}$ denotes the expectation under the law of $\bar{Z}$. In other words, the mean-field equation can be written, denoting by \olivier{$m_t^{\gamma}(dz)$ the law of $\bar{Z}^{\gamma}_t$ (hence also of $\bar{X}^\gamma_t$)}:
\begin{multline} \label{eq:MFE1}
	d\bar{X}^{\alpha}_t= f_{\alpha}(t,\bar{X}^{\alpha}_t)\,dt + \sum_{\gamma=1}^P\olivier{\left(\int_{\R^d} b_{\alpha\gamma}(\bar{X}^{\alpha}_t, z) dm^{\gamma}_t(z)\right)}\,dt + g_{\alpha}(t,\bar{X}^{\alpha}_t)\,dW_t^{\alpha}\\
	+ \sum_{\gamma=1}^P \olivier{\left(\int_{\R^d} \beta_{\alpha\gamma}(\bar{X}^{\alpha}_t, z) dm^{\gamma}_t(z)\right)}\,dB_t^{\alpha\gamma}
\end{multline}
 In these equations, $W_t^{\alpha}$, for $\alpha=1\cdots P$, are independent standard $m$-dimensional Brownian motions. Let us point out the fact that the righthand side of \eqref{eq:MFE} and \eqref{eq:MFE1} depends on the law of the solution, this fact is sometimes referred to as ``the process $\bar{X}$ is attracted by its own law''. This equation is also classically written as the McKean-Vlasov-Fokker-Planck equation on the probability distribution $p$ of the solution.
\olivier{This equation which we use in section~\ref{section:numerics}
can be easily derived from equation~\eqref{eq:MFE}. Let $p_\alpha(t,z)$, $z=(z_1,\cdots,z_d)$, be the probability density at time $t$ of the solution $\bar{X}_t^\alpha$ to~\eqref{eq:MFE} (this is equivalent to $dm^\alpha_t(z)=p_\alpha(t,z)dz$), we have:
\begin{multline}\label{eq:McKeanVlasofFP}
	\partial_t p_\alpha(t,z)= \\
- {\rm div}_{z} \left(\left(f_\alpha\left(t,z\right) + \sum_{\gamma=1}^P \int b_{\alpha \gamma}\left(z, y\right) p_\gamma\left(t,y\right)\,dy \right) p_\alpha\left(t,z\right)\right) \\
	+  \frac 1 2 \sum_{i,\,j=1}^d \frac{\partial^2}{\partial z_i\partial z_j} \left(D_{ij}^\alpha\left(z\right)p_\alpha\left(t,z\right)\right)\quad \alpha=1,\cdots,P,
\end{multline}
where the $d \times d$ matrix $D^\alpha$ is given by
\[
D^\alpha(z)=\sum_{\gamma=1}^P \Exp_Z[\beta_{\alpha \gamma}(z,Z)]\Exp_Z^T[\beta_{\alpha \gamma}(z,Z)].
\]
\johnNew{with}
\[
 \Exp_Z[\beta_{\alpha \gamma}(z,Z)]=\int \beta_{\alpha \gamma}(z,y) p_\gamma(t,y)\,dy.
\]
}
The $P$ equations \eqref{eq:McKeanVlasofFP} yield the probability densities of the solutions $\bar{X}^\alpha_t$ of the mean-field equations \eqref{eq:MFE}. Because of the propagation of chaos result, the $\bar{X}^\alpha_t$ are statistically independent, but their probability functions are clearly functionally dependent.

\olivier{We now spend some time on notations in order to obtain a somewhat more compact form of equation~\eqref{eq:MFE}.}
We define $\bar{X}_t$ to be the $d\,P$-dimensional process $\bar{X}_t=(\bar{X}^{\alpha}_t\;;\; \alpha=1\cdots P)$. We similarly define $f$, $g$, $b$ and $\beta$ the concatenation of functions $f_{\alpha}$, $g_{\alpha}$, $b_{\alpha,\beta}$ and $\beta_{\alpha,\gamma}$. In details $f(t,\bar{X}) = (f_{\alpha}(t,\bar{X}^{\alpha}_t)\;;\;\alpha=1\cdots P)$, $b(X,Y) = (\sum_{\gamma=1}^P b_{\alpha\gamma}(X^{\alpha},Y^{\gamma})\;;\;\alpha=1\cdots P)$), and $W=(W^{\alpha}\;;\;\alpha=1\cdots P)$. The term of noisy synaptic interactions requires a more careful treatment. We define $\beta=(\beta_{\alpha\gamma}\;;\; \alpha,\gamma =1\cdots P) \in \big(\R^{d\times \delta}\big)^{P\times P}$ and $B=(B^{\alpha\gamma}\;;\;\;\; \alpha,\gamma =1\cdots P) \in \big(\R^{\delta}\big)^{P\times P}$, and the product $\odot$ of an element $M \in \big(\R^{d\times \delta}\big)^{P\times P}$ and an element $X\in\big(\R^{\delta}\big)^{P\times P}$ the element of $\big(\R^{d})^{P}$:
\[(M\odot X)_{\alpha} = \sum_{\gamma} M_{\alpha\gamma} X^{\alpha\gamma} \]
We obtain the equivalent compact  mean-field equation:
\begin{equation}\label{eq:CompactMFE}
	d\bar{X}_t=\Big(f(t,\bar{X}_t) + \Exp_{\bar{Z}}[b(\bar{X}_t, \bar{Z}_t)] \Big)\,dt + g(t,\bar{X}_t)\,dW_t \\
	+ \Exp_{\bar{Z}}[\beta(\bar{X}_t, \bar{Z}_t)] \odot dB_t
\end{equation}

The equations~\eqref{eq:MFE} and~\eqref{eq:McKeanVlasofFP} are implicit equations on the law of $\bar{X}_t$ .


We now state the main theoretical results of the paper as two theorems. The first theorem is about the well-posedness of the mean-field equation~\eqref{eq:MFE}. The second is about the convergence of the solutions of the network equations to those of the mean-field equations. Since the proof of the second theorem involves similar ideas to those used in the proof of the first, it is given in the appendix~\ref{section:proofs}. 
\begin{theorem}\label{thm:Principal1}
	Under the assumptions \ref{Assump:LocLipsch} to \ref{Assump:MonotoneGrowth}  
there exists a unique solution to the mean-field equation~\eqref{eq:MFE} on $[0,T]$ for any $T>0$.
\end{theorem}
Let us denote by $\mathcal{M}(\mathcal{C})$ the set of probability distributions on $\mathcal{C}$ the set continuous functions $[0,T] \mapsto \big(\R^{d}\big)^{P}$, and $\M^2(\C)$ the space of square-integrable processes. Let also $(W^{\alpha}\;;\; \alpha =1\cdots P)$ (respectively $(B^{\alpha\gamma}, \ \alpha,\gamma =1\cdots P)$) be a family of $P$ (\olivier{respectively} $P^2$) independent $m$-(\olivier{respectively} $\delta$)-dimensional  adapted standard Brownian motions on $(\Omega,\F,P)$. Let us also note $X_0 \in \mathcal{M}(\R^d)^P$ the (random) initial condition of the mean-field equation. We introduce the map $\Phi$ acting on stochastic processes and defined by:
\[\Phi: 
\begin{cases}
	\mathcal{M}(\mathcal{C}) &\mapsto \mathcal{M}(\mathcal{C})\\
	X &\mapsto (Y_t=\{Y^{\alpha}_t, \alpha=1\cdots P\})_t \text{ with } \\
	Y_t^{\alpha}&=X_0^{\alpha} +  \int_{0}^t \Big (f_{\alpha}(s,X_s^{\alpha}) + \sum_{\gamma=1}^P\Exp_Z [b_{\alpha\gamma}(X_s^{\alpha},Z_s^{\gamma})] \Big) \,ds + \\
        & \quad \int_0^t g_{\alpha}(s,X_s^{\alpha})\,dW_s^{\alpha}\\
	 &\quad \quad + \sum_{\gamma=1}^P\int_0^t \Exp_Z [\beta_{\alpha\gamma}(X_s^{\alpha},Z_s^{\gamma})]\,dB_s^{\alpha\gamma}\quad \alpha=1,\cdots,P
\end{cases}
\]
\olivier{We have introduced in the previous formula the process $Z_t$ with the same law as and independent of $X_t$}. There is a trivial identification between the solutions of the mean-field equation~\eqref{eq:MFE} and the fixed points of the map $\Phi$: any fixed point of $\Phi$ provides a solution for equation~\eqref{eq:MFE} and conversely any solution of equation~\eqref{eq:MFE} is a fixed point of $\Phi$. 

The following lemma is useful to prove the theorem. 


\begin{lemma}\label{lem:SoluL2}
	Let $X_0 \in \mathbbm{L}^2(\johnNew{(\R^{d})^P})$ 
 a square-integrable random variable. Let $X$ be a solution of the mean-field equation~\eqref{eq:MFE} with initial condition $X_0$.
Under the assumption~\eqref{eq:MonotoneCondition}, there exists a constant $C(T)>0$ depending on the parameters of the system and on the horizon $T$, such that: 
\[
\Exp[ \Vert X_t \Vert^2]\leq C(T),\ \olivier{\forall t \in [0,T]}
\] 
\end{lemma} 
\begin{proof}
Using It\^o formula for $\|X_t\|^2$, we have:
\begin{multline*}
\|X_t\|^2=\|X_0\|^2 + 2\int_0^t \Big(X_s^T f(s,X_s) + \frac 1 2 \| g(s,X_s)\|^2 +X_s^T \Exp_{Z}[b(X_s,Z_s)] + \\
\olivier{\frac 1 2}\| \Exp_{Z}[\beta(X_s,Z_s)]\|^2)\Big)\,ds + \olivier{N_t}
\end{multline*}
where \olivier{$N_t$} is a stochastic integral, hence with a null expectation, \olivier{$\Exp[N_t]=0$}. 

This expression involves the term  $x^T \,\olivier{b}(x,z)$. Because of assumption~\ref{Assump:bBound}, we clearly have:
\[\vert x^T \,b(x,z)\vert \leq \| x\| \; \| \olivier{b}(x,z)\| \leq \| x\| \sqrt{\tilde{K}(1+\| x \|^2)}\leq \olivier{\sqrt{\tilde{K}}} (1+\| x \|^2) \]
\olivier{It also involves the term $x^T \,f(t,x)+\frac 1 2 \| g(t,x) \|^2$ which, because of assumption~\ref{Assump:MonotoneGrowth}, is upperbounded by $K(1+\| x \|^2)$. Finally, assumption~\ref{Assump:bBound} again allows us to upperbound the term $\frac 1 2\| \Exp_{Z}[\beta(X_s,Z_s)]\|^2$ by $\frac{\tilde{K}}{2} (1+\| X_s \|^2)$.}\\
\olivier{Finally we obtain}
\[\Exp[1+\|X_t\|^2]=\Exp[1+\|X_0\|^2] + 2\,(K+\frac{\tilde{K}}{2} + \sqrt{\tilde{K}}) \int_0^t \Exp[1+\|X_s\|^2]\,ds.\]
Using Gronwall's inequality, we deduce the $\L^2$ boundedness of the solutions of the mean-field equations. 
\end{proof}

\olivier{This lemma puts us} in a position to prove the existence and uniqueness theorem:


\begin{proof} 
	We start by showing the existence of solutions, and then prove the uniqueness property. We recall that by application of lemma~\ref{lem:SoluL2}, the 
solutions will all have bounded second order moment.  \\
	\noindent {\it Existence:}\\
	Let \olivier{$X^0=(X^0_t=\{X^{0\,\alpha}_t,\,\alpha=1\cdots P\})\in \mathcal{M}(\mathcal{C})$ be} a given stochastic process, and define the sequence of probability distributions $(X^k)_{k \geq 0}$ on $\M(\C)$ defined by induction by $X^{k+1}=\Phi(X^k)$. \olivier{Define also} a sequence of processes $Z^k$, $k \geq 0$, independent of the \olivier{sequence} of processes $X^k$ and having the same law. \olivier{We note this as ``$X$ and $Z$ i.i.d.'' below.} We stop the processes at the time $\tau^k_{U}$ the first hitting time of the \olivier{norm} of $X^k$ to the constant value $U$. For convenience, we will make an abuse of notation in the proof and denote $X^k_t=X^k_{t\wedge \tau^k_U}$.  \olivier{This implies that $X^k_t$ belongs to $B_U^d$ for all times $t \in [0,T]$.}

	
	\olivier{Using the notations introduced for equation~\eqref{eq:CompactMFE} we decompose the difference $X^{k+1}_t-X^k_t$ as follows:
	\begin{align*}
		X^{k+1}_t-X^k_t & =\underbrace{\int_0^t \Big (f(s,X^{k}_s) - f(s,X^{k-1}_s)\Big)\,ds}_{A_t} \\
		& \quad + \underbrace{\int_0^t \Exp_Z \Big [ b(X^{k}_s, Z^{k}_s) - b(X^{k-1}_s, Z^{k-1}_s) \Big] \, ds}_{B_t} \\ 
		 & \quad +\underbrace{\int_0^t \Big ( g(s,X^{k}_s)-g(s,X^{k-1}_s) \Big)\, dW_s}_{C_t}\\
		 & \quad + \underbrace{\int_0^t \Exp_Z \Big [ \beta(X^{k}_s, Z^{k}_s) - \beta(X^{k-1}_s, Z^{k-1}_s) \Big] \, \odot dB_s}_{D_t} \\
	\end{align*}
        }
        \olivier{and find an upperbound for $M^k_t:=\Exp\Big[\sup_{s\leq t} \Vert X^{k+1}_s-X^k_s \Vert^2 \Big] $ by finding
        upperbounds for  the corresponding norms of}
	the four terms $A_t$, $B_t$, $C_t$ and $D_t$. \olivier{Applying the discrete Cauchy-Schwartz} inequality we have:
	\[\Vert X^{k+1}_t-X^k_t \Vert^2 \leq 4 \Big( \Vert A_t\Vert^2 + \Vert B_t\Vert^2 + \Vert C_t\Vert^2+ \Vert D_t\Vert^2\Big)\]
	and treat each term separately.  \olivier{The upperbounds for the first two terms are obtained using the Cauchy-Schwartz inequality, those of the last two terms using the Burkholder-Davis-Gundy martingale moment inequality.}
	
	\olivier{The term $A_t$ is easily controlled using the Cauchy-Schwarz inequality. In detail
	\begin{align*}
          \| A_s \|^2 & = \left\| \int_0^s \left(f(u,X_u^k)-f(u,X_u^{k-1})\right) \, du      \right\|^2\\
           \text{(Cauchy-Schwartz)} & \leq T \int_0^s \left\| f(u,X_u^k)-f(u,X_u^{k-1})  \right\|^2\,du\\
          \text{(assumption~\ref{Assump:LocLipsch})} & \leq K_U^2 T \int_0^s \left\|X_u^k-X_u^{k-1}\right\|^2\,du 
          \end{align*}
Taking the $\sup$ of both sides of the last inequality we obtain
\[
\sup_{s \leq t}  \| A_s \|^2 \leq  K_U^2 T \int_0^t \left\|X_s^k-X_s^{k-1}\right\|^2\,ds \leq  K_U^2 T \int_0^t \sup_{u \leq s} \left\|X_u^k-X_u^{k-1}\right\|^2\,ds,
\]
from which follows the fact that
\[
\Exp\left[\sup_{s \leq t}  \| A_s \|^2\right] \leq K_U^2 T \int_0^t \Exp\left[\sup_{u \leq s} \left\|X_u^k-X_u^{k-1}\right\|^2 \right]\,ds
\]
}
\olivier{The term $B_t$ is controlled using the Cauchy-Schwartz inequality, the uniform continuity of the function $b$ and the fact that the processes $X$ and $Z$ are independent with the same law. In detail
\begin{align*}
\|B_s\|^2 &= \left\| \int_0^s \Exp_{Z} \left[ b(X_u^k,Z_u^k)- b(X_u^{k-1},Z_u^{k-1}) \right] \, du \right\|^2\\
\text{(Cauchy-Schwartz)} &\leq \john{T\int_0^s \left\| \Exp_{Z} \left[ b(X_u^k,Z_u^k)- b(X_u^{k-1},Z_u^{k-1}) \right] \right\|^2 \, du}\\
\text{(Cauchy-Schwartz)} &\leq T\int_0^s  \Exp_{Z} \left[ \left\| b(X_u^k,Z_u^k)- b(X_u^{k-1},Z_u^{k-1})\right\|^2\right] \, du\\
\text{(assumption \ref{Assump:LocLipschb})} & \leq 2TL^2 \int_0^s  \Exp_{Z}\left[  \left\|X_u^k-X_u^{k-1}\right\|^2+ \left\|Z_u^k-Z_u^{k-1}\right\|^2\right]\, du\\
\text{($X$ and $Z$ i.i.d.)}& \leq 2TL^2 \int_0^s\left( \left\|X_u^k-X_u^{k-1}\right\|^2+ \Exp \left[ \left\|X_u^k-X_u^{k-1}\right\|^2\right]\right)\, du\\
\end{align*}
}
\olivier{
Taking the $\sup$ of both sides of the last inequality we obtain
\begin{align*}
\sup_{s \leq t} \|B_s\|^2 &\leq  2TL^2 \int_0^t \left( \left\|X_s^k-X_s^{k-1}\right\|^2+ \Exp \left[ \left\|X_s^k-X_s^{k-1}\right\|^2\right]\right)\, ds\\
& \john{\leq  2TL^2 \int_0^t \left( \sup_{u \leq s} \left\|X_u^k-X_u^{k-1}\right\|^2+ \sup_{u \leq s} \Exp \left[ \left\|X_u^k-X_u^{k-1}\right\|^2\right]\right)\, ds}\\
& \leq  2TL^2 \int_0^t \left( \sup_{u \leq s} \left\|X_u^k-X_u^{k-1}\right\|^2+  \Exp \left[\sup_{u \leq s} \left\|X_u^k-X_u^{k-1}\right\|^2\right]\right)\, ds,
\end{align*}
}
\olivier{from which follows the fact that
\[
\Exp\left[\sup_{s \leq t} \|B_s\|^2 \right] \leq 4TL^2 \int_0^t \Exp \left[\sup_{u \leq s} \left\|X_u^k-X_u^{k-1}\right\|^2\right]\, ds
\]
}

\olivier{The term $C_t$ is controlled using the fact that it is a martingale and applying the Burkholder-Davis-Gundy martingale moment inequality:
\begin{align*}
 \Exp\left[\sup_{s \leq t} \|C_s\|^2 \right] 
 &\leq 4 \int_0^t \Exp\left[ \left\| g(s,X^{k}_s)-g(s,X^{k-1}_s) \right\|^2 \right]\,ds\\
 \text{(assumption~\ref{Assump:LocLipsch})} &\leq 4K_U^2 \int_0^t  \Exp\left[ \left\|X^{k}_s-X^{k-1}_s \right\|^2 \right]\,ds\\
&\leq 4K_U^2 \int_0^t \Exp\left[\sup_{u \leq s}  \left\| X^{k}_u-X^{k-1}_u \right\|^2 \right]\,ds\\
\end{align*} 
The term $D_t$ is also controlled using the fact that it is a martingale and applying the Burkholder-Davis-Gundy martingale moment inequality:

\begin{align*}
 \Exp\left[\sup_{s \leq t} \|D_t\|^2 \right] 
&\leq  4  \int_0^t \Exp\left[ \left\|\Exp_Z \left[ \beta(X^{k}_s, Z^{k}_s) - \beta(X^{k-1}_s, Z^{k-1}_s) \right] \right\|^2 \right]\,ds\\
\text{(Cauchy-Schwartz)} &\leq 4  \int_0^t \Exp\left[\Exp_Z \left[ \left\|\beta(X^{k}_s, Z^{k}_s) - \beta(X^{k-1}_s, Z^{k-1}_s)\right\|^2  \right] \right]\,ds\\
\text{(assumption \ref{Assump:LocLipschb})} & \leq 8 L^2  \int_0^t \Exp\left[\Exp_Z \left[ \left\|X^{k}_s-X^{k-1}_s\right\|^2+\left\|Z^{k}_s-Z^{k-1}_s\right\|^2\right]\right]\,ds\\
\text{($X$ and $Z$ i.i.d.)}& \leq 16 L^2  \int_0^t  \Exp\left[\left\|X^{k}_s-X^{k-1}_s\right\|^2\right]\,ds\\
& \leq 16 L^2  \int_0^t  \Exp\left[\sup_{u \leq s} \left\|X^{k}_u-X^{k-1}_u\right\|^2\right]\,ds
\end{align*}
}
	Putting all this together, we get:
	\begin{multline}\label{eq:Bounds}
		\Exp\Big[\sup_{s\leq t} \Vert X^{k+1}_s-X^k_s \Vert^2 \Big] \leq 4(T+4)(K_U^2 + 4L^2) \int_0^t \Exp[ \sup_{u\leq s} \Vert X^{k}_u - X^{k-1}_u \Vert^2] ds
	\end{multline}

From the relation $M^k_t \leq K'' \int_0^t M^{k-1}_s\,ds$ with \olivier{$K''=4(T+4)(K_U^2 + 4L^2)$}, we get by an immediate recursion:
\begin{align}
	\nonumber M^k_t &\leq (K'')^k \int_0^t\int_0^{s_1}\cdots \int_0^{s_{k-1}}M^0_{s_k}\;ds_1\ldots ds_k\\
	\label{eq:CauchySeq} & \leq \frac{(K'')^k\, t^k}{k!} M^0_T
\end{align}
and \olivier{$M^0_T$} is finite because the processes are bounded. Bienaym\'e-Tchebychev inequality \olivier{and~\eqref{eq:CauchySeq}} now give:
\[
P\left(\sup_{s\leq t} \Vert X^{k+1}_s-X^k_s \Vert^{\olivier{2}} > \frac 1 {2^{\olivier{2(k+1)}}}\right) \leq 4\frac{(4\,K'' t)^k}{k!} M^0_T
\]
ans this upper bound is the term of a convergent series. From the Borel-Cantelli lemma stems that for almost any $\omega\in \Omega$ there exists a positive integer \olivier{$k_0(\omega)$ ($\omega$ denotes an element of the probability space $\Omega$)} such that
\[
\sup_{s\leq t} \Vert X^{k+1}_s-X^k_s \Vert^{\olivier{2}} \leq \frac 1 {2^{\olivier{2(k+1)}}} \qquad \forall \; \olivier{k\geq k_0(\omega)},
\]
\olivier{and hence
\[
\sup_{s\leq t} \Vert X^{k+1}_s-X^k_s \Vert \leq \frac{1}{2^{k+1}} \qquad \forall \; k\geq k_0(\omega).
\]
}
It follows that with probability $1$ the partial sums:
\[X^0_t+\sum_{k=0}^n (X^{k+1}_t-X^k_t) = X^n_t\]
are uniformly (in $t\in [0,T]$) convergent. Denote by $\bar{X}_t$ the thus defined limit. It is clearly continuous and $\F_t$ adapted. On the other hand, \olivier{the inequality~\eqref{eq:CauchySeq} shows} that for every fixed $t$, the sequence $\{X^n_t\}_{n \geq 1}$ \olivier{is} a Cauchy sequence in $\L^2$. 
\olivier{Lemma~\ref{lem:SoluL2} shows} that $\bar{X} \in \M^2(\C)$. 

\john{It is easy to show using routine methods that $\bar{X}$ indeed satisfies the equation~\eqref{eq:MFE}}.

\olivier{To complete the proof we use} a standard truncation property. This method replaces the function $f$ by the truncated function:
\[f_U(t,x)= \begin{cases} f(t,x) & \Vert x \Vert \leq U \\ f(t,Ux/\Vert x \Vert ) & \Vert x \Vert > U \end{cases}\]
and similarly for $g$. The functions $f_U$ and $g_U$ are globally Lipchitz-continuous, hence the previous proof shows that there exists a unique solution $\bar{X}_U$ to equations~\eqref{eq:MFE} associated with the truncated functions. 
\olivier{This solution satisfies the equation
\begin{multline}\label{eq:integral}
	\bar{X}_U(t)=X_0+\int_0^t\left(f_U(t,\bar{X}_U(s)) + \Exp_{\bar{Z}}[b(\bar{X}_U(s), \bar{Z}_s)] \right)\,ds +\int_0^t g_U(t,\bar{X}_U(s))\,dW_s \\
	+ \int_0^t\Exp_{\bar{Z}}\left[\beta(\bar{X}_U(s), \bar{Z}_s)\right] \odot dB_s\quad t \in [0,\,T]
\end{multline}
}
Let us now define the stopping time:
\[
\tau_U=\inf\{t\in [0,T], \Vert \bar{X}_U(t)\Vert \geq U\}.
\]
\olivier{It is easy to show that 
\begin{equation}\label{eq:increasing}
\bar{X}_U(t)=\bar{X}_{U'}(t) \quad \text{if} \quad  0\leq t \leq \tau_U,\,U' \geq U,
\end{equation}
}
implying that the sequence of stopping times $\tau_U$ is increasing. 
\olivier{Using lemma~\ref{lem:SoluL2} which implies that the solution to \eqref{eq:MFE} is almost surely bounded, for almost all $\omega \in \Omega$, there exists $U_0(\omega)$ such that $\tau_U=T$ for all $U\geq U_0$. Now define $\bar{X}(t)=\bar{X}_{U_0}(t)$, $t \in [0,\,T]$. Because of~\eqref{eq:increasing} we have $\bar{X}(t \wedge \tau_U)=\bar{X}_U(t \wedge \tau_U)$ and it follows from~\eqref{eq:integral} that
\begin{align*}
	\bar{X}(t\wedge\tau_U) &= X_0 +  \int_{0}^{t\wedge \tau_U} \Big (f_U(s,\bar{X}_s) + \Exp_{\bar{Z}} [b(\bar{X}_s,\bar{Z}_s)] \Big) \,ds + \int_0^{t\wedge \tau_U} g_U(s,\bar{X}_s)\,dW_s\\
        & \quad \quad +\int_0^{t\wedge \tau_U}  \Exp_{\bar{Z}}\left[\beta(\bar{X}_U(s), \bar{Z}_s)\right] \odot dB_s\\
	&= X_0 +  \int_{0}^{t\wedge \tau_U} \Big (f(s,\bar{X}_s) + \Exp_{\bar{Z}} [b(\bar{X}_s,\bar{Z}_s)] \Big) \,ds + \int_0^{t\wedge \tau_U} g(s,\bar{X}_s)\,dW_s\\
        & \quad \quad +\int_0^{t\wedge \tau_U}  \Exp_{\bar{Z}}\left[\beta(\bar{X}_U(s), \bar{Z}_s)\right] \odot dB_s
\end{align*}
and letting $U\to \infty$ we have shown the existence of solution to equation~\eqref{eq:MFE} which by lemma~\ref{lem:SoluL2} is square integrable. }

\noindent{\it Uniqueness:}\\
Assume that $X$ and $Y$ are two solutions of the mean-field equations \eqref{eq:MFE}. From lemma \ref{lem:SoluL2}, we know that both solutions are in $\M^2(\C)$. Moreover, using the bound \eqref{eq:Bounds}, we directly obtain the inequality:
\[
	\Exp\Big[\sup_{s\leq t} \Vert X_s-Y_s \Vert^2 \Big] \leq K'' \int_0^t \Exp \Big[ \sup_{u\leq s} \Vert X_u - Y_u \Vert^2 \Big]  \, ds
\]
which by Gronwall's theorem directly implies \[\Exp\Big[\sup_{s\leq t} \Vert X_s-Y_s \Vert^2 \Big]=0\] which ends the proof.
\end{proof}
We have proved the well-posedness of the mean-field equations. It remains to show that the solutions to the network equations converge to the solutions of the
mean-field equations. This is what is achieved in the next theorem.
\begin{theorem}\label{thm:Principal2}
Under the assumptions \ref{Assump:LocLipsch} to \ref{Assump:MonotoneGrowth}  
the following holds true:
	\begin{itemize}
		\item {\bf Convergence\footnote{The type of convergence is specified in the proof given in Appendix~\ref{section:proofs}.:}} For each neuron $i$ of population $\alpha$, the law of the multidimensional process $X^{i,N}$ converges towards the law of the solution of the mean-field equation related to population $\alpha$, namely $\bar{X}^{\alpha}$. 
		\item {\bf Propagation of chaos:} For any $k\in \N^*$, and any $k$-uplet $(i_1, \ldots, i_k)$, the law of the process $(X^{i_1,N}_t, \ldots, X^{i_n,N}_t, t\leq T)$ converges towards\footnote{\olivier{The notation $m_t^\alpha$ was introduced right after equation~\eqref{eq:MFE}.}} \olivier{$m_t^{p(i_1)}\otimes\ldots\otimes m_t^{p(i_n)}$}, i.e. the asymptotic processes have the law of the solution of the mean-field equations and are all independent. 
	\end{itemize}
\end{theorem}
This theorem has \olivier{important} implications in neuroscience that we discuss in section~\ref{section:Discussion}. Its proof is given in Appendix~\ref{section:proofs}.

\section{Numerical simulations}\label{section:numerics}
At this point, we have provided a compact description of the activity of the network when the number of neurons tends to infinity. However, the structure of the solutions of these equations is complicated to understand from the implicit mean-field equations~\eqref{eq:MFE} and of their variants (such as the McKean-Vlasov-Fokker-Planck equations~\eqref{eq:McKeanVlasofFP}). In this section, we present some classical ways to numerically approximate the solutions to these equations and give some indications about the rate of convergence and the accuracy of the simulation. \olivier{These numerical schemes allow us to compute and visualise the solutions. We then compare  the results of the two schemes for a network of Fitzhugh-Nagumo neurons belonging to a single population and show their good agreement.}

\olivier{The main difficulty one faces when developing numerical schemes for equations~\eqref{eq:MFE} and~\eqref{eq:McKeanVlasofFP} is that they are non local. By this we mean that, in the case of the McKean-Vlasov equations, they contain the expectation of a certain function under the law of the solution to the equations, see~\eqref{eq:MFE}. In the case of the corresponding Fokker-Planck equation, it contains integrals of the probability density functions which is solution to the equation, see~\eqref{eq:McKeanVlasofFP}.}   

\subsection{Numerical simulations of the McKean-Vlasov equations}\label{subsection:mckv}
The fact that McKean-Vlasov equations involve an expectation of a certain function under the law of the solution of the equation makes them particularly hard to simulate directly. One is often reduced to use Monte-Carlo simulations to compute this expectation, which amounts to simulating the solution of the network equations themselves (see~\cite{talay-vaillant:03}). \olivier{This is the method we used. In its simplest fashion, it consists of a Monte-Carlo simulation where one  solves numerically the $N$ network equations~\eqref{eq:Network} with the classical Euler-Maruyama method a number of times with different initial conditions, and averages the trajectories of the solutions over the number of simulations.}

\olivier{In detail let $\Delta t>0$ and $N\in \N^*$. The discrete-time dynamics implemented in the stochastic numerical simulations consists in simulating $M$ times a $P$-population discrete-time process $(\tilde{X}^i_{n}, n \leq T/\Delta t, i=1 \cdots N)$, solution of the recursion, for $i$ in population $\alpha$:
\begin{multline}\label{eq:NumericalScheme}
	\tilde{X}^{i,r}_{n+1} = \tilde{X}^{i,r}_{n} + \Delta t \Big\{f_{\alpha}(t,\tilde{X}^{i,r}_{n}) \, dt + \sum_{\gamma=1}^P \frac{1}{N_{\gamma}} \sum_{j=1, p(j)=\gamma}^{N_{\gamma}} b_{\alpha\gamma}(\tilde{X}^{i,r}_{n},\tilde{X}^{j,r}_{n})\Big\} \\
	+ \sqrt{\Delta t} \Big\{g_{\alpha}(t,\tilde{X}^{i,r}_{n})  \xi^{i,r}_{n+1}
	+ \sum_{\gamma=1}^P \frac{1}{N_{\gamma}} \sum_{j=1, p(j)=\gamma}^{N_{\gamma}} \beta_{\alpha\gamma}(\tilde{X}^{i,r}_{n} , \tilde{X}^{j,r}_{n}) \cdot \zeta^{i\gamma}_{n+1}\Big\}
\end{multline}
where $\xi^{i,r}_{n}$ and $\zeta^{i\gamma,r}_{n}$ are independent $d$- and $\delta$-dimensional standard normal random variables.} \olivier{The initial conditions $\tilde{X}^{i,r}_1$, $i=1,\cdots,N$ are drawn independently from the same law within each population for each Monte-Carlo simulation $r=1,\cdots,M$. One then chooses one neuron $i_\alpha$ in each population $\alpha=1,\cdots,P$.}
\olivier{If the size $N$ of the population is large enough, theorem~\ref{thm:Principal2} states that the law, noted $p_\alpha(t,X)$,  of $X^{i_\alpha}$ should be close to that of the solution $\bar{X}^\alpha$ of the mean-field equations, for $\alpha=1,\cdots,P$. Hence, in effect, simulating the network is a good approximation, see below, of the simulation of the mean-field or McKean-Vlasov equations \cite{bossy-talay:97,talay-vaillant:03}. An approximation of $p_\alpha(t,X)$ can be obtained from the Monte-Carlo simulations by quantizing the phase space and incrementing the count of each bin whenever the trajectory of the $i_\alpha$ neuron at time $t$ falls into that particular bin. The resulting histogram can then be compared to the solution of the McKean-Vlasov-Fokker-Planck equation~\eqref{eq:McKeanVlasofFP} corresponding to population $\alpha$ whose numerical solution is described next.
}

The mean square error between the solution of the numerical recursion~\eqref{eq:NumericalScheme} $\tilde{X}^i_{n}$ and the solution of the mean-field equations~\eqref{eq:MFE} $\bar{X}^i_{n\Delta t}$ is of order $O(\sqrt{\Delta t} + 1/\sqrt{N})$, the first term being related to the error made by approximating the solution of the network of size $N$, $X^{i,N}_{n\Delta t}$ by an Euler-Maruyama method, and the second term to the convergence of $X^{i,N}_{n\Delta t}$ towards the mean-field equation $\bar{X}^i_{n\Delta t}$ \johnNew{ when considering globally Lipschitz-continuous dynamics (see proof of theorem~\ref{thm:Principal2} in appendix~\ref{section:proofs}).}

\subsection{Numerical simulations of the McKean-Vlasov-Fokker-Planck equation}
For solving the McKean-Vlasov-Fokker-Planck equation~\eqref{eq:McKeanVlasofFP} we have used the  {\it method of lines} ~\cite{schiesser:91,schiesser-griffiths:09}.
Its basic idea is to discretize the phase space and to keep the time continuous.
In this way \olivier{the values $p_\alpha(t,X)$, $\alpha=1,\cdots,P$ of  the probability density function  of population $\alpha$ at each sample point $X$ of the phase space are the solutions of $P$ ordinary differential equations where the independent variable is the time.}
Each sample point in the phase space generates $P$ ODEs, resulting in a system of coupled ODEs. The solutions to this system yield the values of the probability density functions $p_\alpha$ solution of~\eqref{eq:McKeanVlasofFP} at the sample points.
The computation of the integral terms that appear in the McKean-Vlasov-Fokker-Planck equation is achieved through a recursive scheme\john{, the  Newton-Cotes method of order 6 \cite{ueberhuber:97}. The dimensionality of the space being large and numerical errors increasing with the dimensionality of the integrand, such precise integration schemes are necessary}. For an arbitrary real function $f$ to be integrated between the values $x_1$ and $x_2$, \john{this} numerical scheme reads:
\begin{multline*}
	\int_{x_{1}}^{x_{2}}f(x)dx \approx \frac{5}{288}\Delta x\sum_{i=1}^{M/5}[19f(x_{1}+(5i-5)\Delta x)+75f(x_{1}+(5i-4)\Delta x)+ \\
        +50f(x_{1}+(5i-3)\Delta x)+50f(x_{1}+(5i-2)\Delta x)+75f(x_{1}+(5i-1)\Delta x)+19f(x_{1}+5i\Delta x)]
\end{multline*}
where $\Delta x$ is the integration step and $M=(x_{2}-x_{1})/\Delta x$ is chosen to be an integer multiple of 5.

The discretization of the derivatives with respect to the phase space parameters is done through the following fourth-order central difference scheme:
\begin{equation*}
	\frac{df(x)}{dx} \approx \frac{f(x-2\Delta x)-8f(x-\Delta x)+8f(x+\Delta x)-f(x+2\Delta x)}{12\Delta x},
\end{equation*}
for first-order derivatives, and
\begin{equation*}
	\frac{d^{2}f(x)}{dx^{2}} \approx \frac{-f(x-2\Delta x)+16f(x-\Delta x)-30f(x)+16f(x+\Delta x)-f(x+2\Delta x)}{12\Delta x^{2}},
\end{equation*}
for the second-order derivatives, see \cite{morton-mayers:05}.

Finally we have used a Runge-Kutta method of order 2 (RK2) for the numerical integration of the resulting system of ODEs.

\subsection{Comparison between the solutions to the network and the mean-field equations}
As advertised, we illustrate these ideas with the example of a network of 100 Fitzhugh-Nagumo neurons belonging to one, excitatory, population. We choose a finite volume outside of which we assume that the p.d.f., solution of~\eqref{eq:McKeanVlasofFP}, is zero. We then discretize this volume with  $n_{V}n_{W}n_{Y}$ points defined by
\begin{align*}\label{eq:DiscretizationPoints}
	n_{V} &\stackrel{\text{def}}{=}(V_{max}-V_{min})/\Delta V \\
	n_{w} &\stackrel{\text{def}}{=}(w_{max}-w_{min})/\Delta w \\
        n_{y} &\stackrel{\text{def}}{=}(y_{max}-y_{min})/\Delta y
\end{align*}
where $V_{min}$, $V_{max}$, $w_{min}$, $w_{max}$, $y_{min}$, $y_{max}$ define the volume in which we solve the McKean-Vlasov-Fokker-Planck equation and estimate the histogram defined in section~\ref{subsection:mckv} while $\Delta V$, $\Delta w$, $\Delta y$ are the quantization steps in each dimension of the phase space.

In general, the total number of coupled ODEs is the product $P n_{V}n_{w}n_{y}$ (in our case we chose $P=1$). This can become fairly large if we increase the precision of the phase space discretization.
\johnNew{Moreover, increasing the precision of the simulation in the phase space requires, in order to ensure the numerical stability of the method of lines, to decrease the time step $\Delta t$ used in the RK2 scheme. This can strongly impact the efficiency of the numerical method, see section~\ref{subsection:gpus}.}


\olivier{In the simulations shown in the lefthand parts of figures~\ref{fig:FNfocus} and~\ref{fig:FNcycle} we have used  one population of $100$ excitatory FitzHugh-Nagumo neurons connected with chemical synapses. We performed $10,000$ Monte Carlo simulations of the network equations~\eqref{eq:FNNetwork} with the Euler-Maruyama method in order to approximate the probability density. The model for the time variation of the synaptic
weights is the simple model. The p.d.f. $p(0,V,w,y)$ of the initial condition is Gaussian and reads
\begin{equation}\label{eq:InitialConditionPdf}
	p(0,V,w,y) = \frac{1}{(2\pi)^{3/2}\sigma_{V_{0}}\sigma_{w_{0}}\sigma_{y_{0}}}e^{-\frac{(V-\overline{V}_{0})^{2}}{2\sigma_{V_{0}}^{2}}-\frac{(w-\overline{w}_{0})^{2}}{2\sigma_{w_{0}}^{2}}-\frac{(y-\overline{y}_{0})^{2}}{2\sigma_{y_{0}}^{2}}}
\end{equation}
The parameters are given in the first column of table~\ref{tab:Neural-Network-Parameters}. In this table, the parameter $t_{fin}$ is the time at which we stop the computation of the trajectories in the case of the network equations, and the computation of the solution of McKean-Vlasov-Fokker-Planck equation in the case of the mean-field equations. The sequence $[0.5,\,1.5,\,1.8,\,3.0]$ indicates that we compute the solutions at those four time instants corresponding to the four rows of figure~\ref{fig:FNcycle}. The means take two different values corresponding to the two numerical experiments described below. The phase space has been quantized with the parameters shown in the second column of the same table to solve the Fokker-Planck equation. This quantization has also been used to build the histograms that represent the marginal probability densities with respect to the pair $(V,y)$ of coordinates of the state vector of a particular neuron. These histograms have then been interpolated to build the surfaces shown in the lefthand side of figures~\ref{fig:FNfocus} and~\ref{fig:FNcycle}. The parameters of the Fitzhugh-Nagumo model are the same for each neuron of the population: they are shown in the third column of table~\ref{tab:Neural-Network-Parameters}. Two values of the external current are given, corresponding to the two numerical experiments whose results are shown below.}

\olivier{The parameters for the noisy model of maximum conductances of equation~\eqref{eq:weightssimple} are shown in the fourth column of the table. For these values of $\overline{J}$ and $\sigma_{J}$ the probability that the maximum conductances change sign is very small. Finally, the parameters of the chemical synapses are shown in the sixth column. The parameters $\Gamma$ and $\Lambda$ are those of the
$\chi$ function~\eqref{eq:ChiFunction}. The solutions are computed over an interval of $t_{fin}=3$ time units for figure~\ref{fig:FNfocus} and for the four values shown in brackets for figure~\ref{fig:FNcycle} with a time sampling of $\Delta t=0.1$.}\\

\johnNew{The marginals estimated from the trajectories of the network solutions are then compared} with those obtained from the numerical solution of the McKean-Vlasov-Fokker-Planck equation (see figures~\ref{fig:FNfocus} and~\ref{fig:FNcycle} right), using the method of lines explained above, and starting from the same initial conditions~\eqref{eq:InitialConditionPdf} as the neural network.

\begin{center}
\begin{table}[htbp]
\begin{tabular}{| c | c | c | c | c | }
\hline
Initial Condition                            &  Phase space             & FitzHugh-          & Synaptic           &       Synapse \\
                                             &                          & Nagumo            &   Weights          &               \\
\hline
  $t_{fin}=3.0$,                             &          $V_{min}=-3$    & $a=0.7$           & $\overline{J}=1$   &          $V_{rev}=1$  \\
$[0.5, 1.5, 1.8, 3.0]$                       &                          &                   &                    &                      \\
  $\Delta t=0.1$                             &           $V_{max}=3$    & $b=0.8$           &  $\sigma_{J}=0.2$  &        $a_r=1$   \\
$\overline{V}_{0}=-1.472,\ 0.0$              &         $\Delta V=0.1$   & $c=0.08$          &                    &        $a_d=1$  \\
$\sigma_{V_{0}}=0.2$                         &          $w_{min}=-2$    & $I=0,\ 0.7$       &                    &        $T_{max}=1$ \\
$\overline{w}_{0}=-0.965,\ 0.5$              &          $w_{max}=2$     & $\sigma_{ext}=0$  &                    &        $\lambda=0.2$\\
$\sigma_{w_{0}}=0.2$                         &          $\Delta w=0.1$  &                   &                    &        $V_{T}=2$   \\
 $\overline{y}_{0}=0.250,\ 0.3$              &      $y_{min}=0$         &                   &                    &         $\Gamma=0.1$ \\
$\sigma_{y_{0}}=0.05$                        &          $y_{max}=1$     &                   &                    &          $\Lambda=0.5$\\
                                             &        $\Delta y=0.06$   &                   &                    &           \\
\hline
\end{tabular}
\caption{Parameters used in the simulations of the neural network and for solving the McKean-Vlasov-Fokker-Planck equation whose results are shown in figures~\ref{fig:FNfocus} and \ref{fig:FNcycle}, see text.}
\label{tab:Neural-Network-Parameters}
\end{table}
\end{center}

\olivier{In a first series of experiments we choose the external input current value to be $I=0$ (see third column of table~\ref{tab:Neural-Network-Parameters}). This corresponds to a stable fixed point for the isolated Fitzhugh-Nagumo neuron.}\\
\olivier{The initial conditions (see first column of table~\ref{tab:Neural-Network-Parameters}) are chosen in such a way that on average, the initial point is very close to this fixed point. We therefore expect that the solutions of the neural network and the McKean-Vlasov-Fokker-Planck equation will not change much over time and will concentrate their mass around this fixed point. This is what is observed in figure~\ref{fig:FNfocus}, where the simulation of the neural network (lefthand side) is in excellent agreement with the results of the simulation of the McKean-Vlasov-Fokker-Planck equation (righthand side).}
\begin{figure}[htbp]
\centerline{\includegraphics[width=0.5\textwidth]{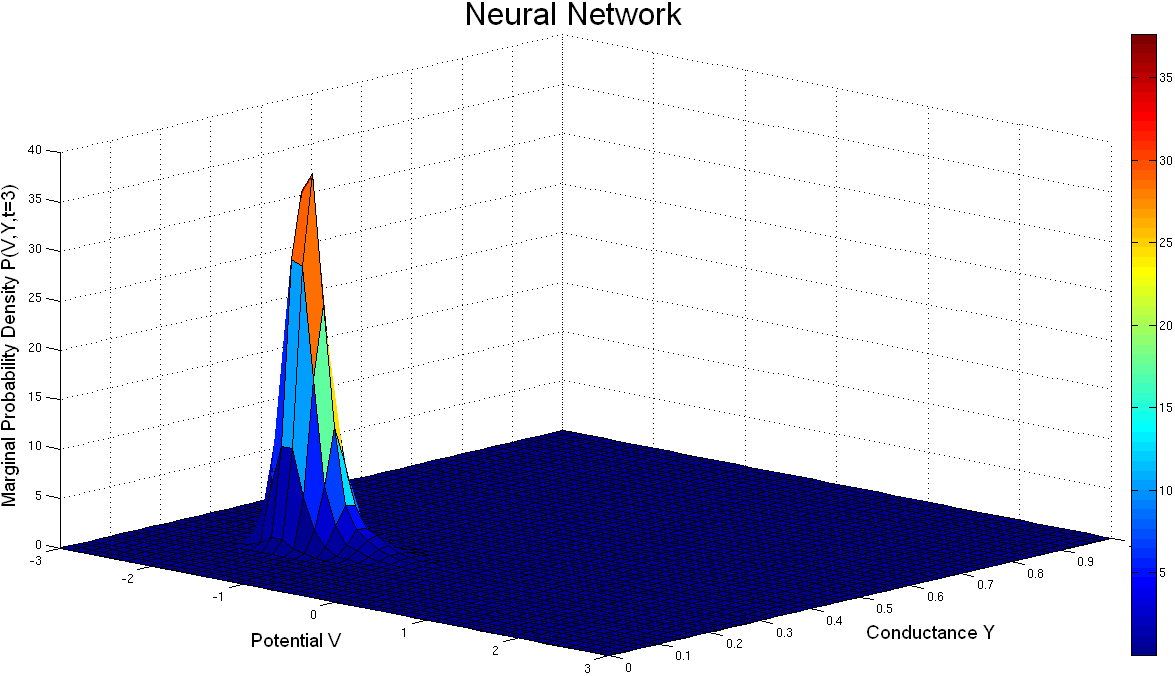}
\includegraphics[width=0.5\textwidth]{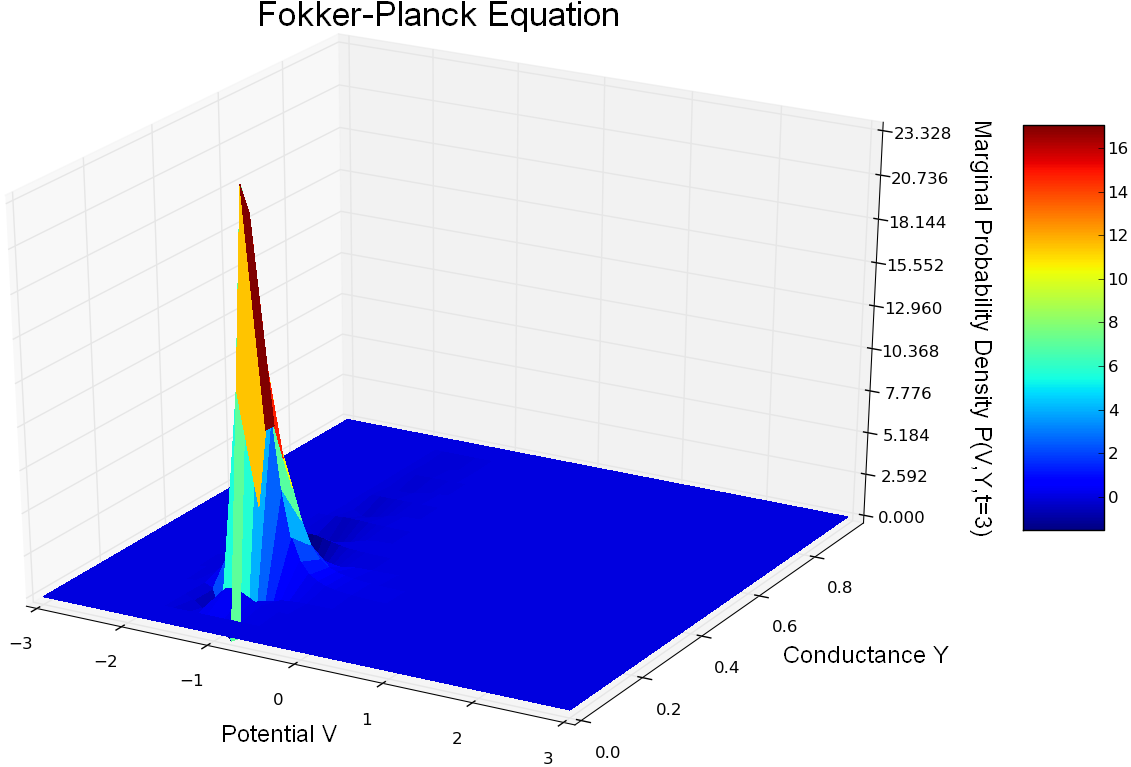}
}
\caption{\johnNew{Joint probability distribution of $(V,y)$, (Left): Monte-Carlo simulation of the McKean-Vlasov equations~\eqref{eq:MFE} and (Right): Simulation of the McKean-Vlasov-Fokker-Planck equation~\eqref{eq:McKeanVlasofFP}. Parameters of table ~\ref{tab:Neural-Network-Parameters}, $I=0$. These parameters corresponds to a stable equilibrium. Initial conditions (first column of table~\ref{tab:Neural-Network-Parameters}) are concentrated around this point. The two distributions are very similar and centered around the equilibrium point, see text.}}
\label{fig:FNfocus}
\end{figure}


\olivier{We then changed the value of the external current to $I=0.7$ (this value corresponds to the existence of a stable limit cycle for the isolated Fitzhugh-Nagumo neuron), keeping the other parameters values the same except for the initial conditions which now take the values $\overline{V}_{0}=0$, $\overline{w}_{0}=0.5$, $\overline{y}_{0}=0.3$.  We have also chosen four values of $t_{fin}$, i.e. $0.5$, $1.5$, $1.8$ and $3.0$, instead of only one in the previous numerical experiment. They correspond to four snapshots of the network dynamics. We therefore expect that the solutions of the neural network and the McKean-Vlasov-Fokker-Planck equation will concentrate their mass around the limit cycle. This is what is observed in figure~\ref{fig:FNcycle}, where the simulation of the neural network (lefthand side) is in excellent agreement with the results of the simulation of the McKean-Vlasov-Fokker-Planck equation (righthand side). Note that the densities display two peaks. These two peaks correspond to the fact that, depending upon the position of the initial condition with respect to the nullclines of the Fitzhugh-Nagumo equations, the points in the phase space follow two different classes of trajectories, as shown in figure~\ref{fig:PhaseSpaceCycle}. The two peaks then rotate along the limit cycle, see also section~\ref{subsection:gpus}.}\\
\begin{figure}[htbp]

\centerline{
\includegraphics[width=0.45\textwidth]{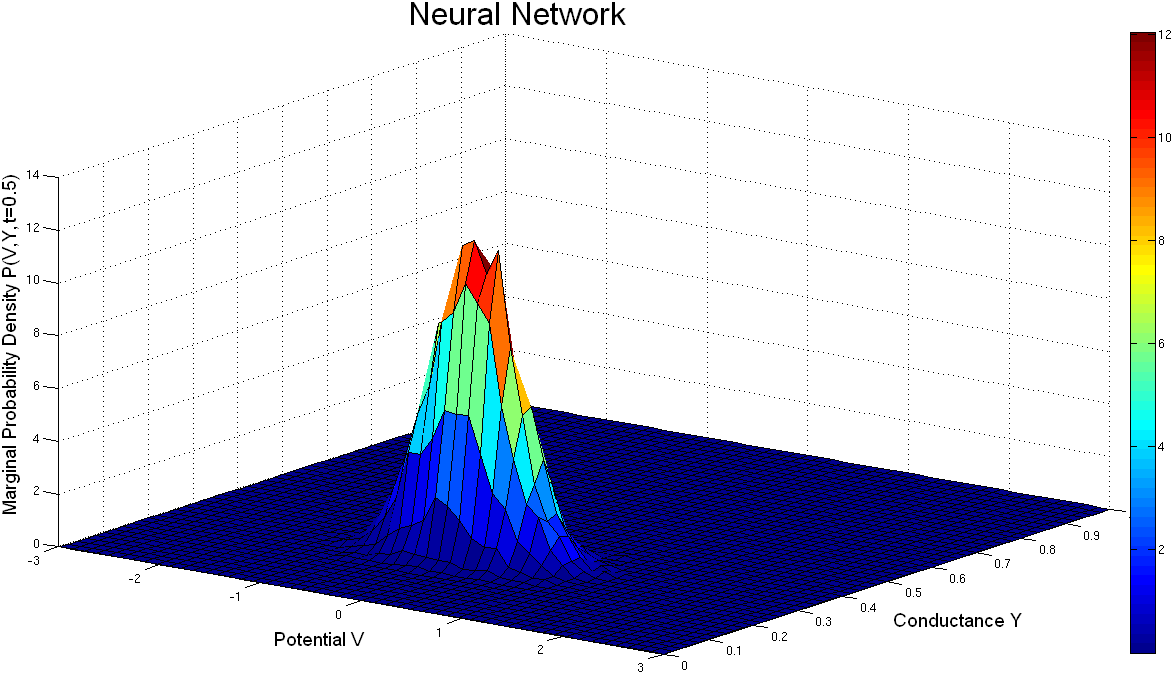}
\includegraphics[width=0.45\textwidth]{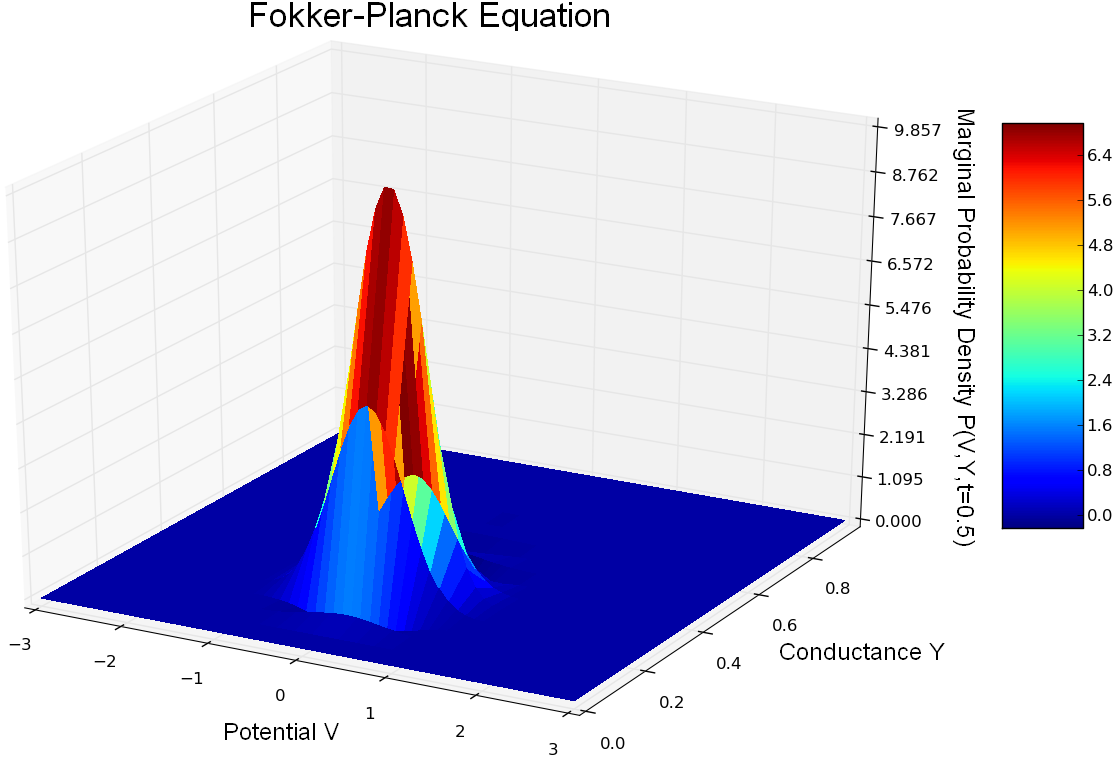}
}

\centerline{
\includegraphics[width=0.45\textwidth]{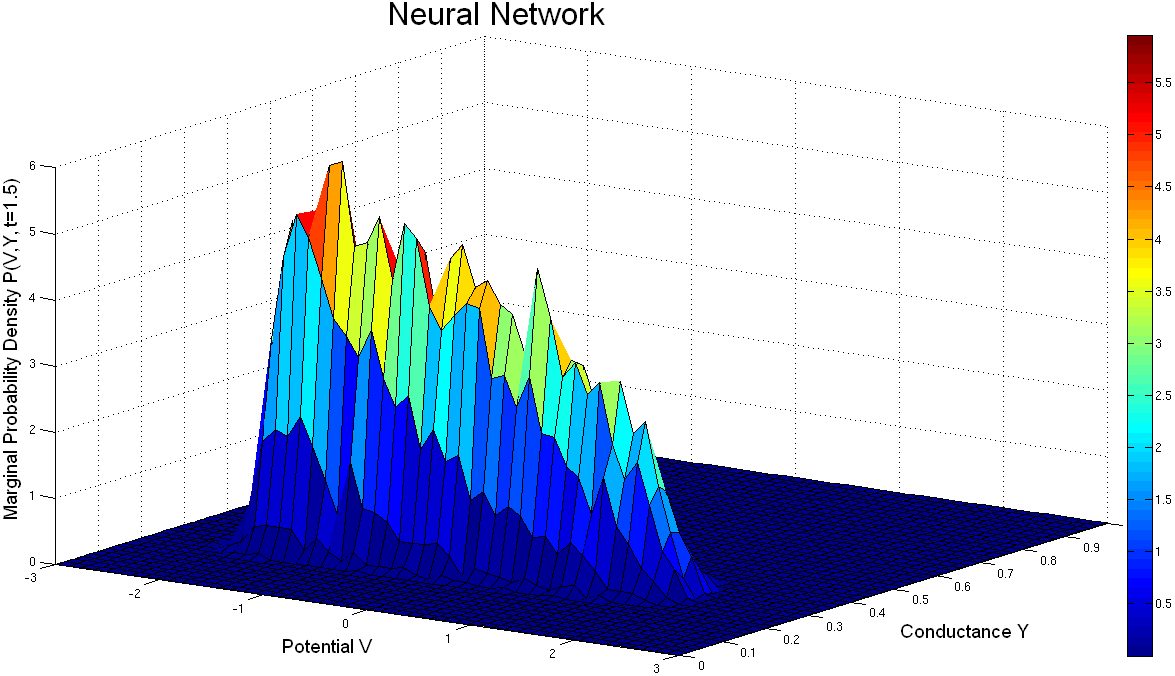}
\includegraphics[width=0.45\textwidth]{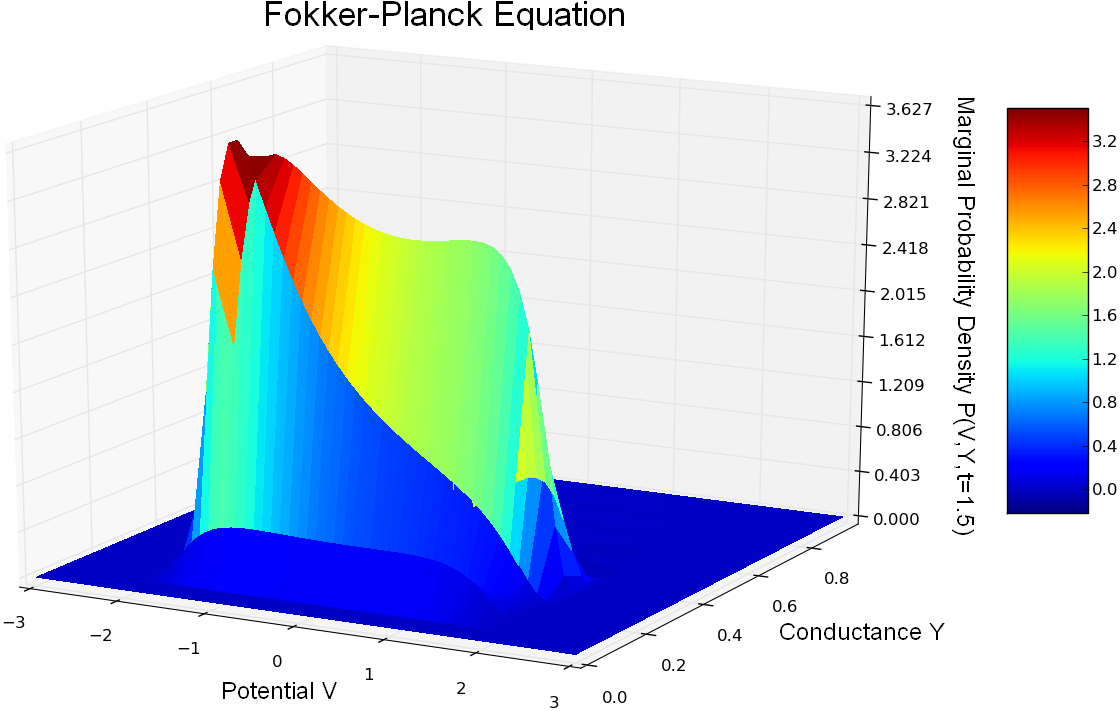}
}

\centerline{
\includegraphics[width=0.45\textwidth]{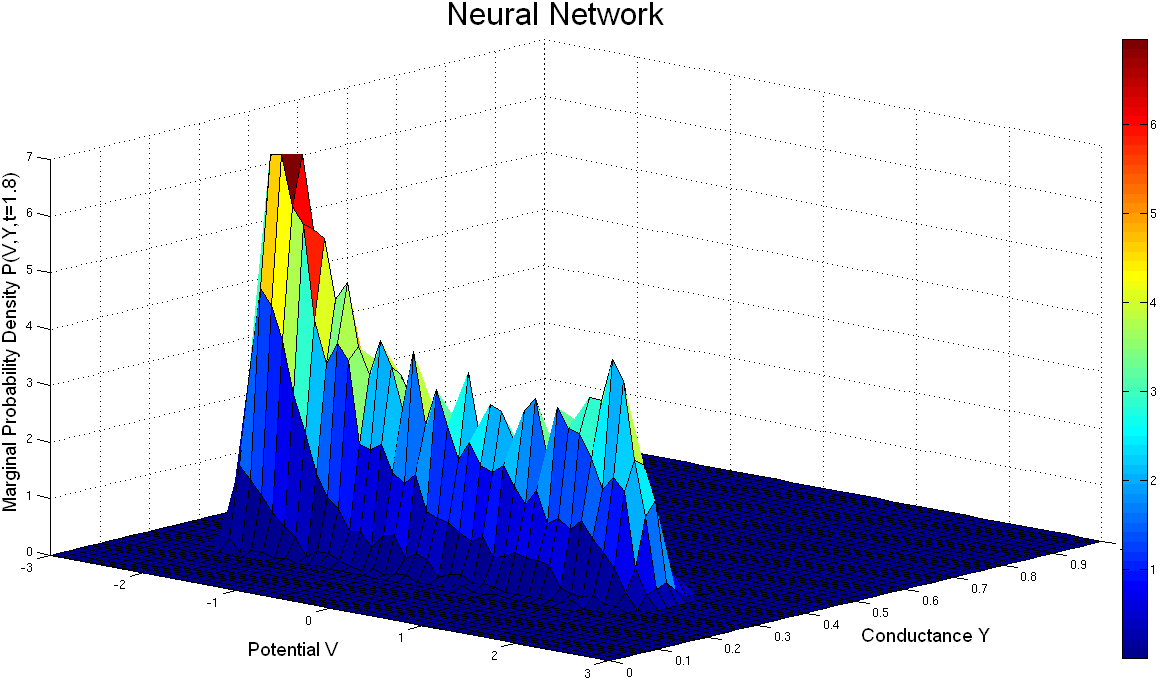}
\includegraphics[width=0.45\textwidth]{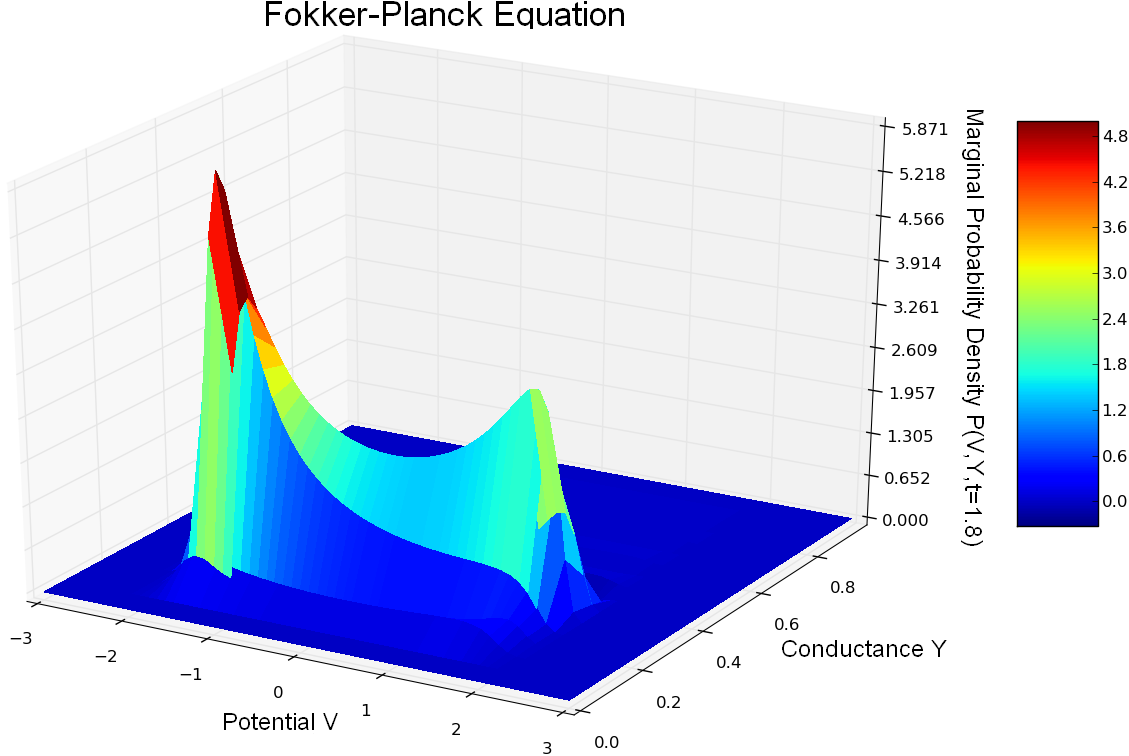}
}

\centerline{
\includegraphics[width=0.45\textwidth]{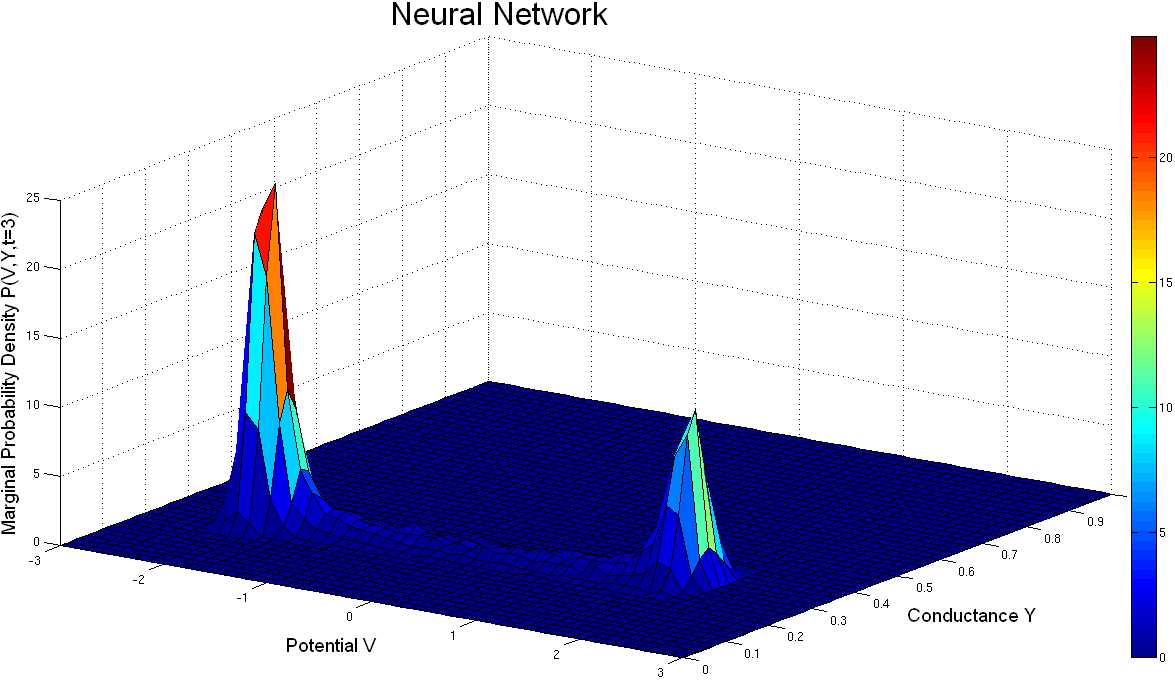}
\includegraphics[width=0.45\textwidth]{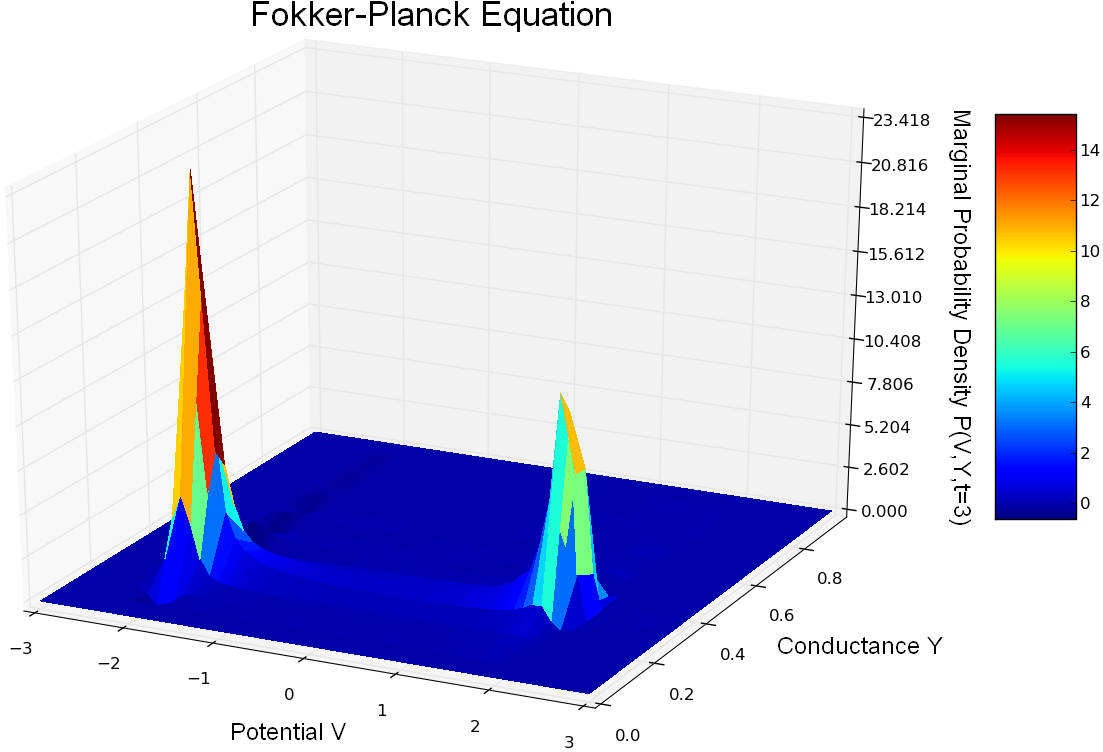}
}

\caption{\johnNew{Joint probability distribution of $(V,y)$ computed with the Monte-Carlo algorithm for the network equations~\eqref{eq:MFE} (left) compared with the solution of the McKean-Vlasov-Fokker-Planck equation~\eqref{eq:McKeanVlasofFP} (right), sampled at for four times $t_{fin}$. Parameters given in table~\ref{tab:Neural-Network-Parameters}, with a current $I= 0.7$ corresponding to a stable limit cycle. Initial conditions (first column of table~\ref{tab:Neural-Network-Parameters}) are concentrated inside this limit cycle. The two distributions are similar and centered around the limit cycle with two peaks, see text.}}
\label{fig:FNcycle}
\end{figure}


\begin{figure}[htbp]

\centerline{
\includegraphics[width=0.5\textwidth]{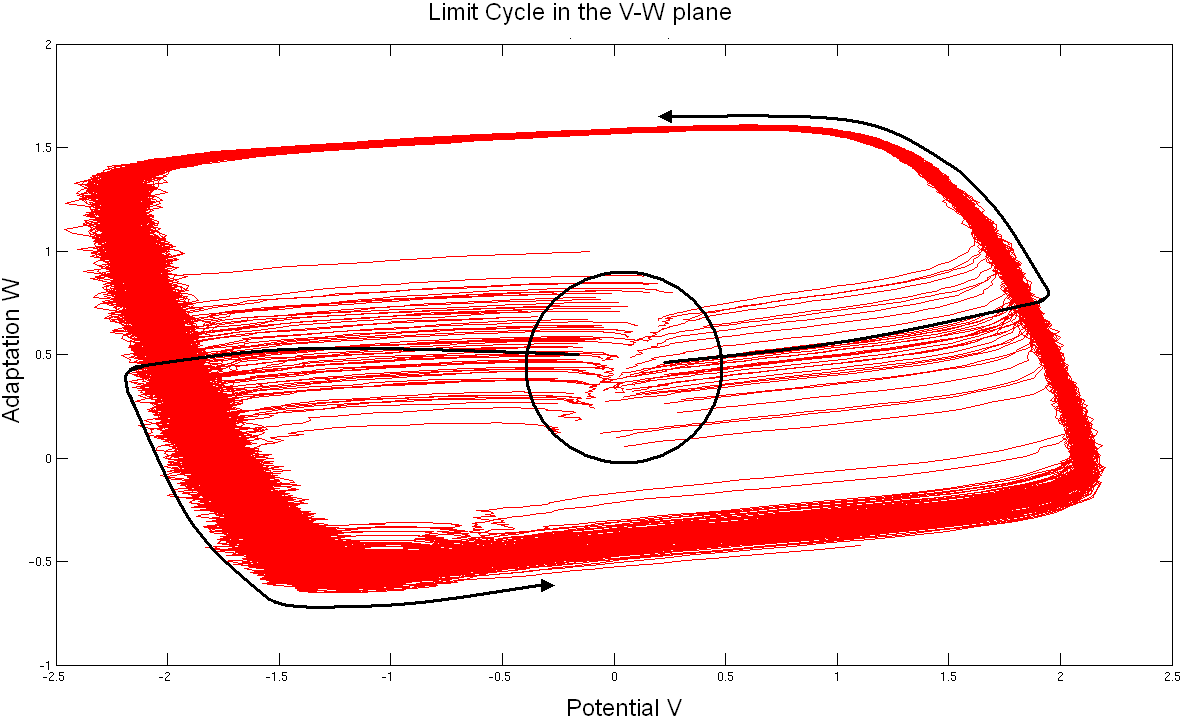}
\includegraphics[width=0.5\textwidth]{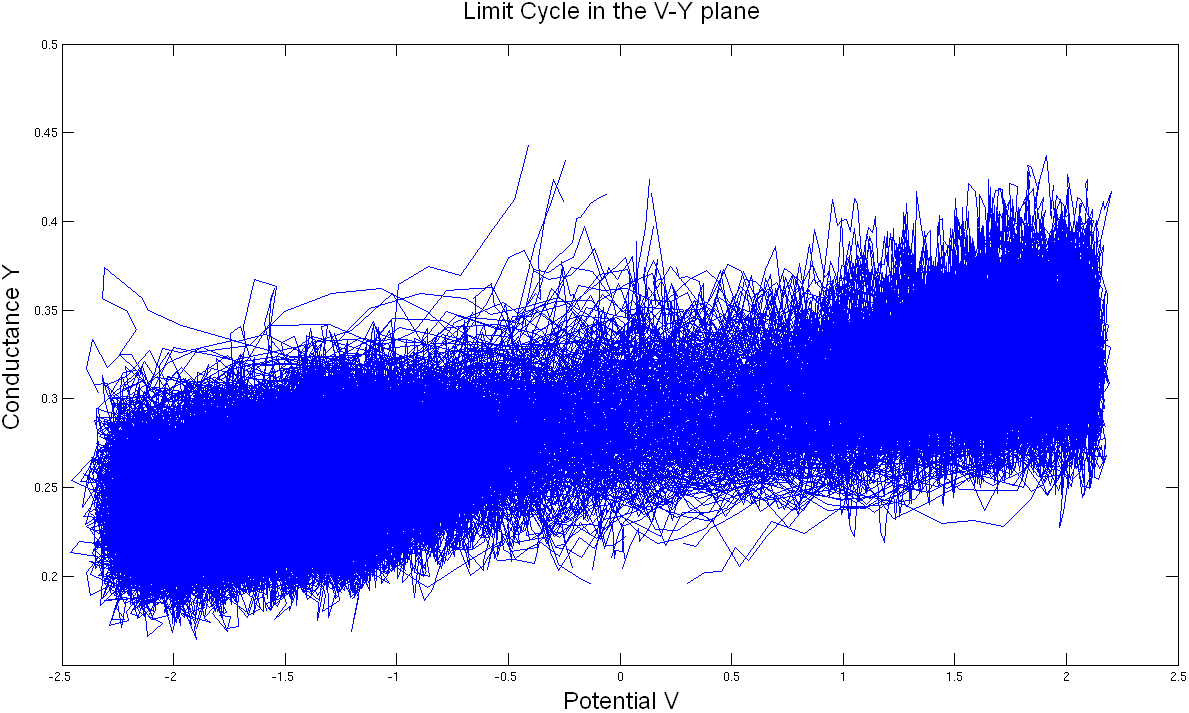}
}

\centerline{
\includegraphics[width=0.5\textwidth]{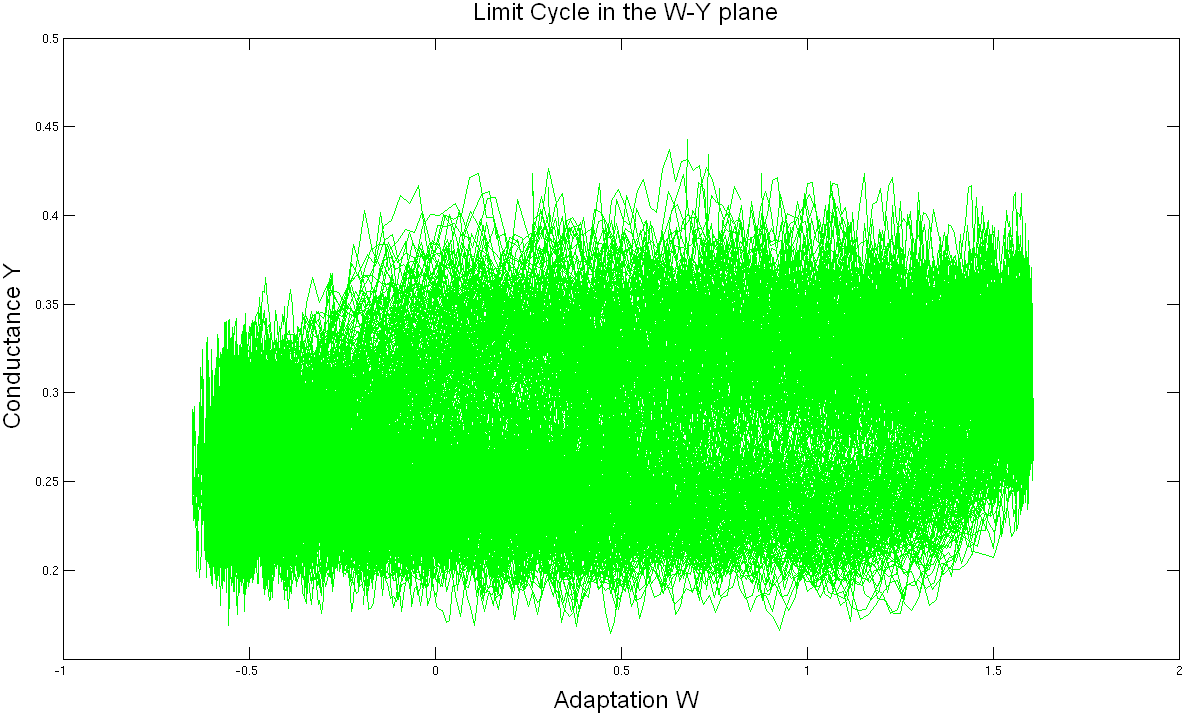}
}

\caption{Projection of $100$ trajectories in the $(V,w)$ (top left), $(V,y)$ (top right) and $(w,y)$ (bottom) planes. The limit cycle is  especially visible in the $(V,w)$ projection (red curves). The initial conditions split the trajectories into two classes corresponding to the two peaks shown in figure~\ref{fig:FNcycle}. Parameters similar to those of figure~\ref{fig:FNcycle}.}
\label{fig:PhaseSpaceCycle}
\end{figure}


\subsection{Numerical simulations with GPUs}\label{subsection:gpus}

Unfortunately the algorithm for solving the McKean-Vlasov-Fokker-Planck equation described in the previous section is computationally very expensive.
In fact, when the number of points in the discretized grid of the $(V,w,y)$ phase space is big, i.e. when the discretization steps $\Delta V$, $\Delta w$ and $\Delta y$ are small, we also need to keep $\Delta t$ small enough in order to guarantee the stability of the algorithm.
This implies that the number of equations that must be solved has to be large and moreover that they must be solved with a small time step if we want to reduce the numerical error.
This will  inevitably slow down the simulations. We have dealt with this problem by using a more powerful hardware, Graphical Processing Units (GPUs). GPUs are low cost highly parallel devices which are currently not only used for image or video processing but also for numerical simulations in physics or other scientific areas. They are specially designed for running a huge amount of low cost threads featuring a high speed context change between them. All these capabilities make them suitable for solving the McKean-Vlasov-Fokker-Planck equations.

We have improved the Runge-Kutta scheme of order 2 used for the simulations shown in section~\ref{section:numerics}, and adopted a Runge-Kutta scheme of order 4.

One of the purposes of the numerical study is to get a feeling for how the different parameters, in particular those related to the sources of noise, influence the solutions of the McKean-Vlasov-Fokker-Planck equation. This is meant to prepare the ground for the study of the bifurcation of these solutions with respect to these parameters, as was done in~\cite{touboul-hermann-faugeras:11} in a different context. For this preliminary study we varied the input current $I$ and the parameter $\sigma_{\rm ext}$ controlling the intensity of the noise on the membrane potential in equations~\eqref{eq:FNNetwork}.
The McKean-Vlasov-Fokker-Planck equation writes in this case\footnote{We have included a small noise (controlled by the parameter $\sigma_w$) on the adaptation variable $w$. This does not change the previous analysis, in particular proposition~\ref{prop:FN}, but makes the McKean-Vlasov-Fokker-Planck equation well-posed in a cube of the state space with 0 boundary value, see e.g., \cite{evans:98}.}:
\begin{multline}\label{eq:McKeanVlasofFPforGPUs}
        \frac{\partial}{\partial t}p(t,V,w,y)= \\
	-\frac{\partial}{\partial V} \Big\{\Big[V-\frac{V^3}{3}-W+I+\bar{J}(V-V_{rev})\int_{\mathbb{R}^3}y'p(t,V',w',y')dV'dw'dy'\Big]p(t,V,w,y)\Big\}+ \\
	-\frac{\partial}{\partial w} [c(V+a-bW)p(t,V,w,y)]-\frac{\partial}{\partial y} \{[a_{r} S(V)(1-y)-a_{d} y]p(t,V,w,y)\}+ \\
        +\frac{1}{2}\frac{\partial^2}{\partial V^2}\Big\{\Big[\sigma_{ext}^{2}+\sigma_{J}^{2}(V-V_{rev})^2\int_{\mathbb{R}^3}y'p(t,V',w',y')dV'dw'dy'\Big]p(t,V,w,y)\Big\}+ \\
        +\frac{1}{2}\sigma_{w}^{2}\frac{\partial^2}{\partial w^2}p(t,V,w,y)+\frac{1}{2}\frac{\partial^2}{\partial y^2}\{[a_{r} S(V)(1-y)+a_{d} y]\chi^2(y)p(t,V,w,y)\}
\end{multline}


\begin{figure}[htbp]
\centerline{
\includegraphics[width=0.5\textwidth]{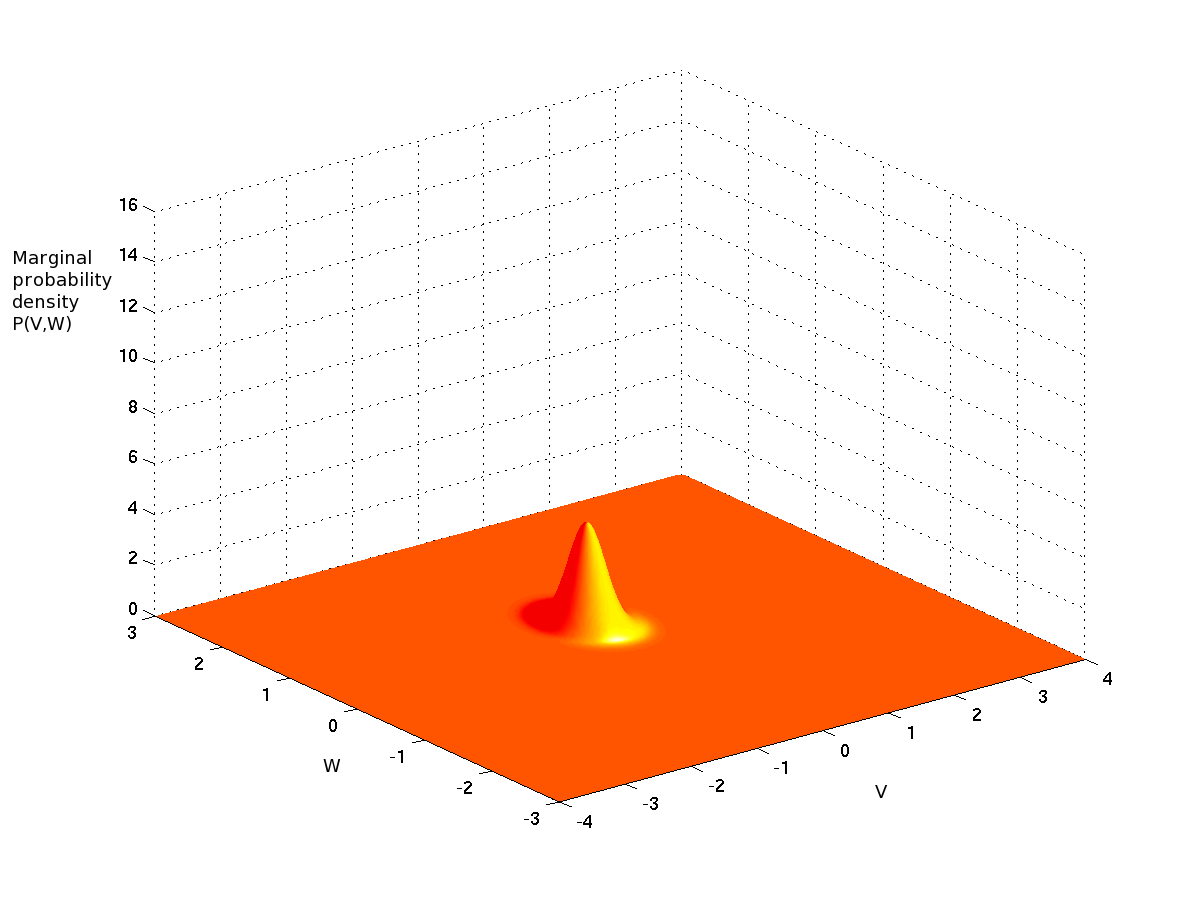}
\includegraphics[width=0.5\textwidth]{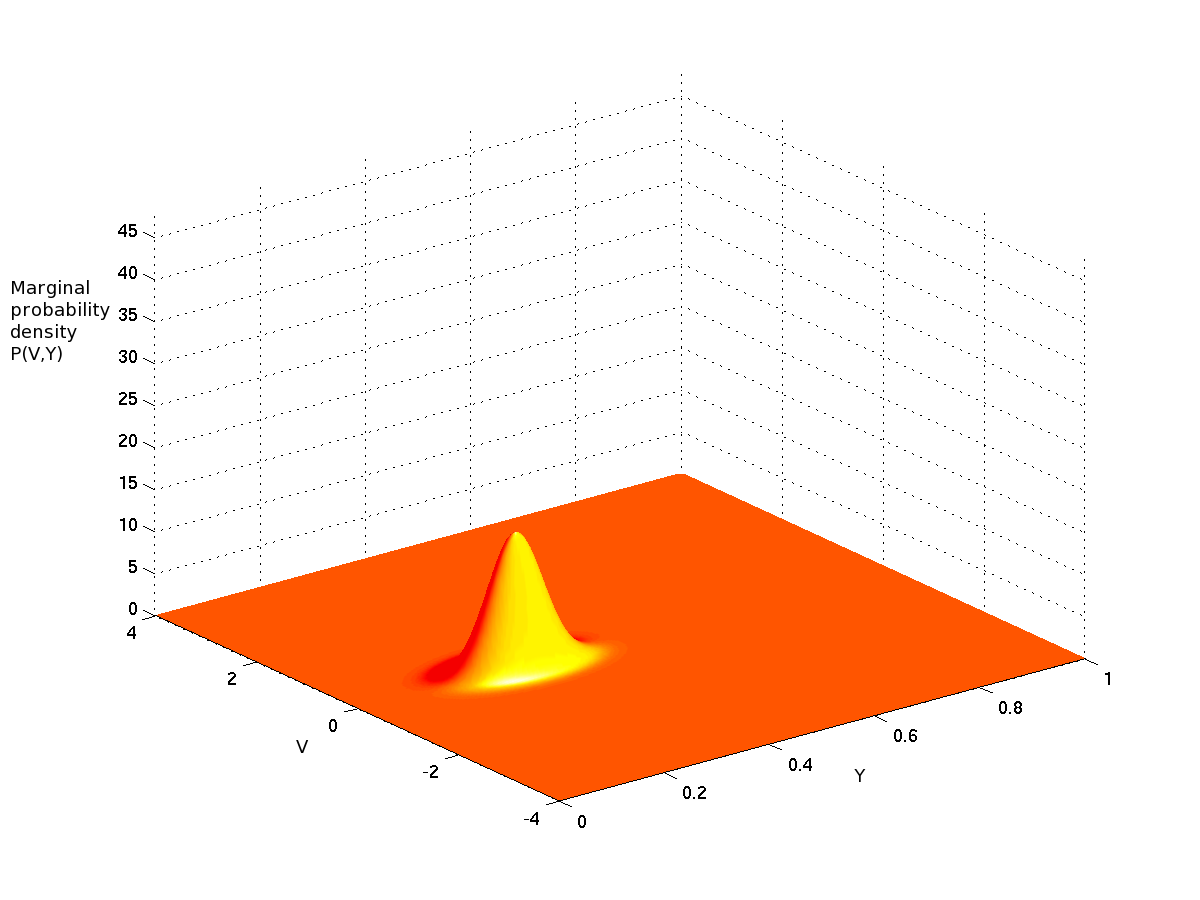}
}
\centerline{
\includegraphics[width=0.5\textwidth]{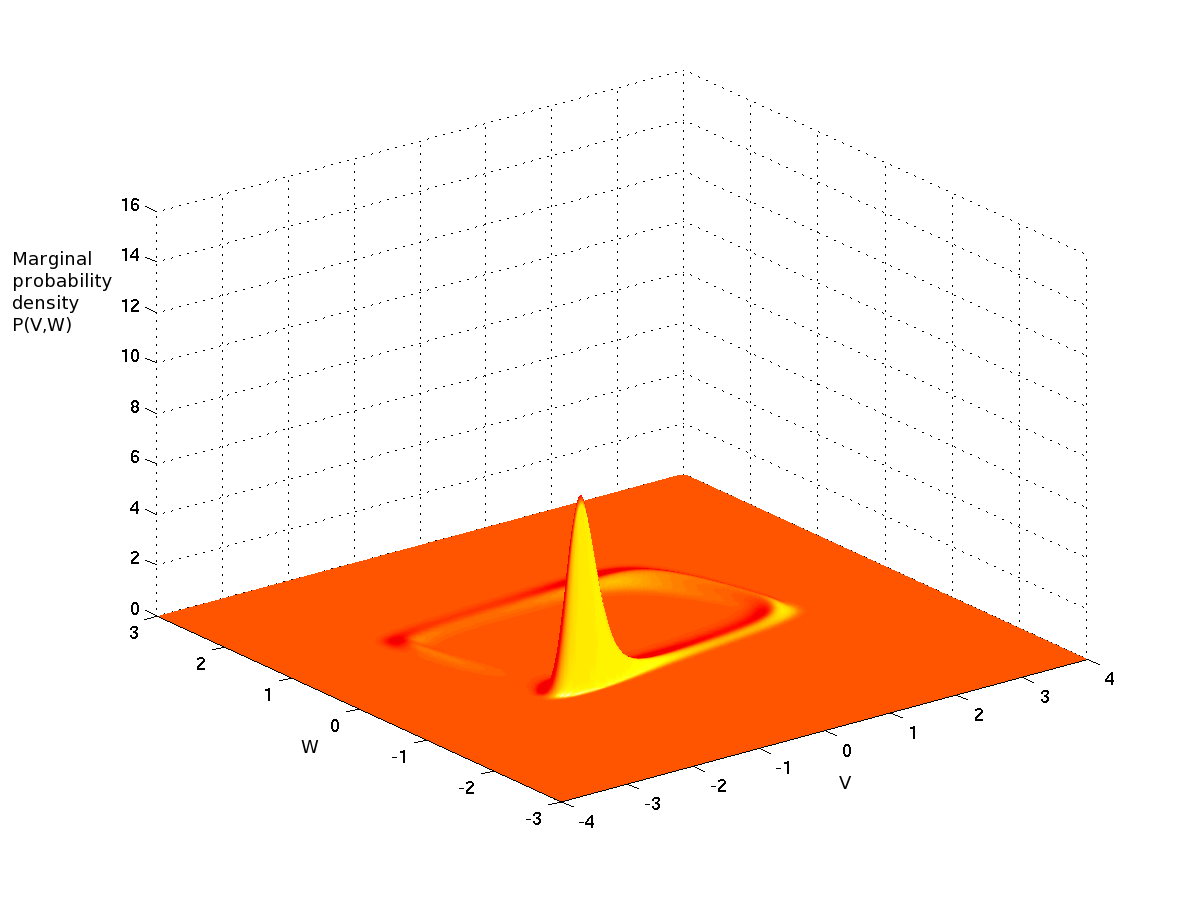}
\includegraphics[width=0.5\textwidth]{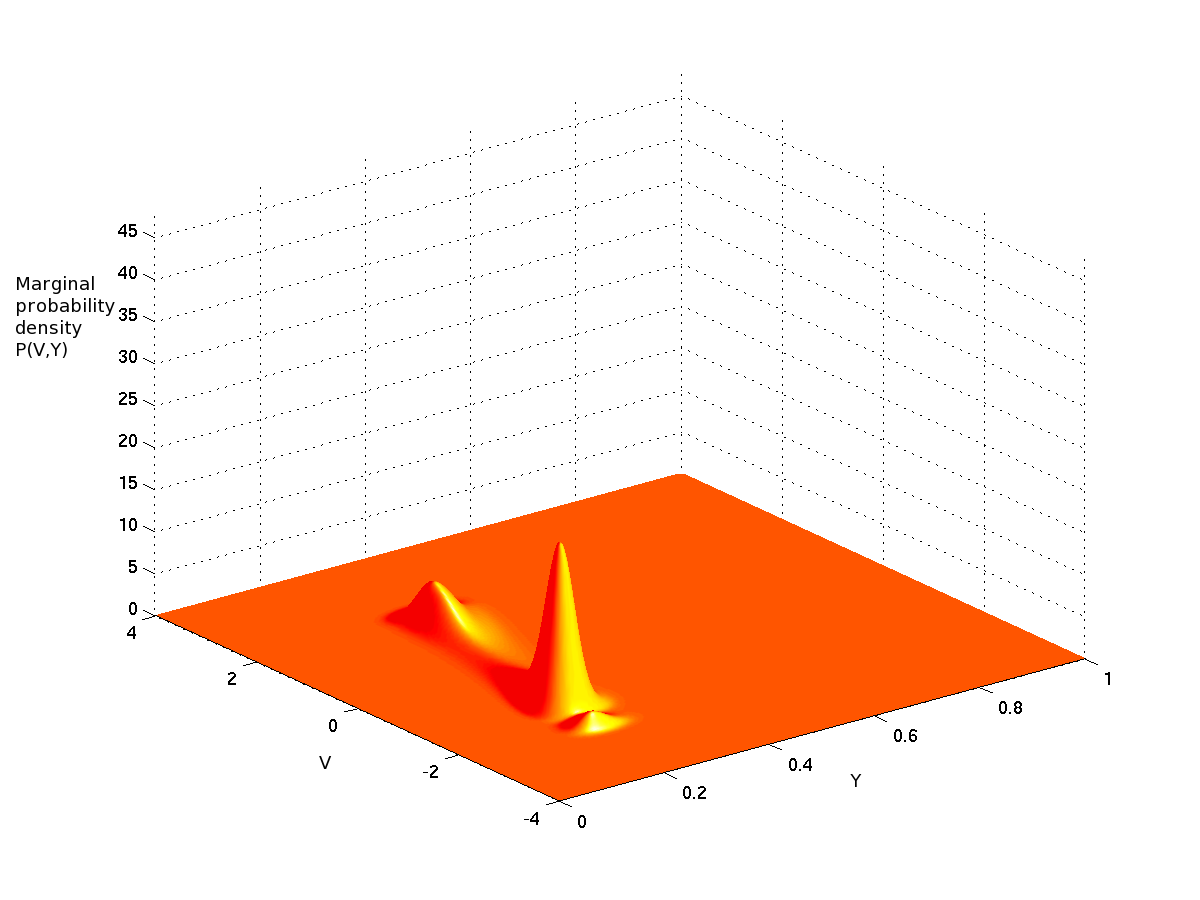}
}
\centerline{
\includegraphics[width=0.5\textwidth]{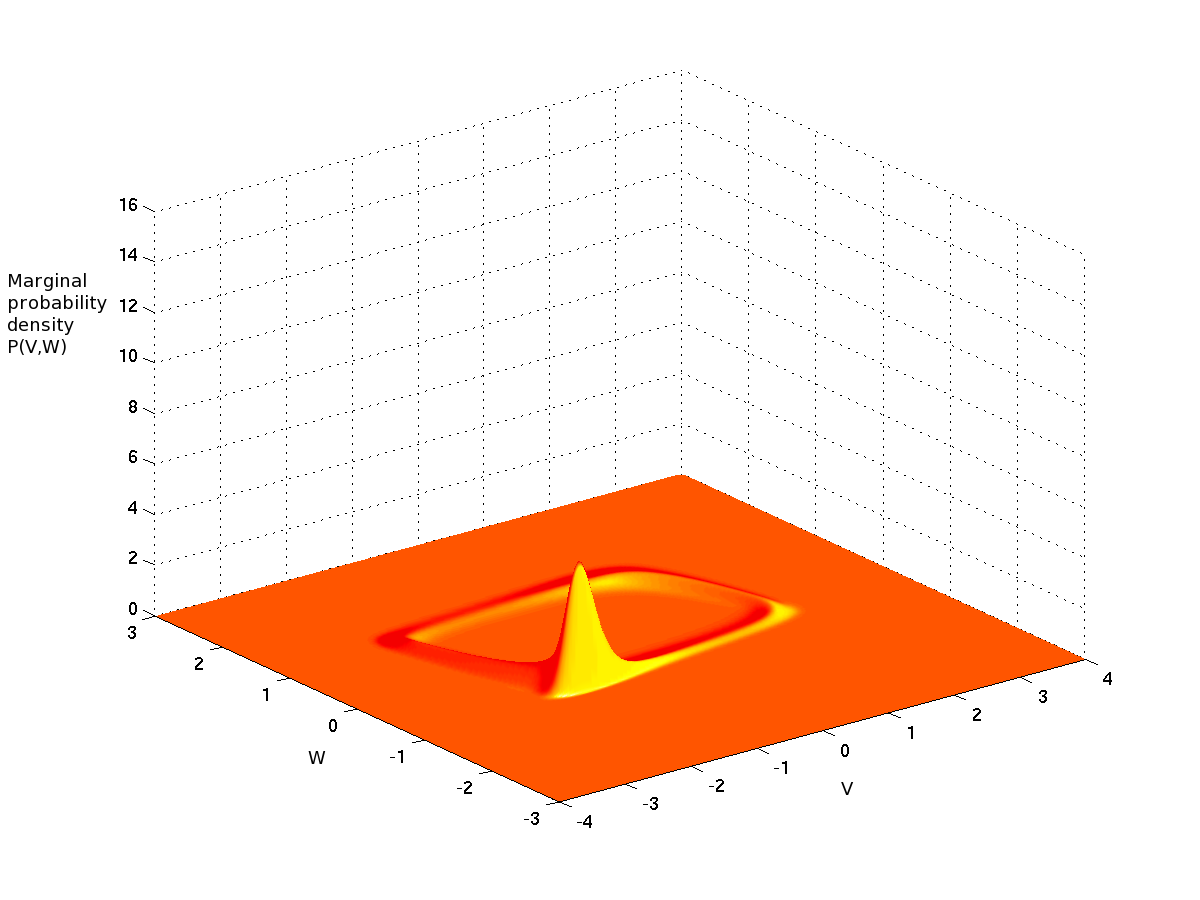}
\includegraphics[width=0.5\textwidth]{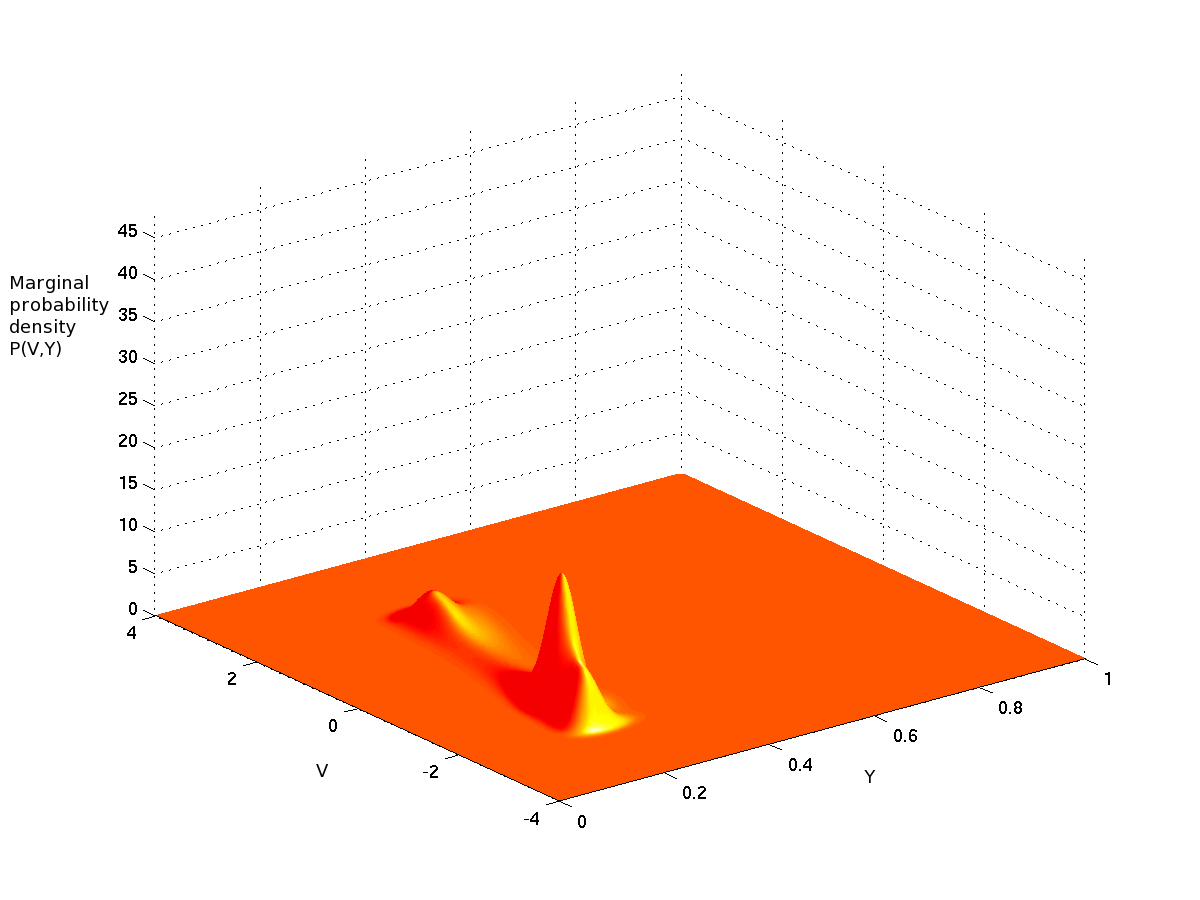}
}
\caption{Marginals with respect to the $V$ and $w$  variables (Left) and to the $V$ and $y$ variables (Right) of the solution of the McKean-Vlasov-Fokker-Planck equation. The first row shows the initial conditions, the second the marginals at time $0.05$ and the third the stationary (large time) solutions. The input current $I$ is equal to $0.4$ and $\sigma_{\rm ext}=0.27$.
These are screenshots at different times of movies available as supplementary material. The corresponding movies are in the files Vw-n0271-in04.avi and Vy-n0271-in04.avi.}
\label{fig:GPUs-LimitCycle1}
\end{figure}

\begin{figure}[htbp]
\centerline{
\includegraphics[width=0.5\textwidth]{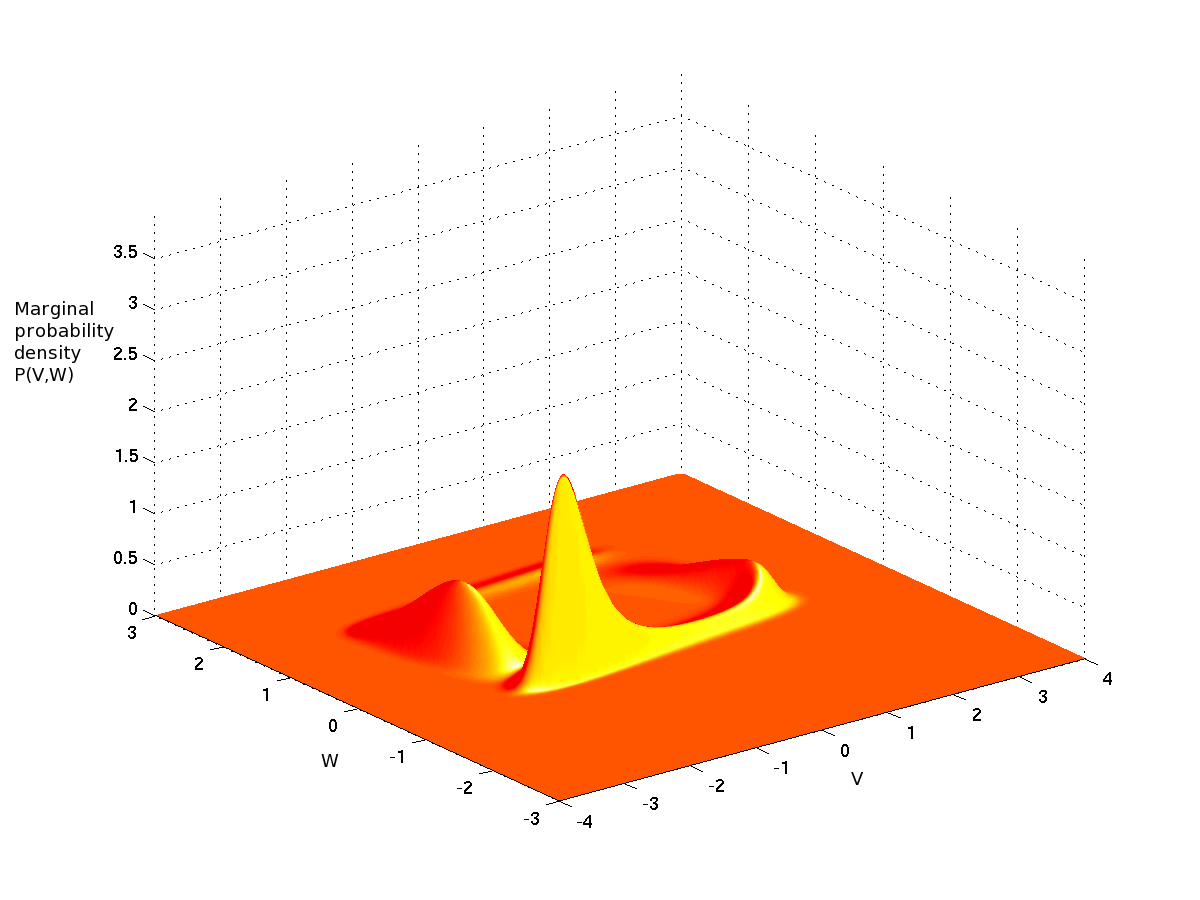}
\includegraphics[width=0.5\textwidth]{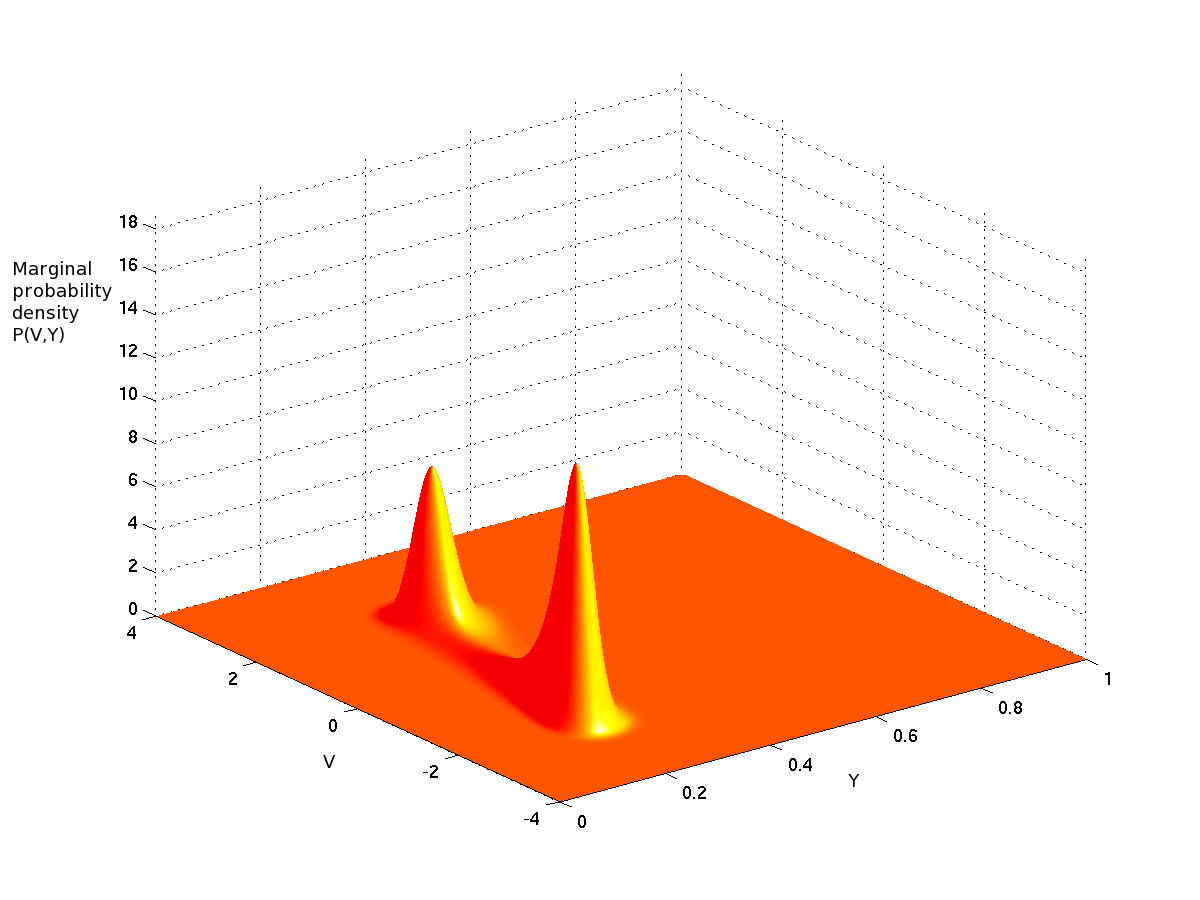}
}
\centerline{
\includegraphics[width=0.5\textwidth]{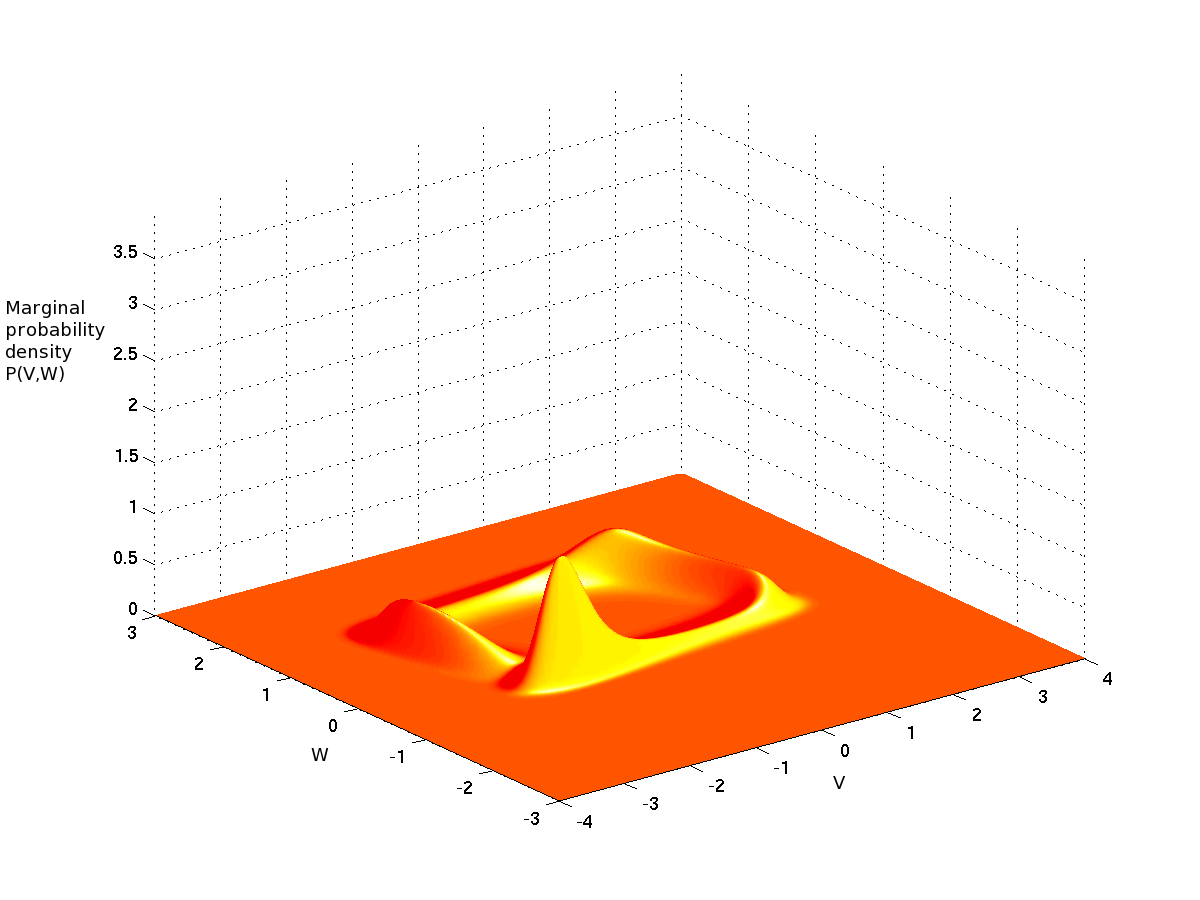}
\includegraphics[width=0.5\textwidth]{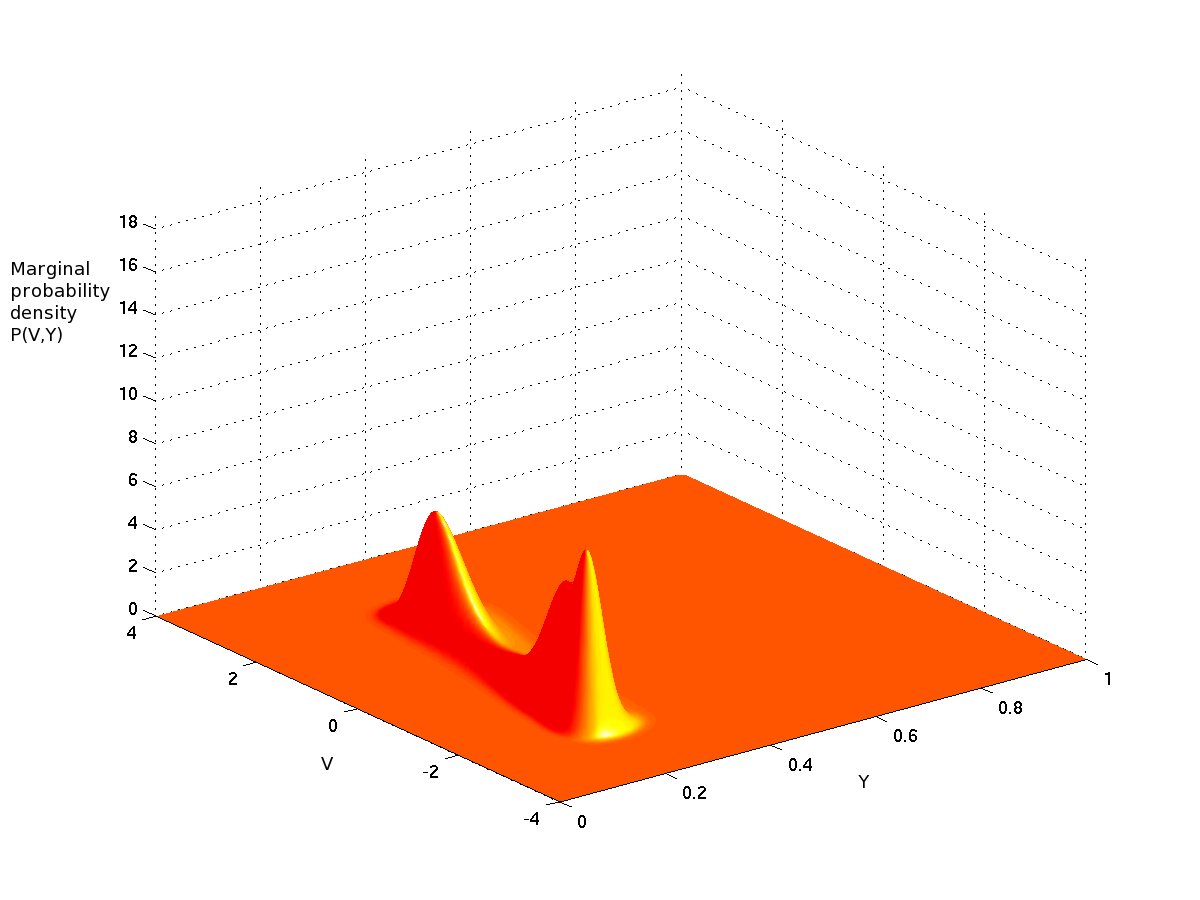}
}
\centerline{
\includegraphics[width=0.5\textwidth]{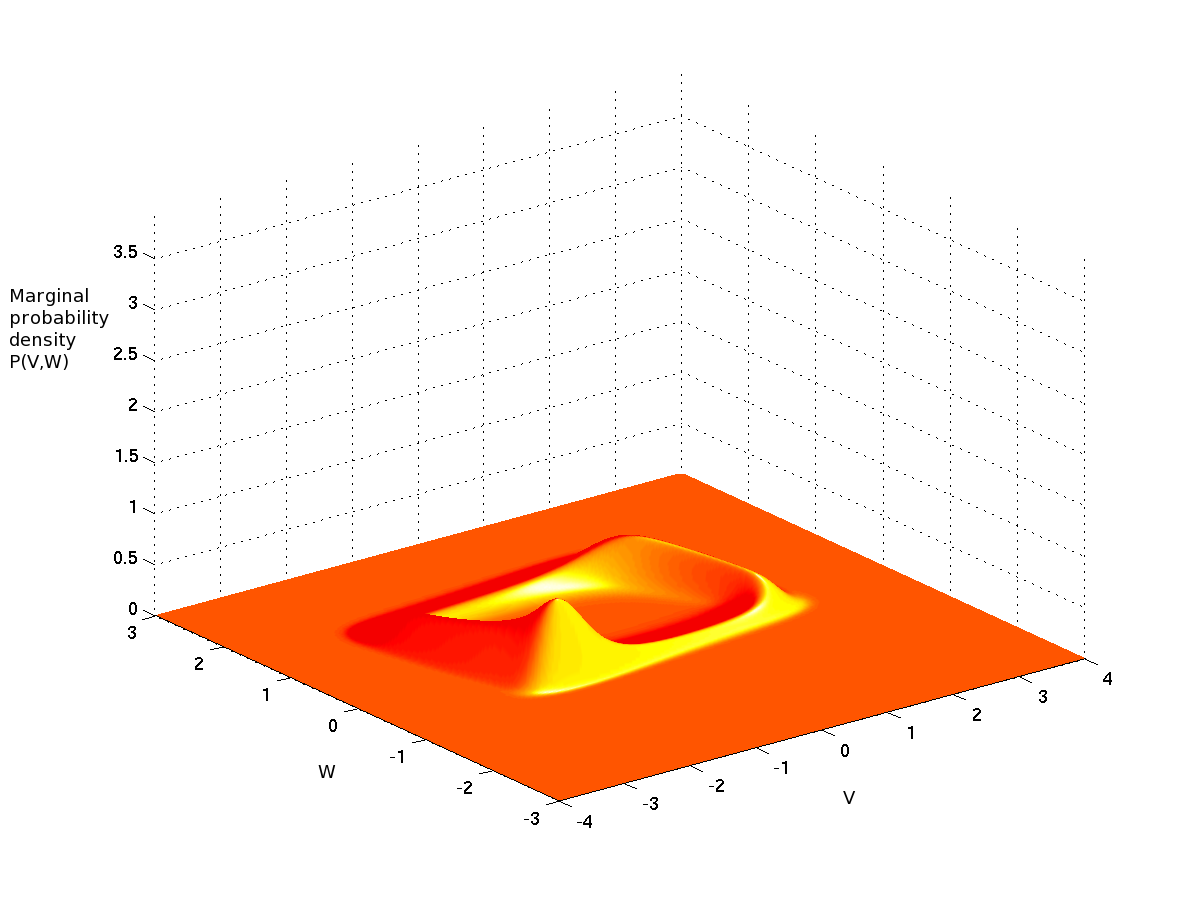}
\includegraphics[width=0.5\textwidth]{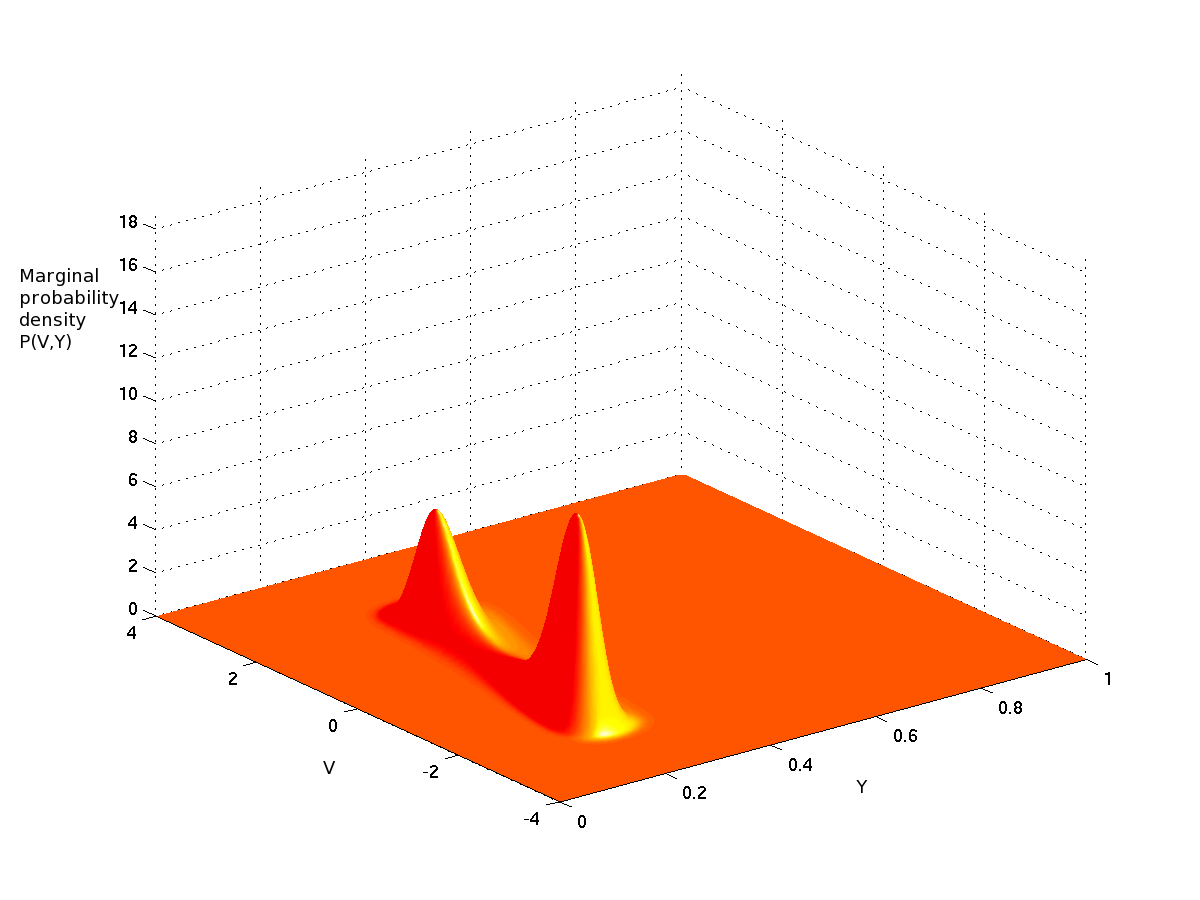}
}
\caption{Marginals with respect to the $V$ and $w$  variables (Left) and to the $V$ and $y$ variables (Right) of the solution of the McKean-Vlasov-Fokker-Planck equation. The first row shows the marginals at time $0.05$, the second the marginals at time $0.10$ and the third the stationary (large time) solutions. The input current $I$ is equal to $0.7$ and $\sigma_{\rm ext}=0.45$.
These are screenshots at different times of movies available as supplementary material. The corresponding movies are in the files Vw-n0451-in07.avi and Vy-n0451-in07.avi.}
\label{fig:GPUs-LimitCycle2}
\end{figure}

The simulations were run with the $\chi$ function~\eqref{eq:ChiFunction}, the initial condition described by~\eqref{eq:InitialConditionPdf} and with the parameters shown in table~\ref{tab:GPUs-Parameters}. 
Three snapshots of the solution are shown in figures~\ref{fig:GPUs-LimitCycle1} (corresponding to the values $I=0.4,\,\sigma_{\rm ext}=0.27$ of the external input current and of the standard deviation of the noise on the membrane potential) and~\ref{fig:GPUs-LimitCycle2} (corresponding to the values $I=0.7,\,\sigma_{\rm ext}=0.45$). In the figures the left column corresponds to the values of the marginal $p(t,V,w)$, the right column to the values of the marginal $p(t.V,y)$. Both are necessary to get an idea of the shape of the full distribution $p(t,V,w,y)$.  The first row of figure~\ref{fig:GPUs-LimitCycle1}  shows the initial conditions. They are the same for the results shown in figure~\ref{fig:GPUs-LimitCycle2}.  The second and third rows of figure~\ref{fig:GPUs-LimitCycle1} show the time instants $t=0.05$ and at convergence $t=0.33$ (the time units differ from those of the previous section but it is irrelevant to this discussion). The three rows of figure~\ref{fig:GPUs-LimitCycle2} show the time instants $t=0.05$, $t=0.10$ and $t=0.33$. In both cases the solution appears to converge to a stationary distributions whose mass is distributed over a ``blurred'' version of the limit cycle of the isolated neuron. The ``blurriness''  increases with the variance of the noise. The four movies for these two cases are available as supplementary material. 
\begin{center}
\begin{table}[htbp]
\begin{tabular}{| c | c | c | c | }
\hline
Initial Condition &                                     Phase space               & Stochastic & Synaptic \\
                  &                                                &FN neuron   &   Weights \\
\hline                             &          $V_{min}=-4$    & $a=0.7$             & $\overline{J}=1$ \\
                                             &                          &                     &    \\
 $\Delta t=0.0025,0.0012$,     &           $V_{max}=4$     & $b=0.8$             &  $\sigma_{J}=0.01$ \\
$\overline{V}_{0}=0.0$                    &         $\Delta V=0.027$  & $c=0.08$           &                \\
$\sigma_{V_{0}}=0.2$                         &          $w_{min}=-3$    & $I=0.4,\,0.7$       &  \\
$\overline{w}_{0}=-0.5$              &          $w_{max}=3$     & $\sigma_{ext}=0.27,0.45,$   &          \\
$\sigma_{w_{0}}=0.2$                         &          $\Delta w=0.02$  & $\sigma_{w}=0.0007$   &   \\
$\overline{y}_{0}=0.3$                     &      $y_{min}=0$     &            &   \\
$\sigma_{y_{0}}=0.05$                        &          $y_{max}=1$     &                 &  \\
                                             &        $\Delta y=0.003$ &               &         \\
\hline
\end{tabular}
\caption{Parameters used in the simulations of the McKean-Vlasov-Fokker-Planck equation with GPUs \diego{shown in figures~\ref{fig:GPUs-LimitCycle1} and  ~\ref{fig:GPUs-LimitCycle2}, and in the supplementary material.}}
\label{tab:GPUs-Parameters}
\end{table}
\end{center}

\section{Discussion and conclusion}\label{section:Discussion}

\john{In this article, we addressed the problem of the limit in law of networks of biologically-inspired neurons as the number of neuron tends to infinity. We emphasized the necessity of dealing with biologically-inspired models and discussed at length the type of models relevant to this study. We chose to address the case conductance-based network models that are relevant description of the neuronal activity. Mathematical results on the analysis of these diffusion processes in interaction resulted to } replace a set of \olivier{$NP$} \olivier{$d$}-dimensional coupled equations (the network equations) in the limit of large $N$s by \olivier{$P$ $d$-dimensional} mean-field equations describing the global behavior of the network. \john{However, the} price to pay \john{for this reduction was the fact that} the \john{resulting} \olivier{mean-field} equations are non-standard stochastic differential equations\john{, similar to McKean-Vlasov equations. These can} either be expressed as implicit equations on the law of the solution, or, in terms of probability density function through the McKean-Vlasov-Fokker-Planck equations as a nonlinear, nonlocal partial differential equation. These equations are in general hard to study theoretically.

\olivier{Beside the fact that we explicitely model real spiking neurons, the mathematical part of our work differs from that of previous authors such as McKean, Tanaka and Sznitman (see section \ref{section:intro}) because we are considering several populations with the effect that the analysis is significantly more complicated. Our hypotheses are also more general,  e.g., the drift and diffusion functions are non-trivial and satisfy the general condition \ref{Assump:MonotoneGrowth} which is more general than 
the usual linear growth condition. Also they are only assumed locally (and not globally ) Lipschitz continuous to be able to deal, for example, with the Fitzhugh-Nagumo model. \johnNew{A} locally Lipschitz continuous case was recently addressed in a different context for a model of swarming in~\cite{bolley-canizo-etal:10}. }
Proofs of our results for somewhat stronger hypotheses than ours and in special cases are scattered in  the literature. Our main contribution is that we provide a complete self-sufficient proof in a fairly general case by gathering all the ingredients that are  required for our neuroscience applications. In particular, the case of the Fitzhugh-Nagumo model where the drift function does not satisfy the linear growth condition involves a generalization of previous works by using a Has'minskii-like argument based on the presence of a Lyapunov function to prove existence of solutions and non-explosion.

\john{The simulation of these equations can itself be very costly. We hence addressed in section~\ref{section:numerics} numerical methods to compute the solutions of these equations, in the probabilistic framework using the convergence result and standard integration methods of differential equations, or in the Fokker-Planck framework}.
The simulations performed for different values of the external input current parameter and one of the parameters controlling the noise allowed us to show that the spatio-temporal shape of the probability density function describing the solution of the McKean-Vlasov-Fokker-Planck equation  was sensitive to the variations of these parameters, as shown in figures ~\ref{fig:GPUs-LimitCycle1} and~\ref{fig:GPUs-LimitCycle2}. However, we did not address the full characterization of the dynamics of the solutions in the present article. This appears to be a complex question that will be the subject of future work. It is known that for different McKean-Vlasov equations stationary solutions of these equations do not necessarily exist and, when they do, are not necessarily unique (see~\cite{herrmann-tugaut:10}). A very particular case of these equations was treated in~\cite{touboul-hermann-faugeras:11} where the authors consider that the function $f_{\alpha}$ is linear, $g_{\alpha}$ constant and $b_{\alpha\beta}(x,y)=S_{\beta}(y)$. This model, known as the firing-rate model, is shown in that paper to have Gaussian solutions when the initial data is Gaussian, and the dynamics of the solutions can be exactly reduced to a set of $2P$ coupled ordinary differential equations governing the mean and the standard deviation of the solution. Under these assumptions a complete study of the solutions is possible, and the dependence upon the parameters can be understood through bifurcation analysis. The authors show that intrinsic noise levels govern the dynamics, creating or destroying fixed points and periodic orbits. 


\olivier{The mean-field description has also deep theoretical implications in neuroscience. Indeed, it points towards the fact that neurons encode their responses to stimuli through probability distributions. This type of coding was evoked by several authors~\cite{rolls-deco:10}, and the mean-field approach shows that under some mild conditions this phenomenon arises: all neurons belonging to a particular population can be seen as independent realizations of the same process, governed by the mean-field equation. The relevance of this phenomenon is reinforced by the fact that it has recently been observed experimentally that neurons had correlation levels significantly below what had been previously reported \cite{ecker-berens-etal:10}. This independence has deep implications on the efficiency of neural coding which the propagation of chaos theory accounts for. To illustrate this phenomenon we have performed the following simulations. Considering a network of 2, 10 and 100 Fitzhugh-Nagumo neurons, we have simulated 2000 times the network equations over some
time interval $[0,100]$. We have picked at random a pair of neurons and computed the time variation of the cross-correlation of the values of their state variables. The results are shown in figure~\ref{fig:decorr}. It appears that the propagation of chaos is observable for relatively small values of the number of neurons in the network, thus indicating once more that the theory developed in this paper in the limit case of an infinite number of neurons is quite robust to finite size effects\footnote{\diegoNew{Note that we did not estimate the correlation within larger networks since, as predicted by theorem~\ref{thm:Principal2}, it will be smaller and smaller requiring an increasingly large number of Monte Carlo simulations.}}
}
\begin{figure}[htbp]
\centerline{
\includegraphics[width=0.5\textwidth]{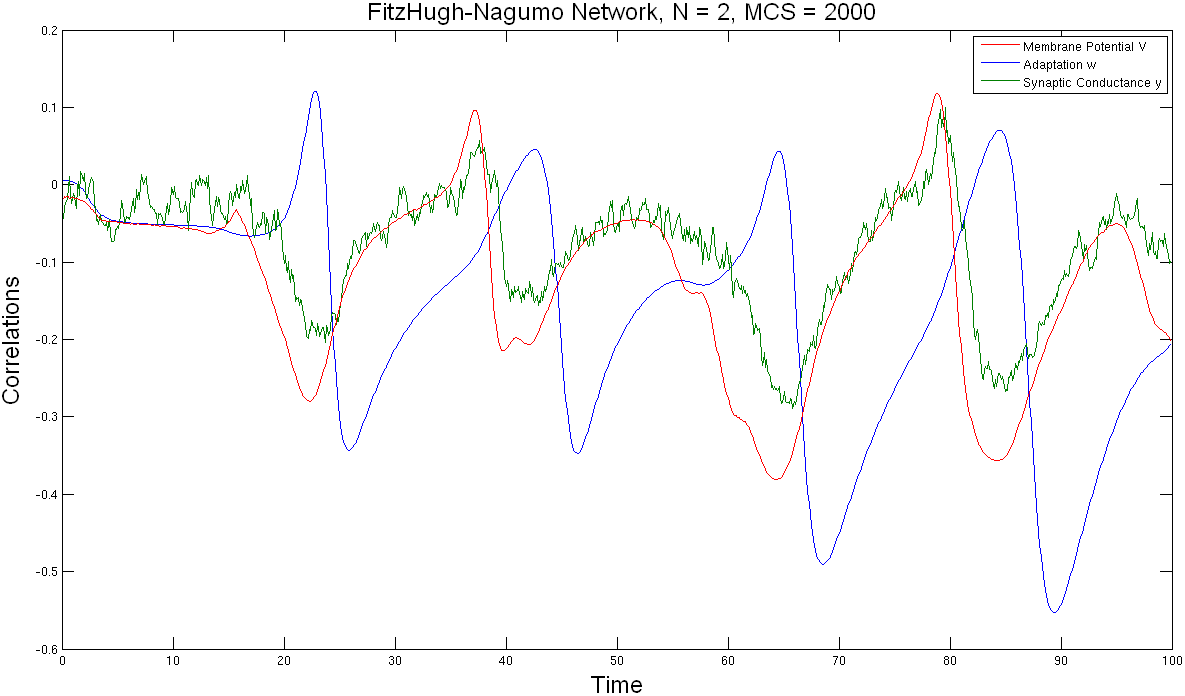}\hspace{0.1cm}\includegraphics[width=0.5\textwidth]{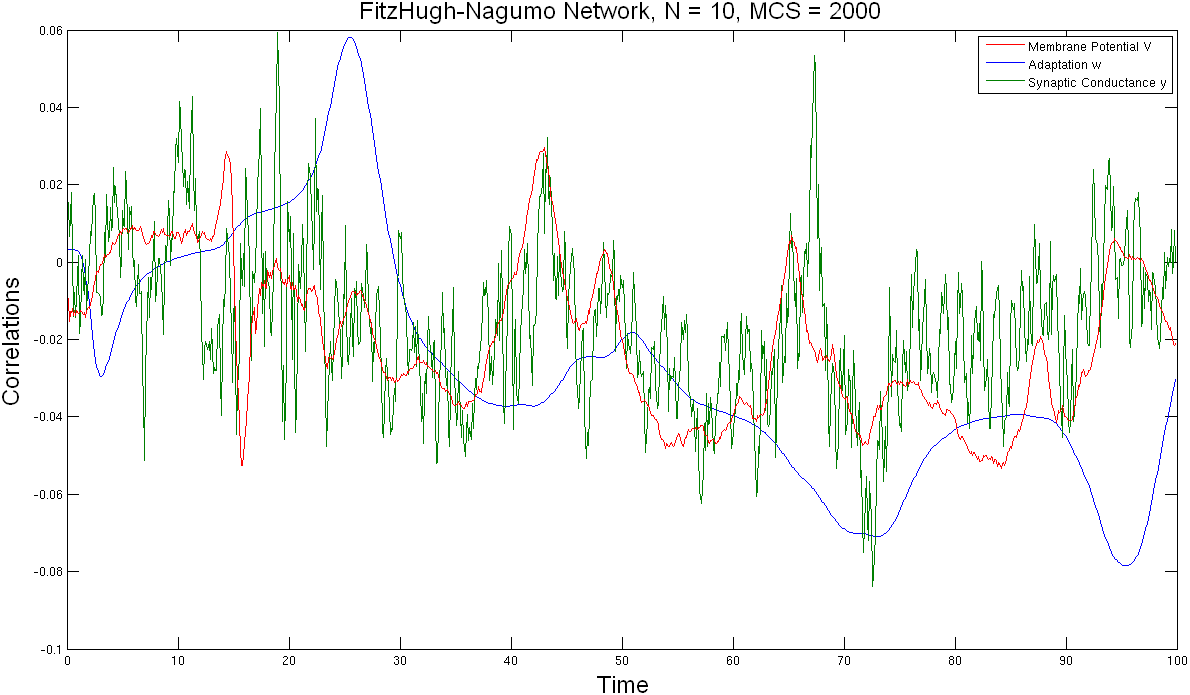}
}
\centerline{
\includegraphics[width=0.5\textwidth]{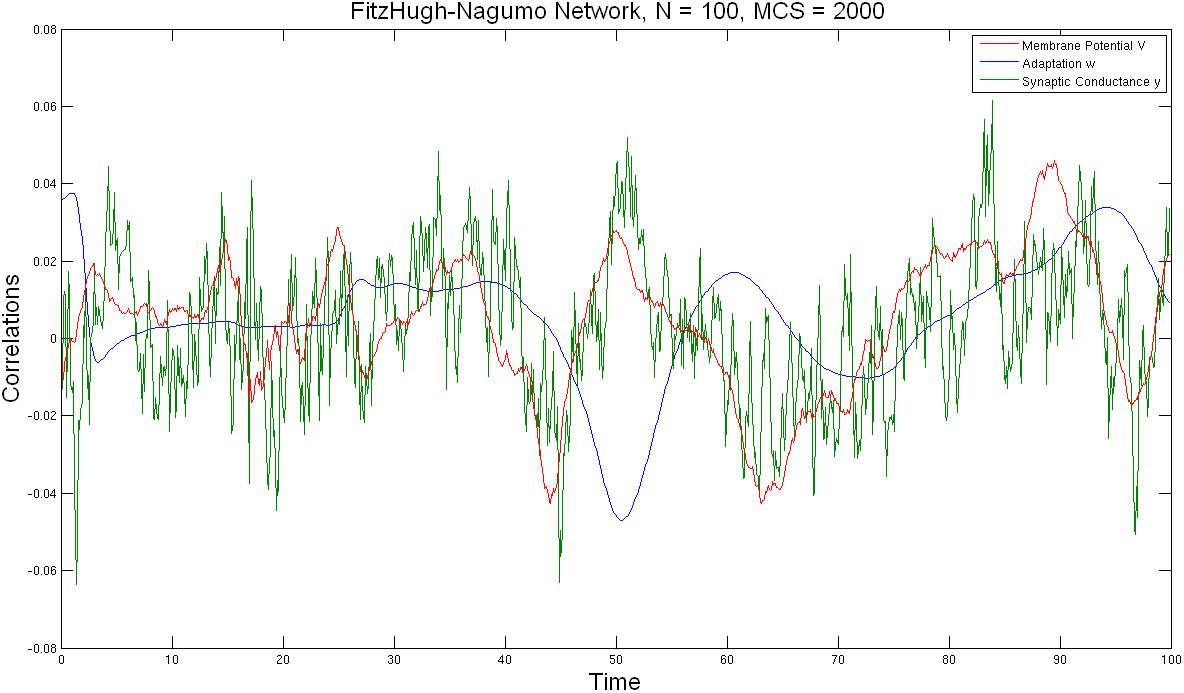}
}
\caption{Variations over time of the cross-correlation of the $(V,w,y)$ variables of two Fitzhugh-Nagumo neurons in a network. Top left: 2 neurons. Top right: 10 neurons. Bottom: 100 neurons. The cross-correlation decreases steadily with the number of neurons in the network.}
\label{fig:decorr}
\end{figure}

The present study develops theoretical arguments to derive the mean-field equations resulting from the activity of large neurons ensembles. However, the rigorous and formal approach developed here does not allow direct characterization of brain states. The paper however opens the way to rigorous analysis of the dynamics of large neuron ensembles through derivations of different quantities that may be relevant. A first approach could be to derive the equations of the successive moments of the solutions. Truncating this expansion would yield systems of ordinary differential equations that can give approximate information on the solution. However, the choice of the number of moments taken into account is still an open question that can raise several deep questions~\cite{ly-tranchina:07}. 

\noindent
{\bf Acknowledgements}\\
This work was partially supported by the ERC grant \#227747 NerVi and by the FACETS-ITN Marie-Curie Initial Training Network.
\appendix
\section{Proof of theorem \ref{thm:Principal2}}\label{section:proofs}



\johnNew{In this appendix we prove the convergence of the network equations towards the mean-field equations~\eqref{eq:MFE} and of the propagation of chaos property. The proof follows standard proofs in the domain as generally done in particular by Tanaka or Sznitman \cite{tanaka:78,sznitman:84a}, adapted to our particular case where we consider a non-zero drift function and a time and space dependent diffusion function. It is based on the very powerful coupling argument, that identifies the almost sure limit of the process \olivier{$X^i$} as the number of neuron tends to infinity, as popularized by Sznitman in~\cite{sznitman:89}, but whose idea dates back from the 70s (for instance Dobrushin uses it in~\cite{dobrushin:70}). This process is exactly the solution of the mean-field equation driven by the same Brownian motion as \olivier{$X^i$} and with the same initial condition random variable. In our case, this leads us to introduce the sequence of independent stochastic processes $(\bar{X}^i_t)_{i=1\ldots N}$ having the same law as $\bar{X}^{\olivier{\alpha}}$, $\alpha=p(i)$, solution of the mean-field equation:}
\begin{multline}\label{eq:Coupling}
	d\bar{X}^i_t = f_{\alpha}(t,\bar{X}^i_t)\,dt + \sum_{\gamma=1}^P\Exp_Z[b_{\alpha\gamma}(\bar{X}^i_t,Z_t^{\gamma})]\,dt + g_{\alpha}(t,\bar{X}^{i}_t)\, dW^i_t \\
	+ \sum_{\gamma=1}^P\Exp_Z[\beta_{\alpha\gamma}(\bar{X}^i_t,Z_t^{\gamma})]\,dB^i_t,
\end{multline}
\johnNew{with initial condition $\bar{X}^i_0=X^i_0$, the initial condition of the neuron $i$ in the network, which were assumed to be independent and identically distributed. $(W^i_t)$ and $(B^i_t)$ are the Brownian motions involved in the network equation \eqref{eq:Network}. \olivier{As previously $Z=(Z^1,\cdots,Z^P)$ is a process independent of $\bar{X}$ that has the same law}}. \johnNew{Denoting as previously by $\olivier{m}_t^{\alpha}$ the probability distribution of $\bar{X}_t^{\alpha}$ solution of the mean-field equation \eqref{eq:MFE}, the law of the collection of processes $(\bar{X}^{i_k}_t)$ for some fixed $k\in \N^*$, namely ${m}^{p(i_1)} \otimes\ldots\otimes {m}^{p(i_k)}$, is shown to be the limit of the processes $(X^i_t)$ solution of the network equations \eqref{eq:Network} as $N$ goes to infinity}. 

We recall for completeness theorem~\ref{thm:Principal2}:
\setcounter{theorem}{3}
\begin{theorem}
Under the assumptions \ref{Assump:LocLipsch} to \ref{Assump:MonotoneGrowth}  
the following holds true:
	\begin{itemize}
		\item {\bf Convergence:} For each neuron $i$ of population $\alpha$, the law of the multidimensional process $X^{i,N}$ converges towards the law of the solution of the mean-field equation related to population $\alpha$, namely $\bar{X}^{\alpha}$. 
		\item {\bf Propagation of chaos:} For any $k\in \N^*$, and any $k$-uplet $(i_1, \ldots, i_k)$, the law of the process $(X^{i_1,N}_t, \ldots, X^{i_n,N}_t, t\leq T)$ converges towards \olivier{$m_t^{p(i_1)}\otimes\ldots\otimes m_t^{p(i_n)}$}, i.e. the asymptotic processes have the law of the solution of the mean-field equations and are all independent. 
	\end{itemize}
\end{theorem}
\begin{proof}
On our way we also prove that
\begin{equation}\label{eq:Propchaos}
	\max_{i=1\cdots N} N \; \Exp\Big [\sup_{s\leq T} \Vert X^{i,N}_s - \bar{X}^i_s\Vert^2 \Big]< \infty,
\end{equation}
which implies in particular convergence in law of the process $(\olivier{X^{i,N}_t}, t\leq T)$ towards $(\bar{X}^{\alpha}_t, t\leq T)$ solution of the mean-field equations \eqref{eq:MFE}. 

	The proof basically consists in thoroughly analyzing the difference between the two processes as $N$ tends to infinity. The difference is the sum of eight terms (we dropped the index $N$ for the sake of simplicity of notations) denoted $A_t$ through $H_t$:
\olivier{
	\begin{align}
		\nonumber X^i_t-\bar{X}^i_t&=\underbrace{\int_0^t f_{\alpha}(s,X^i_s)-f_{\alpha}(s,\bar{X}^i_s) \, ds}_{A_t} + \underbrace{\int_0^t g_{\alpha}(s,X^i_s)-g_{\alpha}(s,\bar{X}^i_s) \, dW^i_s}_{B_t}\\
		\nonumber &\qquad + \underbrace{\sum_{\gamma=1}^P \int_0^t \frac 1 {N_{\gamma}} \sum_{j=1}^{N_{\gamma}} b_{\alpha\gamma}(X^i_s,X^j_s)-b_{\alpha\gamma}(\bar{X}^i_s,X^j_s)\, ds}_{C_t}\\
		\nonumber &\qquad + \underbrace{\sum_{\gamma=1}^P \int_0^t \frac 1 {N_{\gamma}} \sum_{j=1}^{N_{\gamma}} b_{\alpha\gamma}(\bar{X}^i_s,X^j_s)-b_{\alpha\gamma}(\bar{X}^i_s,\bar{X}^j_s)\, ds}_{D_t}\\
		\nonumber &\qquad + \underbrace{\sum_{\gamma=1}^P\int_0^t \frac 1 {N_{\gamma}} \sum_{j=1}^{N_{\gamma}} b_{\alpha\gamma}(\bar{X}^i_s,\bar{X}^j_s)-\Exp_Z[b_{\alpha\gamma}(\bar{X}^i_s,Z_s^\gamma)]\, ds}_{E_t}\\
		\nonumber &\qquad + \underbrace{\sum_{\gamma=1}^P\int_0^t \frac 1 {N_{\gamma}} \sum_{j=1}^{N_{\gamma}}\beta_{\alpha\gamma}(X^i_s,X^j_s)-\beta_{\alpha\gamma}(\bar{X}^i_s,X^j_s)\, dB^{i\gamma}_s}_{F_t}\\
		\nonumber &\qquad + \underbrace{\sum_{\gamma=1}^P\int_0^t \frac 1 {N_{\gamma}} \sum_{j=1}^{N_{\gamma}}\beta_{\alpha\gamma}(\bar{X}^i_s,X^j_s)-\beta_{\alpha\gamma}(\bar{X}^i_s,\bar{X}^j_s)\, dB^{i\gamma}_s}_{G_t}\\
	\label{eq:EightTerms}		 &\qquad + \underbrace{\sum_{\gamma=1}^P\int_0^t \frac 1 {N_{\gamma}} \sum_{j=1}^{N_{\gamma}}\beta_{\alpha\gamma}(\bar{X}^i_s,\bar{X}^j_s)-\Exp_Z[\beta_{\alpha\gamma}(\bar{X}^i_s,Z_s^\gamma)\, dB^{i\gamma}_s}_{H_t}
	\end{align}
}
	It is important to note that the probability distribution of these terms does not depend on the neuron $i$. We are interested in the limit, as $N$ goes to infinity, of the quantity $\Exp [\sup_{s\leq T} \Vert X^{i,N}_s - \bar{X}^i_s\Vert^2 ]$. We decompose this expression into the sum of the eight terms involved in equation \eqref{eq:EightTerms} using H\"older's inequality and upperbound each term separately. The terms $A_t$ and $B_t$ are treated exactly as in the proof of theorem \ref{thm:Principal1}. We start by assuming that $f$ and $g$ are \olivier{uniformly} globally $K$ Lipschitz-continuous with respect to the second variable. The locally-Lipschitz case is treated in the same manner as done in the proof of theorem~\ref{thm:Principal1} by i) stopping the process at time $\tau_U$, ii) using the Lipschitz-continuity of $f$ and $g$ in the ball of radius $U$ and iii) by a truncation argument and using the almost sure boundedness of the solutions extending the convergence to the locally-Lipschitz case. 

\noindent	
\olivier{As seen previously} we have:
	\begin{align*}
		\Exp[\sup_{s\leq t} \Vert A_s \Vert^2] & \leq K^2\, T\, \int_0^t \Exp[\sup_{u\leq s} \Vert X_u^i-\bar{X}_u^i\Vert^2 ]\,\olivier{ds}\\
		\Exp[\sup_{s\leq t} \Vert B_s \Vert^2] & \leq 4\, K^2\, \int_0^t \Exp[\sup_{u\leq s} \Vert X_u^i-\bar{X}_u^i\Vert^2 ]\, \olivier{ds}
	\end{align*}
\olivier{Now for $C_t$}:
	\begin{align*}
          \|C_s\|^2 &= \left\|\sum_{\gamma=1}^P \int_0^s \frac 1 {N_{\gamma}} \sum_{j=1}^{N_{\gamma}} b_{\alpha\gamma}(X^i_u,X^j_u)-b_{\alpha\gamma}(\bar{X}^i_u,X^j_u)\, du \right\|^2\\
\text{(Cauchy-Schwarz)} &\leq TP \int_0^s \sum_{\gamma=1}^P  \frac 1 {N_{\gamma}} \sum_{j=1}^{N_{\gamma}}\left\|b_{\alpha\gamma}(X^i_u,X^j_u)-b_{\alpha\gamma}(\bar{X}^i_u,X^j_u)\right\|^2\,du\\
\text{(assumption~\ref{Assump:LocLipschb})} &\leq TPL^2 \int_0^s \left\| X^i_u-\bar{X}^i_u \right\|^2\,du
	\end{align*}
Therefore
\begin{align*}
\sup_{s \leq t} \|C_s\|^2 &\leq TPL^2 \int_0^t \left\| X^i_s-\bar{X}^i_s \right\|^2\,ds\\
\Exp\left[\sup_{s \leq t} \|C_s\|^2 \right] &\leq TPL^2 \int_0^t \Exp\left[\sup_{u \leq s}  \left\| X^i_u-\bar{X}^i_u \right\|^2\right]\,ds.
\end{align*}
\olivier{Similarly for $D_t$:}
	\begin{align*}
		\sup_{s\leq t}\Vert D_s \Vert ^2 & \leq T \int_0^t \left\Vert\sum_{\gamma=1}^P \frac 1 N_{\gamma} \sum_{j=1}^{N_{\gamma}} b_{\alpha\gamma}(\bar{X}^i_s,X^j_s)-b_{\alpha\gamma}(\bar{X}^i_s,\bar{X}^j_s) \right\Vert^2 \, ds \\
                \text{(Cauchy-Schwartz)} &\leq PT  \int_0^t \left(\sum_{\gamma=1}^P \frac 1 N_{\gamma}\sum_{j=1}^{N_{\gamma}} \left\| b_{\alpha\gamma}(\bar{X}^i_s,X^j_s)-b_{\alpha\gamma}(\bar{X}^i_s,\bar{X}^j_s) \right\|^2\right)\,ds\\
\text{(assumption~\ref{Assump:LocLipschb})} &\leq PTL^2  \int_0^t \left(\sum_{\gamma=1}^P \frac 1 N_{\gamma}\sum_{j=1}^{N_{\gamma}} \left\|X^j_s-\bar{X}^j_s\right\|^2\right)\,ds
	\end{align*}
Hence we have
	\begin{align*}
\Exp\left[ \sup_{s\leq t}\Vert D_s \Vert ^2\right] & \leq PTL^2 \int_0^t \left(\sum_{\gamma=1}^P \frac 1 N_{\gamma}\sum_{j=1}^{N_{\gamma}} \Exp\left[ \left\|X^j_s-\bar{X}^j_s\right\|^2\right]\right)\,ds\\
&\leq  PTL^2 \int_0^t \left(\sum_{\gamma=1}^P \frac 1 N_{\gamma}\sum_{j=1}^{N_{\gamma}} \Exp\left[ \sup_{u\leq s}  \left\|X^j_u-\bar{X}^j_u\right\|^2\right]\right)\,ds
	\end{align*}
Therefore
\[
\Exp\left[ \sup_{s\leq t}\Vert D_s \Vert ^2\right] \leq P^2 T L^2  \int_0^t  \max_{j=1\cdots N} \Exp\left[ \sup_{u\leq s}  \left\|X^j_u-\bar{X}^j_u\right\|^2\right]\,ds.
\]
	The terms $F_t$ and $G_t$ are treated in the same fashion, but instead of using the Cauchy-Schwartz inequality make use of \olivier{the Burkholder-Davis-Gundy martingale moment inequality. For $F_t$, in detail:}
\begin{align*}
\Exp[\sup_{s\leq t}\Vert F_s \Vert ^2] 
\text{(Cauchy-Schwartz)}& \leq 4P \sum_{\gamma=1}^P \int_0^t \Exp\left[\left\| \frac 1 {N_{\gamma}} \sum_{j=1}^{N_{\gamma}}\beta_{\alpha\gamma}(X^i_s,X^j_s)-\beta_{\alpha\gamma}(\bar{X}^i_s,X^j_s)\,\right\|^2\right]\,ds\\
\text{(Cauchy-Schwartz)}&\leq 4P \sum_{\gamma=1}^P \int_0^t \frac 1 {N_{\gamma}}  \sum_{j=1}^{N_{\gamma}} \Exp\left[\left\| \beta_{\alpha\gamma}(X^i_s,X^j_s)-\beta_{\alpha\gamma}(\bar{X}^i_s,X^j_s)\,\right\|^2\right]\,ds\\
\text{(assumption~\ref{Assump:LocLipschb})} &\leq 4P^2 L^2 \int_0^t \Exp\left[\left\| X^i_s-\bar{X}^i_s\right\|^2\right]\,ds\\
&\leq 4L^2 P^2 \int_0^t  \Exp\left[\sup_{u\leq s}\Vert X^i_u-\bar{X}^i_u \Vert^2\right] \, ds.
\end{align*}
\olivier{Similarly for $G_t$ we obtain:}
	\[
\Exp\left[\sup_{s\leq t}\Vert G_s \Vert ^2\right]\leq 4 L^2 P \int_0^t \max_{j=1\cdots N} \Exp\left[\sup_{0\leq u \leq s}\Vert X^j_u-\bar{X}^j_u \Vert^2\right] \, ds
\]
	We are left \olivier{with the problem of} controlling the terms $E_t$ and \olivier{$H_t$} that involve sums of processes with bounded second moment thanks to proposition~\ref{lem:SoluL2} and assumption~\ref{Assump:bBound}. 
We have:
	\begin{align*}
		\Exp\left[\sup_{s\leq t}\Vert E_s \Vert ^2\right] &=\Exp\left[\sup_{s\leq t} \left\Vert \int_0^s \sum_{\gamma=1}^P \frac 1 {N_{\gamma}} \sum_{j=1}^{N_{\gamma}} b_{\alpha\gamma}(\bar{X}^i_u,\bar{X}^j_u)-\Exp_Z\left[b_{\alpha\gamma}(\bar{X}^i_u,Z_u)\right]\, du\right\Vert^2\right]\\
	\text{(Cauchy-Schwartz)}	& \leq T P \sum_{\gamma=1}^P \int_0^t \Exp\left[\left\Vert \frac 1 {N_{\gamma}} \sum_{j=1}^{N_{\gamma}} b_{\alpha\gamma}(\bar{X}^i_s,\bar{X}^j_s)- \Exp_Z\left[b_{\alpha\gamma}(\bar{X}^i_s,Z_s)\right] \right\Vert^2\right] \, ds\\
	\end{align*}
	and using \olivier{the Burkholder-Davis-Gundy martingale moment inequality}
	\[
		\Exp\left[\sup_{s\leq t}\Vert H_s \Vert ^2\right] \leq 4 P \sum_{\gamma=1}^P \int_0^t \Exp\left[\left\Vert \frac 1 {N_{\gamma}} \sum_{j=1}^{N_{\gamma}} \beta_{\alpha\gamma}(\bar{X}^i_s,\bar{X}^j_s)- \Exp_Z\left[\beta_{\alpha\gamma}(\bar{X}^i_s,Z_s^{\olivier{\gamma}})\right] \right\Vert^2\right] \, ds
\]
		Each of these \olivier{two expressions involves an expectation which we write
	\[
		\Exp\left[\left\Vert \frac 1 N_\gamma \sum_{j=1}^{N_\gamma} \Theta(\bar{X}^i_s,\bar{X}^j_s)-\Exp_Z\left[\Theta(\bar{X}^i_s,Z_s^\gamma)\right]\right\Vert^2\right] ,
	\]
	where $\Theta\in\{b_{\alpha\gamma}, \beta_{\alpha\gamma}\}$, and expand as
\[
\frac 1 {N_\gamma^2} \sum_{j,k=1}^{N_\gamma} \Exp\left[(\Theta(\bar{X}^i_s,\bar{X}^j_s)-\Exp_Z\left[\Theta(\bar{X}^i_s,Z_s^\gamma)\right])^T
		 (\Theta(\bar{X}^i_s,\bar{X}^k_s)-\Exp_Z[\Theta(\bar{X}^i_s,Z_s^\gamma)])\right]
\]
}
\olivier{All the terms of the sum corresponding to indexes $j$ and $k$ such that the three conditions $j\neq i$, $k\neq i$ and $j\neq k$ are satisfied are null, since in that case $\bar{X}^i_t$, $\bar{X}^j_t$, $\bar{X}^k_t$ and $Z_t^\gamma$ are independent and with the same law for $p(j)=p(k)=\gamma$\footnote{Note that $i\neq j$ and $i \neq k$ as soon as $p(i) \neq p(j)=p(k)=\gamma$. In the case where $p(i)=\gamma$, it is easy to check that when $j$ (respectively $k$) is equal to $i$, all terms such that $k \neq j$ (respectively $j \neq k$) are equal to 0.}. In effect, denoting by $m_t^{\gamma}$ the measure of their common law, we have}:
	\begin{multline*}
\Exp\left[(\Theta(\bar{X}^i_s,\bar{X}^j_s)-\Exp_Z[\Theta(\bar{X}^i_s,Z_s^\gamma)])^T  (\Theta(\bar{X}^i_s,\bar{X}^k_s)-\Exp_Z[\Theta(\bar{X}^i_s,Z_s^\gamma)])\right]=\\
\Exp[\Theta(\bar{X}^i_s,\bar{X}^j_s)^T \Theta(\bar{X}^i_s,\bar{X}^k_s)] - \\
\Exp[\Theta(\bar{X}^i_s,\bar{X}^j_s)^T \int \Theta(\bar{X}^i_s,z)\,m_s^\gamma(dz)]-\\
\Exp[\int \Theta(\bar{X}^i_s,z)^T \,m_s^\gamma(dz)  \Theta(\bar{X}^i_s,\bar{X}^k_s) ] + \\
\Exp[\int \Theta(\bar{X}^i_s,z)^T\,m_s^\gamma(dz)\int \Theta(\bar{X}^i_s,z)\,m_s^\gamma(dz)],
\end{multline*}
expanding further \olivier{and renaming the second $z$ variable to $y$ in the last term}  we obtain 
\begin{multline*}
\int\int\int \Theta(x,y)^T\Theta(x,z) \, m_s^\gamma(dx)m_s^\gamma(dy)m_s^\gamma(dz) - \\
\int\int \Theta(x,y)^T \int \Theta(x,z)\,m_s^\gamma(dz)m_s^\gamma(dx)m_s^\gamma(dy)-\\
\int\int \int \Theta(x,z)^T\,m_s^\gamma(dz) \Theta(x,y) \, m_s^\gamma(dx)m_s^\gamma(dy) + \\
\int \int \Theta(x,z)^T m_s^\gamma(dz) \int \Theta(x,y)m_s^\gamma(dy) \, m_s^\gamma(dx)
	\end{multline*}
which is indeed equal to 0 by Fubini theorem. 

Therefore, there are \olivier{no more than $\johnNew{3\,}N_{\gamma}$} non-null terms in the sum, and all the terms have the same value (that depends on $\Theta$), which is bounded \olivier{by lemma~\ref{lem:SoluL2} and assumption~\ref{Assump:bBound}}. We denote by $C\johnNew{/3}$ the supremum of these $2\,P^2$ values for $\Theta\in\{b_{\alpha\gamma}, \beta_{\alpha\gamma}\}$ across all possible pairs of populations, and by \olivier{$N_{min}$ the smallest value of the $N_\gamma,\,\gamma=1\cdots P$}. \olivier{We have shown that}
\[
\Exp\left[\sup_{s\leq t}\Vert E_s \Vert ^2\right] \quad \text{and} \quad \Exp\left[\sup_{s\leq t}\Vert H_s \Vert ^2\right]  \leq \frac{4CTP^2}{N_{min}}
\]	
\olivier{Finally we have:	
\[
\max_{i=1\cdots N}\Exp\left[\sup_{s\leq t} \left\Vert X^i_s-\bar{X}^i_s\right\Vert^2\right] \leq K_1 \int_0^t\max_{j=1\cdots N}\Exp\Big[\sup_{u\leq s}\Vert X^j_u - \bar{X}^j_u\Vert^2\Big]\, du+\frac{K_2}{N_{min}},
\]
for some positive constants $K_1$ and $K_2$.}
	Using Gronwall's inequality, we obtain: \olivier{
	\begin{equation}\label{eq:Gron1}
		\max_{i=1\cdots N}\Exp\left[\sup_{s\leq t} \left\Vert X^i_s-\bar{X}^i_s\right\Vert^2\right] \leq \frac{K_3}{N_{min}}
	\end{equation}
for some positive constant $K_3$, \johnNew{term that tends to zero as $N$ goes to infinity proving the propagation of chaos property. In order to show a convergence with speed $1/\sqrt{N}$ as stated in the theorem, we use the fact}}
\olivier{	
	\[\max_{i=1\cdots N} N \; \Exp\left[\sup_{s\leq T} \left\Vert X^{i,N}_s - \bar{X}^i_s\right\Vert^2 \right]\leq K_3 \frac{N}{N_{min}},
\]
and the righthand side of the inequality is bounded for all $N$s because of the hypothesis $\lim_{N \to \infty} \frac{N_\alpha}{N}=c_\alpha \in (0,1)$ for $\alpha=1\cdots P$.}
	
This ends the proof. 
	
\end{proof}

\bibliographystyle{plain}

\end{document}